\documentclass[nolineno, a4paper,UKenglish,cleveref, autoref, thm-restate]{socg-lipics-v2021}
\title{T-REX: Fast and Dynamic Journey Planning for Continental-Scale Public Transit Networks}

\titlerunning{T-REX: Journey Planning for Continental-Scale Public Transit Networks}

\author{Jonas Sauer}{Karlsruhe Institute of Technology, Germany}{jonas.sauer@kit.edu}{https://orcid.org/0000-0002-7196-7468}{}
\author{Patrick Steil}{Karlsruhe Institute of Technology, Germany}{patrick@steil.dev}{https://orcid.org/0000-0003-3282-4533}{}
\author{Sascha Witt}{Karlsruhe Institute of Technology, Germany}{sascha.witt@kit.edu}{https://orcid.org/0000-0002-7867-3200}{}

\authorrunning{J.~Sauer, P.~Steil and S.~Witt} %

\Copyright{Jonas Sauer, Patrick Steil and Sascha Witt} %



\supplement{Code: \url{https://github.com/PatrickSteil/TREX}; Datasets:~\cite{sauer_2026_19697732}}

\acknowledgements{We thank Christian Schulz for providing us with the original implementation of KaHIP, Lars Gottesbüren and Nikolai Maas for advice on configuring Mt-KaHyPar, Ben Strasser for providing us with the original implementations of CSA and ACSA, and Patrick Brosi and Transitous~\cite{transitousTransitous} for supplying the Germany and Europe datasets used in our experiments. This work is partially based on the master's thesis by Steil~\cite{Ste25}.
}

\hideLIPIcs
\usepackage[all]{nowidow}

\usepackage{pgfplots}
\usetikzlibrary{pgfplots.groupplots}
\usepgfplotslibrary{statistics}
\usetikzlibrary{patterns}
\pgfplotsset{compat=1.18}

\usepackage{amsthm}
\usepackage{amssymb}
\usepackage{xspace}
\usepackage{cancel}

\newcommand{\ie}{i.e.,\xspace~}
\newcommand{\eg}{e.g.,\xspace~}
\newcommand{\cf}{cf.\xspace~}
\newcommand{\etal}{et~al.\xspace~}

\newcommand{\printTime}[3]{%
    \num[minimum-integer-digits=2]{#1}:\num[minimum-integer-digits=2]{#2}:\num[minimum-integer-digits=2]{#3}%
}

\newcommand{\printMinuteSeconds}[2]{%
    \num[minimum-integer-digits=2]{#1}:\num[minimum-integer-digits=2]{#2}%
}

\usepackage{numprint}
\npdecimalsign{.} %
\usepackage{makeidx}
\usepackage{float}
\usepackage{placeins}
\usepackage{color}
\usepackage{hyperref}
\usepackage{xspace}
\usepackage{booktabs}
\usepackage{mathtools}
\usepackage[nice]{nicefrac}
\usepackage{listings}
\usepackage{caption}
\usepackage{subcaption}
\usepackage{makecell} %
\usepackage{upgreek} %
\usepackage{siunitx}
\usepackage{multirow}
\usepackage{graphicx}
\usepackage{niceAlgorithm}

\usepackage{colortbl}
\usepackage{xcolor}
\newcommand{\highlightCell}{\cellcolor{gray!15}}

\usepackage[scaled]{inconsolata}

\newcommand{\cpp}{C\raisebox{0.15ex}{\small++}\xspace}

\newcommand{\absoluteVal}[1]{\left\vert #1 \right\vert}

\newcommand{\atime}{\ensuremath{\tau}\xspace}
\newcommand{\departureTime}{\atime_\mathrm{dep}}
\newcommand{\arrivalTime}{\atime_\mathrm{arr}}
\newcommand{\depTime}[1]{\atime_\mathrm{dep} \left( #1 \right)}
\newcommand{\arrTime}[1]{\atime_\mathrm{arr} \left( #1 \right)}

\newcommand{\graph}{G}
\newcommand{\vertices}{V}
\newcommand{\edges}{E}

\newcommand{\vertexA}{u}
\newcommand{\vertexB}{v}

\newcommand{\stops}{\mathcal{S}}
\newcommand{\aStop}{p}
\newcommand{\stopA}{p}
\newcommand{\stopB}{q}

\newcommand{\sourceStop}{\aStop_s}
\newcommand{\targetStop}{\aStop_t}
\newcommand{\stopEvents}{\mathcal{E}}
\newcommand{\stopEvent}{\varepsilon}
\newcommand{\sourceStopEvent}{\stopEvent_{\mathrm{s}}}
\newcommand{\targetStopEvent}{\stopEvent_{\mathrm{t}}}

\newcommand{\lines}{\mathcal{L}}
\newcommand{\aLine}{L}
\newcommand{\trips}{\mathcal{T}}
\newcommand{\aTrip}{T}
\newcommand{\tripA}{T_a}
\newcommand{\tripB}{T_b}
\newcommand{\tripC}{T_c}

\newcommand{\tripSegment}[3]{#1[#2,#3]}
\newcommand{\footpaths}{\mathcal{F}}

\newcommand{\footpath}{f}
\newcommand{\transfertime}[2]{\atime_{\mathrm{f}}( #1, #2 )}
\newcommand{\query}{Q}

\newcommand{\crit}{c}

\newcommand{\paretoRep}{\mathfrak{J}}

\newcommand{\journeys}{\mathcal{J}}
\newcommand{\aJourney}{J}
\newcommand{\anEventSubJourney}{\aJourney_{\stopEvents}}
\newcommand{\eventSubjourney}[2]{\aJourney[#1,#2]}

\newcommand{\subjourney}[2]{\aJourney[#1,#2]}
\newcommand{\subjourneyAlt}[2]{\aJourney'[#1,#2]}
\newcommand{\transfers}{\mathfrak{T}}
\newcommand{\aTransfer}{\mathfrak{t}}
\newcommand{\transfer}[2]{(#1,#2)}

\newcommand{\layoutGraph}{\graph_L}
\newcommand{\layoutEdges}{\edges_L}
\newcommand{\layoutVertices}{\vertices_L}

\newcommand{\labels}{\ensuremath{\mathcal{L}}\xspace}
\newcommand{\aQueue}{\ensuremath{Q}\xspace}
\newcommand{\reachedIndex}{R}

\newcommand{\minTime}{\ensuremath{\atime_{\textsf{min}}}\xspace}
\SetKwFunction{Enqueue}{Enqueue}
\SetKwFunction{Scan}{Scan}
\SetKwFunction{Split}{Split}
\SetKwFunction{InSameCell}{InSameCell}
\SetKwFunction{Smaller}{operator$<$}
\SetKwFunction{Customize}{Customize}
\SetKwFunction{ProfileQuery}{ProfileQuery}
\SetKwFunction{Reduce}{Reduce}

\SetKwFunction{PrunedBFS}{PrunedBFS}
\SetKwFunction{Preprocess}{Preprocess}
\SetKwFunction{TopoSweep}{TopoSweep}
\SetKwFunction{prune}{prune}
\SetKwFunction{continue}{continue}
\SetKwFunction{seen}{seen}

\newcommand{\tiebreakingSequence}{X}
\newcommand{\tiebreakingSequencePartial}{\tilde{X}}

\newcommand{\lineID}{\ensuremath{\text{ID}_\lines}\xspace}

\newcommand{\transferID}{\ensuremath{\text{ID}_\transfers}\xspace}
\newcommand{\footpathID}{\ensuremath{\text{ID}_\footpaths}\xspace}

\newcommand{\allCells}{\mathcal{C}}
\newcommand{\bipartition}{\mathcal{C}}
\newcommand{\cell}{C}
\newcommand{\numLevels}{K}
\newcommand{\aLevel}{\ell}
\newcommand{\cellId}[1]{c_{\mathrm{ID}}\left( #1 \right)}
\newcommand{\cellIdL}[2]{c_{\mathrm{ID}}\left( #1, #2 \right)}
\newcommand{\rank}[1]{r\left( #1 \right)}
\newcommand{\lcl}[2]{\mathrm{LCL}\left( #1, #2 \right)}
\newcommand{\successor}[2]{\mathrm{succ}(#1,#2)}

\SetKwFunction{Query}{Query}
\SetKwFunction{FirstReachConn}{FirstReachConn}

\newcommand{\shift}{\gg}

\usepackage{xcolor}
\definecolor{KITgreen}          {rgb}{0,    0.588,0.509}
\definecolor{KITgreen70}        {rgb}{0.3,  0.711,0.656}
\definecolor{KITgreen50}        {rgb}{0.5,  0.794,0.754}
\definecolor{KITgreen30}        {rgb}{0.7,  0.876,0.852}
\definecolor{KITgreen15}        {rgb}{0.85, 0.938,0.926}
\definecolor{KITblue}           {rgb}{0.274,0.392,0.666}
\definecolor{KITblue70}         {rgb}{0.492,0.574,0.766}
\definecolor{KITblue50}         {rgb}{0.637,0.696,0.833}
\definecolor{KITblue30}         {rgb}{0.782,0.817,0.9}
\definecolor{KITblue15}         {rgb}{0.891,0.908,0.95}
\definecolor{KITpalegreen}      {rgb}{0.509,0.745,0.235}
\definecolor{KITpalegreen70}    {rgb}{0.656,0.821,0.464}
\definecolor{KITpalegreen50}    {rgb}{0.754,0.872,0.617}
\definecolor{KITpalegreen30}    {rgb}{0.852,0.923,0.77}
\definecolor{KITpalegreen15}    {rgb}{0.926,0.961,0.885}
\definecolor{KITyellow}         {rgb}{0.98, 0.901,0.078}
\definecolor{KITyellow70}       {rgb}{0.986,0.931,0.354}
\definecolor{KITyellow50}       {rgb}{0.99, 0.95, 0.539}
\definecolor{KITyellow30}       {rgb}{0.994,0.97, 0.723}
\definecolor{KITyellow15}       {rgb}{0.997,0.985,0.861}
\definecolor{KITorange}         {rgb}{0.862,0.627,0.117}
\definecolor{KITorange70}       {rgb}{0.903,0.739,0.382}
\definecolor{KITorange50}       {rgb}{0.931,0.813,0.558}
\definecolor{KITorange30}       {rgb}{0.958,0.888,0.735}
\definecolor{KITorange15}       {rgb}{0.979,0.944,0.867}
\definecolor{KITbrown}          {rgb}{0.627,0.509,0.196}
\definecolor{KITbrown70}        {rgb}{0.739,0.656,0.437}
\definecolor{KITbrown50}        {rgb}{0.813,0.754,0.598}
\definecolor{KITbrown30}        {rgb}{0.888,0.852,0.758}
\definecolor{KITbrown15}        {rgb}{0.944,0.926,0.879}
\definecolor{KITred}            {rgb}{0.627,0.117,0.156}
\definecolor{KITred70}          {rgb}{0.739,0.382,0.409}
\definecolor{KITred50}          {rgb}{0.813,0.558,0.578}
\definecolor{KITred30}          {rgb}{0.888,0.735,0.747}
\definecolor{KITred15}          {rgb}{0.944,0.867,0.873}
\definecolor{KITlilac}          {rgb}{0.627,0,    0.47}
\definecolor{KITlilac70}        {rgb}{0.739,0.3,  0.629}
\definecolor{KITlilac50}        {rgb}{0.813,0.5,  0.735}
\definecolor{KITlilac30}        {rgb}{0.888,0.7,  0.841}
\definecolor{KITlilac15}        {rgb}{0.944,0.85, 0.92}
\definecolor{KITcyanblue}       {rgb}{0.313,0.666,0.901}
\definecolor{KITcyanblue70}     {rgb}{0.519,0.766,0.931}
\definecolor{KITcyanblue50}     {rgb}{0.656,0.833,0.95}
\definecolor{KITcyanblue30}     {rgb}{0.794,0.9,  0.97}
\definecolor{KITcyanblue15}     {rgb}{0.897,0.95, 0.985}
\definecolor{KITseablue}        {rgb}{0.196,0.313,0.549}
\definecolor{KITseablue70}      {rgb}{0.437,0.519,0.684}
\definecolor{KITseablue50}      {rgb}{0.598,0.656,0.774}
\definecolor{KITseablue30}      {rgb}{0.758,0.794,0.864}
\definecolor{KITseablue15}      {rgb}{0.879,0.897,0.932}
\definecolor{KITblack}          {rgb}{0,    0,    0}
\definecolor{KITblack90}        {rgb}{0.1,  0.1,  0.1}
\definecolor{KITblack80}        {rgb}{0.2,  0.2,  0.3}
\definecolor{KITblack75}        {rgb}{0.25, 0.25, 0.25}
\definecolor{KITblack70}        {rgb}{0.3,  0.3,  0.3}
\definecolor{KITblack60}        {rgb}{0.4,  0.4,  0.4}
\definecolor{KITblack50}        {rgb}{0.5,  0.5,  0.5}
\definecolor{KITblack40}        {rgb}{0.6,  0.6,  0.6}
\definecolor{KITblack30}        {rgb}{0.7,  0.7,  0.7}
\definecolor{KITblack25}        {rgb}{0.75, 0.75, 0.75}
\definecolor{KITblack20}        {rgb}{0.8,  0.8,  0.8}
\definecolor{KITblack10}        {rgb}{0.9,  0.9,  0.9}
\definecolor{KITwhite}          {rgb}{1,    1,    1}

\usepackage{tikz}
\usepackage{pgfplots}
\usepackage{colortbl}
\usepgfplotslibrary{groupplots}
\usetikzlibrary{positioning,shapes,arrows,fit,backgrounds}
\pgfplotsset{compat=newest}
\pgfkeys{/pgf/number format/.cd, set thousands separator={\,}}

\colorlet{nodeColor}{black!80}
\colorlet{edgeColor}{black!50}
\colorlet{axisColor}{black!80}
\colorlet{legendColor}{black!80}

\newcommand{\gs}{\hphantom{\tiny$\cdot$}}

\tikzstyle{vertex}=[circle,line width=.5pt,minimum size=0.1pt]
\tikzstyle{arrow}=[->, >=stealth]
\tikzstyle{route}=[arrow, line width=2.5pt, rounded corners = 20]
\tikzstyle{edge}=[edgeColor, line width=1pt, rounded corners = 20]
\tikzstyle{directedEdge}=[edge, arrow]
\tikzstyle{stop}=[rectangle, line width=.5pt, inner sep=5pt, draw=nodeColor!100,fill=nodeColor!15,rounded corners=0.1cm]

\pgfplotsset{
    axis line style={axisColor,line width = 0.4pt},
    major tick style={axisColor,line width = 0.4pt},
    minor tick style={axisColor,line width = 0.2pt},
    major tick length=3.5pt,
    minor tick length=2.0pt,
    ytick align=outside,
    xtick align=outside,
    xtick pos=left,
    ytick pos=left,
    ticklabel style = {font=\small},
    label style = {font=\small},
}

\EventEditors{John Q. Open and Joan R. Access}
\EventNoEds{2}
\EventLongTitle{42nd Conference on Very Important Topics (CVIT 2016)}
\EventShortTitle{CVIT 2016}
\EventAcronym{CVIT}
\EventYear{2016}
\EventDate{December 24--27, 2016}
\EventLocation{Little Whinging, United Kingdom}
\EventLogo{}
\SeriesVolume{42}
\ArticleNo{23}

\begin{document}

\maketitle

\begin{abstract}
We present T-REX (Transfer-Ranked EXploration), a new algorithm for journey planning in public transit networks on the country and continental scale.
Our algorithm applies the principles of multi-level overlays to Trip-Based Public Transit Routing~(TB).
Using a multi-level partition of the network, T-REX identifies transfers between trips that are relevant for long-distance travel in a short precomputation phase.
This information is then used to prune irrelevant local transfers during a query.
Like other state-of-the-art algorithms, T-REX Pareto-optimizes arrival time and the number of used trips.
T-REX dramatically outperforms previous overlay-based algorithms for three key reasons: (1) a better partition, (2) reducing the search space by focusing on transfers rather than trips, and (3) a redesigned query algorithm with improved memory efficiency and throughput.
As a result, T-REX answers queries in less than~$\numprint{10}\,\si{\milli\second}$ on a network of Europe, including local and long-distance transit.
This constitutes a speedup of~$\numprint{20}$ compared to TB and~$\numprint{80}$ compared to algorithms without preprocessing.
The memory footprint is moderate and the precomputation takes only two minutes, while real-time schedule updates can be incorporated in a few seconds.
These properties make T-REX the first public transit journey planning algorithm that fulfills the requirements of interactive real-time applications on the continental scale.
\end{abstract}

\section{Introduction}
	\label{ch:intro}

    We encounter journey planning in public transport in our daily lives, whether on the morning commute to work or on vacation across countries.
    Interactive server-based applications, such as Google Maps, receive thousands of requests per second.
    Handling them in a timely fashion requires algorithms that produce accurate results within milliseconds.
    Whereas on road networks it is often sufficient to minimize travel time as a single objective, travelers using public transport are typically also interested in comfort criteria, most importantly the number of different vehicles (trips) used.
    Therefore, it is common practice to compute Pareto-optimal journeys in terms of arrival time and number of trips, minimizing both criteria~\cite{Bas16b}.

    Because Dijkstra's algorithm~\cite{Dij59} is too slow for interactive applications, \emph{speedup techniques} precompute auxiliary information that is used to accelerate queries.
    Their performance is evaluated according to several metrics, including query time, memory consumption, precomputation time and the ability to handle real-time information (\eg traffic or delays).
    For road networks, several techniques perform well in all four metrics, even on the continental scale~\cite{Del17,Dib16,Bla25}.
    On public transit networks, the situation is less satisfactory.
    Here, classic speedup techniques appear to be much less effective~\cite{Bau07,Ber09,Bas09}.
    On the other hand, progress has been made by designing new query algorithms that are tailored to the special structure of public transit networks and modern CPU architectures~\cite{Dib18,Del15b,Wit15}.
    On metropolitan networks, these algorithms are fast enough for interactive applications and support real-time updates, but their query speed is not sufficient for the continental scale.
    Moreover, integrating them with speedup techniques is challenging.
    Previous attempts have either caused the already low speedups to degrade further~\cite{Del17b,Dib18,Aga22} or led to a blowup in the preprocessing effort, such that processing real-time updates becomes infeasible~\cite{Bas10,Bas16,Wit16,Gro25}.
    
	\subparagraph*{Contribution.}
    We present T-REX (Transfer-Ranked EXploration), a public transit journey planning algorithm that achieves interactive query times even on continent-sized networks, with moderate preprocessing time and space.
    T-REX extends Trip-Based Public Transit Routing~(TB)~\cite{Wit15}, a cache-efficient modern query algorithm, by adapting the \emph{multi-level overlays}~(MLO) speedup technique~\cite{Hol08,Del17}.
    It uses a hierarchical nested partition of the network to determine which transfers between trips are required for long-distance travel. 
    During a query, unimportant transfers are ignored.
    Three key improvements allow T-REX to achieve much faster query times than previous partitioning-based algorithms~\cite{Sch02,Del17b,Dib18,Aga22}: (1) It obtains a better multi-level partition from a state-of-the-art graph partitioner. (2) The auxiliary information is computed for transfers rather than trips, which is more effective. (3) The query algorithm is carefully redesigned to exploit the pruning information in a cache-efficient manner.\looseness=-1

    We benchmark T-REX against state-of-the-art algorithms on networks ranging from metropolitan areas to a network covering most of Europe.
    With over $\numprint{1,3}$ million stops and over $\numprint{100}$ million events, our Europe instance is, to our knowledge, the largest network used in a scientific evaluation of journey planning algorithms so far.
    On this instance, T-REX answers queries in less than $\numprint{10}\,\si{\milli\second}$, requiring two minutes of precomputation time and $\numprint{28}\,\si{\giga\byte}$ of data (including $\numprint{1.8}\,\si{\giga\byte}$ for the original timetable).
    This is a speedup of~$\numprint{20}$ over TB and~$\numprint{80}$ over RAPTOR~\cite{Del15b}, currently the fastest Pareto-optimal algorithm without substantial preprocessing.
    Moreover, T-REX can incorporate real-time updates quickly without re-running the entire preprocessing.
    Hence, it is the first public transit journey planning algorithm that satisfies all the criteria for an interactive application on continental-scale networks.\looseness=-1
    
	\subparagraph*{Outline.}
    \Cref{ch:relatedwork} reviews related work in journey planning.
    In~\Cref{ch:fundamentals}, we introduce notation and outline the TB algorithm.
    We then present T-REX in~\Cref{ch:trex} and conduct an experimental evaluation in \Cref{ch:experiments}.
    Finally, Section~\ref{ch:conclusion} summarizes our findings.

	\section{Related Work}
	\label{ch:relatedwork}	
	We mainly consider algorithms that Pareto-optimize arrival time and the number of used trips.
    For a survey of routing techniques published prior to 2016, we refer the reader to~\cite{Bas16}.

    \subparagraph*{Basic Algorithms.}
    Traditional approaches model the timetable as a graph and apply variants of Dijkstra's algorithm~\cite{Dij59,Bro04,Mue07b,Pyr08}.
    More recent algorithms achieve faster query times via cache-friendly scan operations.
    Connection Scan Algorithm~(CSA)~\cite{Dib18} optimizes the arrival time with a single sweep over all elementary connections, but does not support Pareto optimization.
	RAPTOR~\cite{Del15b} is a modified breadth-first search that scans each discovered line in a single sweep.
    Trip-Based Public Transit Routing~(TB)~\cite{Wit15} speeds up this approach by precalculating a set of relevant transfers between trips.
    This allows the search to bypass the step of finding the earliest reachable trip of a line.

    \subparagraph*{Speedup Techniques.}
    Classic speedup techniques for road networks include goal-directed approaches, such as Arc-Flags~\cite{Hil09,Lau09,Moe06} and $\mathrm{A}^{\ast}$ search~\cite{Har68,Gol05}, and hierarchical techniques.
    Among the latter, \emph{multi-level overlays}~(MLO)~\cite{Hol08} recursively select a set of ``important'' vertices and add distance-preserving \emph{shortcuts} between them to create a hierarchy of \emph{overlays}.
    Intuitively, when the search is ``far away'' from the source and the target, it moves upwards in the hierarchy and bypasses ``unimportant'' vertices.
    MLOs can be built via vertex contraction~\cite{Hol08} or a multi-level partition~\cite{Jun02,Del17}.
    In the latter case, the ``important'' vertices are the boundary vertices of the cells in the partition.
    An evolution of MLO are Contraction Hierarchies~(CH)~\cite{Gei12}, which rely on a complete ranking of the vertices by ``importance''.
    
    Because road networks have small and balanced separators~\cite{Sch15, Got19, Ham18}, multilevel partitions can be found without considering the edge weights.
    This enables \emph{customizable} techniques, in which the preprocessing is mostly metric-independent and the edge weights are incorporated in a second, much faster step.
    Thus, metric updates, such as real-time traffic, can be processed within seconds.
    Customizable Route Planning~(CRP)~\cite{Del17} and Customizable Contraction Hierarchies~(CCH)~\cite{Dib16,Bla25} extend MLO and CH, respectively, to the customizable setting and speed up Dijkstra's algorithm by up to four orders of magnitude.

    \subparagraph*{Applications.}
    Traditionally, speedup techniques have seen limited success in public transit~\cite{Bau07,Ber09,Bas09}, with speedups often in the low single digits~\cite{Dis08,Ber09,Del09c}.
    A recent breakthrough is FLASH-TB~\cite{Gro25}, which adapts Arc-Flags to TB.
    The network is partitioned into~$k$ cells, and each transfer is labeled with a $k$-bit vector, where bit~$i$ indicates whether the transfer is required to reach cell~$i$. 
    The query only relaxes transfers that are flagged for the target cell.
    Previous adaptations of Arc-Flags~\cite{Ber09,Del09c} suffered from the presence of many equivalent journeys and the inability to break ties between them in a consistent manner.
    To resolve this issue, FLASH-TB is forced to use a preprocessing step with a quadratic runtime in the network size, which is infeasible for continental-sized networks.
    However, it achieves speedups of up to~$\numprint{500}$ over TB, demonstrating that classic speedup techniques can be effective.

    A challenge in applying hierarchical techniques to public transit networks is that they exhibit a weaker hierarchy than road networks~\cite{Bas09,Dib18}, particularly at the local transit level, which makes it harder to find small separators.
    Node contraction is also less successful due to the presence of high-degree nodes (\eg large train stations)~\cite{Ber09}.
    An early application of MLO to graph-based models that optimized only arrival time achieved a speedup of~$10$~\cite{Sch02}.
    Connection Scan Accelerated~(ACSA)\cite{Dib18}, which applies MLO to CSA, also only optimizes the arrival time.
    For each cell in a multi-level partition, ACSA identifies connections that are required to traverse it.
    During a query, it collects, sorts and scans only the relevant connections.
    On country-scale networks, this achieves a speedup of~$7$ over CSA.
    HypRAPTOR~\cite{Del17b} and HypTB~\cite{Aga22} apply overlay-based ideas to RAPTOR and TB, respectively, using a hypergraph representation of the network with lines as vertices and stops as hyperedges.
    Sadly, these algorithms achieve a speedup of at most two over their respective baselines.
    
	\subparagraph*{Expensive Preprocessing.}
    A family of algorithms with extremely fast query times answers every possible query in advance and compresses the obtained information.
    This requires quadratic precomputation effort, which is infeasible for continental-sized networks.
    The first such algorithm is Transfer Patterns (TP)~\cite{Bas10}, which builds search graphs called transfer patterns.
    These require enormous space, so Scalable Transfer Patterns~\cite{Bas16} split them into ``local'' and ``global'' patterns.
    The query algorithm must re-assemble these, making it much slower than TP and only slightly faster than TB.
    Heuristic variants of (Scalable) TP reduce the precomputation time at the expense of provable optimality.
    Newer representatives include TB using Condensed Search Trees~\cite{Wit16} and FLASH-TB.
    Unlike the other algorithms, the memory consumption of FLASH-TB is configurable, and query times remain extremely fast even with moderate space.
    A different technique with sub-millisecond query times is Public Transit Labeling~(PTL)~\cite{Del15}, which adapts Hub Labeling~\cite{Coh03}.
    However, for Pareto-optimal queries, the required space is already in the tens of gigabytes on metropolitan networks.
    
    \subparagraph*{Dynamic Settings.}
    Updates such as delays and trip cancellations are easy to incorporate into algorithms without heavy preprocessing, including graph-based models~\cite{Mue09,Cio17,Gia19}, RAPTOR and CSA.
    For TB, the precomputed transfer set can be updated quickly without recomputing the entire set~\cite{Wit21}.
    Algorithms with heavy preprocessing generally do not support real-time updates.
	Experimentally, TP has been shown to be fairly robust with respect to delays~\cite{Bas13b}, but optimality cannot be guaranteed.
    Moreover, this robustness may not extend to blocked tracks or closed stations, which are likely to require substantial changes to the stored patterns.
    For PTL, a single delayed connection can require minutes to process~\cite{Emi20}.
    
	\section{Fundamentals}
	\label{ch:fundamentals}

    This section introduces notation and definitions and gives an overview of the TB algorithm.
    
    \subparagraph*{Network.}
    A \emph{public transit network} is a 5-tuple $(\stops,\footpaths,\stopEvents,\trips,\lines)$ consisting of a set of stops $\stops$, footpaths $\footpaths \subseteq \stops \times \stops$, events $\stopEvents$, trips $\trips$ and lines $\lines$.	
	A \emph{stop} $\aStop \in \stops$ is a location at which passengers can enter and exit vehicles.
    A \emph{trip} is a sequence~$\aTrip = \left< \stopEvent_{1}, \dots, \stopEvent_{k} \right>$ of \emph{stop events} performed by a single vehicle, with~$\absoluteVal{\aTrip}=k$.
    The~$i$-th stop event of~$\aTrip$ is denoted by~$\aTrip[i]$, its visited stop by~$\aStop(\aTrip[i])$, its arrival time by~$\arrivalTime(\aTrip[i])$, and its departure time by~$\departureTime(\aTrip[i])$, measured in a resolution of seconds.
    For~$1 \leq i < j \leq \absoluteVal{\aTrip}$, we denote by~$\aTrip[i,j]$ the trip segment from~$\aTrip[i]$ to~$\aTrip[j]$.
    If~$j=i+1$, we call the segment a \emph{connection}.
    Trips are grouped into \emph{lines} $\lines$ such that all trips of the same line visit the same stop sequence and admit a total ordering~$\prec$, where~$\aTrip \prec \aTrip'$ if $\arrivalTime(\aTrip[i]) < \arrivalTime(\aTrip'[i])$ and $\departureTime(\aTrip[i]) < \departureTime(\aTrip'[i])$ for every $1 \leq i \leq \absoluteVal{\aTrip}$.
    We write~$\aTrip \preceq \aTrip'$ as shorthand for~$\aTrip \prec \aTrip' \lor \aTrip = \aTrip'$.
	A \emph{footpath} $(\stopA,\stopB)\in\footpaths$ connects two stops $\stopA, \stopB \in \stops$ and can be traversed in time~$\transfertime{\stopA}{\stopB}$.
    We assume that the footpath~$(\stopA,\stopA)$ exists for every stop~$\stopA$ with~$\transfertime{\stopA}{\stopA}=0$, and that~$\transfertime{\stopA}{\stopB}=\infty$ if~$(\stopA,\stopB)\notin\footpaths$.
    Furthermore, we require that~$\footpaths$ is transitively closed and fulfills the triangle inequality.

    For simplicity of exposition, we do not explicitly model a minimum change time between trips at the same stop.
    Following~\cite{Sau24}, a \emph{buffer time} can instead be incorporated implicitly by subtracting it from the departure time of the respective stop events.

    \subparagraph*{Journeys.}
    A \emph{transfer} is a tuple $\aTransfer = \transfer{\tripA[j]}{\tripB[i]} \in \stopEvents \times \stopEvents$ such that~$\footpaths$ contains the footpath $\left(\aStop\left(\tripA[j]\right), \aStop\left(\tripB[i]\right)\right)$ and the traversal time is short enough to reach $\tripB[i]$ from $\tripA[j]$, \ie $\arrTime{\tripA[j]} + \transfertime{\aStop\left(\tripA[j]\right)}{\aStop\left(\tripB[i]\right)} \leq \depTime{\tripB[i]}$.
    Given a source event~$\sourceStopEvent\in\stopEvents$ and a target event~$\targetStopEvent\in\stopEvents$, an $\sourceStopEvent$-$\targetStopEvent$-\emph{journey} is a sequence~$\aJourney = \left<\tripSegment{\aTrip_1}{i_1}{j_1}, \dots,\tripSegment{\aTrip_k}{i_k}{j_k} \right>$ of trip segments that describes how a passenger travels from~$\sourceStopEvent=\aTrip_1[i_1]$ to~$\targetStopEvent=\aTrip_k[j_k]$, such that the transfer~$\transfer{\aTrip_n[j_n]}{\aTrip_{n+1}[i_{n+1}]}$ exists for~$1 \leq n < k$.
    Given a source stop~$\sourceStop$ and a target stop~$\targetStop$, a~$\sourceStop$-$\targetStop$-journey is a sequence~$\aJourney = \left< \footpath_{0}, \tripSegment{\aTrip_1}{i_1}{j_1}, \dots, \tripSegment{\aTrip_k}{i_k}{j_k}, \footpath_{k+1} \right>$ with initial footpath~$\footpath_{0} = \left(\sourceStop, \aStop\left(\aTrip_1[i_1]\right)\right) \in \footpaths$, final footpath~$\footpath_{k+1} = \left(\aStop\left(\aTrip_k[j_k]\right), \targetStop\right) \in \footpaths$ and transfers as above.
    We omit a footpath from the notation if it has the form~$(\aStop,\aStop)$, \ie connects a stop to itself.
	We denote the number of trip segments in~$\aJourney$ by~$\absoluteVal{\aJourney}$.
    The departure time of~$\aJourney$ is~$\departureTime(\aJourney)\coloneqq\departureTime(\aTrip_1[i_1])-\transfertime{\sourceStop}{\aStop(\aTrip_1[i_1])}$ and the arrival time is~$\arrivalTime(\aJourney)\coloneqq\arrivalTime(\aTrip_k[j_k])+\transfertime{\aStop(\aTrip_k[j_k])}{\targetStop}$.
    
    Given a journey $\aJourney=\left<\tripSegment{\aTrip_1}{i_1}{j_1}, \dots,\tripSegment{\aTrip_k}{i_k}{j_k} \right>$, a \emph{proper subjourney} is a subsequence of~$\aJourney$, \ie it includes trip segments in their entirety, whereas an \emph{event subjourney} may start and end in the middle of a trip segment.
    For indices~$1 \leq m < n \leq k$, we denote the proper subjourney from the~$m$-th to the~$n$-th trip segment by~$\subjourney{m}{n}=\left<\tripSegment{\aTrip_m}{i_m}{j_m}, \dots,\tripSegment{\aTrip_n}{i_n}{j_n} \right>$.
    For indices~$i_m \leq \ell_m < j_m$ and~$i_n < \ell_n \leq j_n$, we denote the event subjourney from~$\aTrip_m[\ell_m]$ to~$\aTrip_n[\ell_n]$ by~$\eventSubjourney{\aTrip_m[\ell_m]}{\aTrip_n[\ell_n]}=\left<\tripSegment{\aTrip_m}{\ell_m}{j_m}, \dots,\tripSegment{\aTrip_n}{i_n}{\ell_n} \right>$.
    A (proper or event) subjourney of~$\aJourney$ is a \emph{prefix} if it starts with~$\aTrip_1[i_1]$ and a \emph{suffix} if it ends with~$\aTrip_k[j_k]$.

    \subparagraph*{Queries.}
    A \emph{query} $\query=(\sourceStop,\targetStop,\departureTime)$ consists of source and target stops~$\sourceStop,\targetStop$ and a departure time~$\departureTime$.
    A~$\sourceStop$-$\targetStop$-journey $\aJourney$ is \emph{feasible} for $\query$ if $\departureTime(\aJourney) \geq \departureTime$.
    A journey~$\aJourney$ \emph{dominates} another journey~$\aJourney'$ if its \emph{cost vector}~$\crit(\aJourney) = \left(\arrivalTime(\aJourney),\absoluteVal{\aJourney}\right)$ is strictly smaller than~$\crit(\aJourney')$ in at least one criterion and not larger in the other.
    A feasible journey~$\aJourney$ and its cost vector are \emph{Pareto-optimal} if no other feasible journey dominates $\aJourney$.
    The \emph{Pareto front} is the set of all Pareto-optimal cost vectors.
    A \emph{representative set} consists of one journey for each cost vector in the Pareto front.
    The objective of~$\query$ is to report the Pareto front and a representative set.
			
	\subparagraph*{Partitioning.}
    For~$k \in \mathbb{N}^{+}$, a \emph{$k$-partition} of a set~$O$ is a family $\allCells(O) = \left\{\cell_{0}, \cell_{1}, \dots, \cell_{k-1}\right\}$ of sets such that $\bigcup_{0 \leq i < k}\cell_{i} = O$ and $\cell_{i} \cap \cell_{j} = \emptyset$ holds for all $0 \leq i,j < k$. 
    Each set $\cell \in \allCells(O)$ is called a \textit{cell}.
	If~$k = 2$, we refer to~$\allCells(O)$ as a \textit{bipartition} of $O$.

    \begin{definition}\label{def:nestedbipartition} 
        A \textbf{nested bipartition} $\bipartition_{\numLevels}(\graph)$ of a $\graph = \left(\vertices, \edges\right)$ into $\numLevels \in \mathbb{N}^+$ levels is a family $\left\{\allCells^\aLevel(V)\right\}_{\aLevel=0}^{\numLevels}$,
        where each $\allCells^\aLevel(V)$ is a $2^{\numLevels - \aLevel}$-partition of $V$, \ie
        $\allCells^\aLevel(V) = \{ \cell_0^\aLevel, \cell_1^\aLevel, \dots, \cell_{2^{\numLevels - \aLevel}-1}^\aLevel \}$,
        and for each $1 \leq \aLevel \leq \numLevels$ and $0 \leq i < 2^{\numLevels - \aLevel}$, the set $\left\{ \cell_{2i}^{\aLevel-1}, \cell_{2i+1}^{\aLevel-1} \right\}$ is a bipartition of~$\cell_i^\aLevel$.
	\end{definition}

	For $\vertexA, \vertexB \in \vertices$, the \textit{lowest common level} $\lcl{\vertexA}{\vertexB}$ is the lowest level~$\aLevel$ on which~$\vertexA$ and~$\vertexB$ are in the same cell.
    Note that all vertices are in the same cell on level $\numLevels$.
    For a vertex $\vertexA$ contained in cell $\cell^\aLevel_i$ on level $\aLevel$, we denote the \emph{cell ID} of~$\vertexA$ on level~$\aLevel$ by~$\cellIdL{\vertexA}{\aLevel} = i$.

    \subparagraph*{TB.}
    We give a brief description of the TB algorithm.
    We refer to Appendix~\ref{app:tb} and~\cite{Wit15} for further details.
    A preprocessing step computes a set $\transfers\subseteq\stopEvents\times\stopEvents$ of transfers between pairs of stop events, which is sufficient to answer all queries optimally.
    First, it generates a set of potentially relevant transfers.
    Then, two pruning rules are applied to remove irrelevant transfers.
    All steps are trivial to parallelize, as trips are processed independently.
			
    \begin{algorithm}[t]
        \caption{
            TB $\texttt{Scan}$ method for round~$n$ and target stop~$\targetStop$.
        }\label{alg:tbts}
            \ForEach{$\tripSegment{\aTrip}{j}{k}\in\aQueue_{n}$}{
                \For{$i$ from $j$ to $k$\label{alg:tb:ts:loop1begin}}{
                    \lIfComment{target pruning}{$\arrivalTime(\aTrip[i])\geq\minTime$}{\Break}\label{alg:tb:targetpruning}
                    \IfComment{create cost vector}{$\arrivalTime(\aTrip[i])+\transfertime{\aStop(\aTrip[i])}{\targetStop}<\minTime$\label{alg:tb:improv}}{
                        $\minTime\leftarrow\arrivalTime(\aTrip[i])+\transfertime{\aStop(\aTrip[i])}{\targetStop}$\;
                        $\labels\leftarrow\labels\cup\left\{\left(\minTime,n\right)\right\}$, removing dominated cost vectors\;\label{alg:tb:domination}
                    }
                }
            }
            \ForEach{$\tripSegment{\aTrip}{j}{k}\in\aQueue_{n}$}{
                \For{$i$ from $j$ to $k$}{
                    \lIfComment{target pruning}{$\arrivalTime\left(\aTrip[i]\right)\geq\minTime$}{$k \gets i-1$}\label{alg:tb:targetpruning2}
                }
            }
            \ForEachComment{relax transfers}{$\tripSegment{\aTrip}{j}{k}\in\aQueue_{n}$}{
                \lForEach{$\aTransfer=\transfer{\aTrip[i]}{\aTrip'[i']}\in\transfers$ with $j \leq i \leq k$\label{alg:tb:ts:loop2begin}}{$\Enqueue\left(\aTransfer,\aQueue_{n+1}\right)$\label{alg:tb:ts:enqueue}
                }
            }
    \end{algorithm}

    The query algorithm operates in rounds, where round~$n$ finds journeys with exactly~$n$ trips by scanning trip segments collected in a FIFO queue~$\aQueue_n$ during the previous round.
    For each trip~$\aTrip$, the \emph{reached index}~$\reachedIndex\left(\aTrip\right)$ is the smallest~$i$ such that for each stop event~$\aTrip[j]$ with~$i \leq j \leq \absoluteVal{\aTrip}$, an event~$\aTrip'[j]$ with~$\aTrip' \preceq \aTrip$ has already been scanned.
    The algorithm also maintains the tentative Pareto front~$\labels$ and the earliest known arrival time~$\minTime$ at~$\targetStop$.

    The main $\Scan$ operation (see~\Cref{alg:tbts}) iterates over all trip segments $\tripSegment{\aTrip}{j}{k}$ in the queue $\aQueue_n$ three times.    
	For each event $\aTrip[i]$ with $j \leq i \leq k$, the first loop checks whether it is possible to reach~$\targetStop$ by exiting the trip at~$\aTrip[i]$ and taking a final footpath.
    If this improves~$\minTime$ (line~\ref{alg:tb:improv}), a corresponding cost vector is added to~$\labels$.
    The found journeys are represented implicitly using parent pointers, which are unpacked at the end of the query.
    \emph{Target pruning} is applied during the loop:
	If~$\arrivalTime(\aTrip[i]) \geq \minTime$ (line~\ref{alg:tb:targetpruning}), then the rest of the trip segment is skipped because all subsequent events have even later arrival times.
    The second loop applies target pruning again because the first loop may have reduced~$\minTime$ further.
    If the arrival time of a stop event~$\aTrip[i]$ exceeds~$\minTime$, then the subsegment~$\tripSegment{\aTrip}{i}{k}$ is cut off.
    
	The third loop relaxes the outgoing transfers of all events~$\aTrip[i]$ with $j \leq i \leq k$.
    These transfers are stored contiguously in memory, which allows them to be scanned in a single for-loop.
	For each relaxed transfer $\transfer{\tripA[i]}{\tripB[j]}$, the \Enqueue method is called.
    The first index at which~$\tripB$ can be exited is~$j+1$.
    If the reached index~$\reachedIndex(\tripB)$ is greater than~$j+1$, then the newly reached trip segment~$\tripSegment{\tripB}{j+1}{\reachedIndex(\tripB)}$ is enqueued.
    Afterwards, the reached index of~$\tripB$ and all later trips of the same line is updated.
    
	To fill the first queue $\aQueue_1$, the query algorithm collects all lines that depart from~$\sourceStop$ or any stop reachable from~$\sourceStop$ via a footpath.
	For each of these lines, the algorithm finds the earliest reachable trip and adds the corresponding trip segment to~$\aQueue_1$ via the \Enqueue operation.
    
	\section{T-REX (Transfer-Ranked EXploration)}
	\label{ch:trex}
    We now describe T-REX in detail.
    The first step is to compute a nested bipartition~$\bipartition_{\numLevels}$ of the timetable and the set~$\transfers$ of transfers, using the TB preprocessing.
    Then, the algorithm computes the \emph{rank}~$\rank{\aTransfer}$ of each transfer~$\aTransfer$, which is the lowest level~$\aLevel$ such that~$\aTransfer$ is not required to traverse the cell containing~$\aTransfer$ on level~$\aLevel$.
    Following the terminology of three-phase customizable techniques for road networks (\eg CRP, CCH), we call this phase the \emph{customization}.
    The ranks are used for pruning during a query:
    Let~$\aStop$ be the current stop and~$\aLevel$ the highest level such that the cell of~$\aStop$ does not contain the source or target stop.
    Then the cell must be traversed, so the search can ignore transfers with rank~$\aLevel$ and below.

    T-REX offers two conceptual improvements over previous partition-based approaches.
    Firstly, it computes pruning information for transfers rather than stop events or connections because this leads to strictly stronger pruning~\cite{Gro25}:
    if a transfer~$\aTransfer=\transfer{\tripA[j]}{\tripB[i]}$ is relevant, then both~$\tripA[j]$ and~$\tripB[i]$ are relevant, but the converse is not necessarily true.
    Secondly, we carefully redesign the query algorithm to ensure that the search space reduction translates into faster query times.
    This is challenging in scan-based algorithms (such as TB) because they rely on branch prediction and cache locality, which are weakened by extra pruning steps.
    This was one of the issues faced by ACSA~\cite{Dib18} and HypRAPTOR~\cite{Del17b}.
		
	\subparagraph*{Partitioning.}
    As in ACSA~\cite{Dib18}, we represent the network as a \emph{compact layout graph} $\layoutGraph = \left(\layoutVertices, \layoutEdges\right)$, in which vertices correspond to stops, and an edge~$(\vertexA,\vertexB)$ exists if there is at least one connection from an event at~$\vertexA$ to an event at~$\vertexB$.
    Stops that are connected by a footpath are contracted into a single vertex to ensure that transfers do not span multiple cells.
    Each vertex is weighted with the number of represented stops, and each edge with the number of represented connections.
    The nested bipartition~$\bipartition_{\numLevels}$ is then computed heuristically on~$\layoutGraph$, aiming to minimize the total weight of the cut edges, \ie edges that cross cell boundaries, while enforcing a maximum imbalance between the total vertex weights of the cells.
    
	\subparagraph*{Query.}
    In its most basic form, the T-REX query algorithm is identical to the TB query algorithm, but with one additional pruning step:
    Let $\sourceStop$ be the source and $\targetStop$ the target stop.
    A transfer~$\aTransfer$ starting from a stop~$\aStop$ is only relaxed if passes the \emph{LCL test}, which is
	\begin{equation}
		\label{eq:lcl}\rank{\aTransfer}\geq \min\left\{\lcl{\aStop}{\sourceStop}, \lcl
		{\aStop}{\targetStop}\right\}.
	\end{equation}
    This test is performed at the start of the $\Enqueue$ method.
    See Appendix~\ref{app:optimizations} for optimizations.

    This basic approach spends most of its time performing unsuccessful LCL tests for irrelevant transfers.
    We therefore redesign the query algorithm to avoid as many LCL tests as possible.
    Consider the scan of a stop event~$\aTrip[i]$ at the stop~$\aStop=\aStop(\aTrip[i])$.
    The relevant level for the LCL test of all outgoing transfers of~$\aTrip[i]$ is~$\aLevel = \min\{ \lcl{\aStop}{\sourceStop}, \lcl{\aStop}{\targetStop} \}$.
    Obviously, recalculating $\aLevel$ for every transfer is wasteful, but we can go one step further.
    On level~$\aLevel-1$, the stop~$\aStop$ is in a different cell than both~$\sourceStop$ and~$\targetStop$.
    If the next stop event~$\aTrip[i+1]$ is in the same cell as~$\aStop$ on level~$\aLevel-1$, then the relevant level for the LCL test remains the same.
    To exploit this, we add a new~$\Split$ procedure at the start of each round, which splits the enqueued trip segments into subsegments along cell boundaries.
    Then the relevant level for the LCL test remains the same within each subsegment.
    We do this in advance rather than on-the-fly to reduce the cache load during the~$\Scan$ operation.

    For every stop event~$\aTrip[i]$ and every level~$\aLevel$, we precalculate a value~$\successor{\aTrip[i]}{\aLevel}$, which is the smallest index~$j$ with~$i \leq j \leq \absoluteVal{\aTrip}$ such that~$\cellIdL{\aStop(\aTrip[i])}{\aLevel} \neq \cellIdL{\aStop(\aTrip[j])}{\aLevel}$, or~$\absoluteVal{\aTrip} + 1$ if there is no such index.
    For each round~$n$, we maintain an additional FIFO queue~$\tilde{\aQueue}_n$, which stores the enqueued trip segments before they are split.
    Each queue element is now a tuple~$(\tripSegment{\aTrip}{j}{k},\aLevel)$, where~$\aLevel$ is the relevant level for the LCL tests of~$\aTrip[j]$.
    To split this segment, the algorithm looks up~$i = \successor{\aTrip[j]}{\max\{\aLevel-1,0\}}$.
    Then the first subsegment is~$\tripSegment{\aTrip}{j}{k'}$ with~$k' = \min\{i-1,k\}$, which is added to~$\aQueue_n$ with level~$\aLevel$.
    If~$i \leq k$, then the algorithm performs the LCL test for~$\aTrip[i]$ and recurses on the remaining trip segment~$\tripSegment{\aTrip}{i}{k}$.

    To skip irrelevant transfers during the~$\Scan$ operation (see~\Cref{alg:trex-scan}), we precalculate for each level~$\aLevel$ a \emph{transfer overlay}~$\transfers_\aLevel$, which contains all transfers with rank at least~$\aLevel$.
    For a queue element~$(\tripSegment{\aTrip}{j}{k},\aLevel)$, the third loop of the~$\Scan$ operation now relaxes the outgoing transfers in~$\transfers_\aLevel$.
    No LCL test is performed because all relaxed transfers would pass it.
    Instead, the relevant level~$\aLevel$ is passed to~$\Enqueue$ and added to trip segments that are inserted into~$\tilde{\aQueue}_{n+1}$.
    The first loop of the~$\Scan$ operation, which relaxes outgoing footpaths to~$\targetStop$, only needs to be performed for trip segments that are in the same cell as~$\targetStop$ on level 0.
    To exploit this, we maintain a third queue~$\aQueue^t_n$ and only run the loop for the elements in this queue.
    When a subsegment~$(\tripSegment{\aTrip}{j'}{k'},\aLevel)$ is created during the~$\Split$ operation, we check whether the stop of~$\aTrip[j']$ is in the same level-$0$ cell as~$\targetStop$.
    If so, then this implies~$\aLevel = 0$ and therefore the entire subsegment is in the same level-$0$ cell.
    Hence, the subsegment is added to~$\aQueue^t_n$.
    
    Outside of the level-$0$ target cell, the only part of the $\Scan$ operation that scans the individual stop events of a trip segment~$\tripSegment{\aTrip}{j}{k}$ is now the second loop, which performs target pruning.
    Especially on higher levels, many trip segments have fewer outgoing transfers than stop events, so this step costs more time than it saves.
    We therefore make the target pruning check lazier: if the arrival time of the first stop event~$\aTrip[j]$ exceeds~$\minTime$, then the entire trip segment is skipped.
    As a result, the $\Scan$ operation now essentially treats each trip segment as a single node, whose outgoing transfers are relaxed in a single sweep.

    Unlike Dijkstra-based MLO algorithms~\cite{Jun02,Del17}, we do not add shortcuts other than the precomputed TB transfers.
    Although shortcuts that bypass trip segments might reduce the search space, the potential for savings is low because most journeys use only a few trips and the average number of outgoing transfers per trip is high ($56$ on our Europe instance).
    Moreover, the $\Enqueue$ method would become less efficient.
    When relaxing a shortcut that skips~$x$ trip segments, the discovered trip segment must be inserted into the queue~$\aQueue_{n+x-1}$.
    Hence, the correct queue must be fetched for each shortcut, which reduces cache locality.

    \begin{algorithm}[t]
        \caption{
            T-REX $\texttt{Scan}$ method for round~$n$ and target stop~$\targetStop$. Changes in blue.
        }\label{alg:trex-scan}
            \ForEach{$\tripSegment{\aTrip}{j}{k}\in\textcolor{blue}{\aQueue^t_n}$}{
                \For{$i$ from $j$ to $k$}{
                    \lIfComment{target pruning}{$\arrivalTime(\aTrip[i])\geq\minTime$}{\Break}
                    \IfComment{create cost vector}{$\arrivalTime(\aTrip[i])+\transfertime{\aStop(\aTrip[i])}{\targetStop}<\minTime$}{
                        $\minTime\leftarrow\arrivalTime(\aTrip[i])+\transfertime{\aStop(\aTrip[i])}{\targetStop}$\;
                        $\labels\leftarrow\labels\cup\left\{\left(\minTime,n\right)\right\}$, removing dominated cost vectors\;
                    }
                }
            }
            \ForEach{$\textcolor{blue}{(\tripSegment{\aTrip}{j}{k},l)}\in\aQueue_{n}$}{
                \textcolor{blue}{
                \lIfComment{lazy target pruning}{$\arrivalTime\left(\aTrip[j]\right)\geq\minTime$}{$k \gets j-1$}
                }
            }
            \ForEachComment{relax transfers}{$\textcolor{blue}{(\tripSegment{\aTrip}{j}{k},l)}\in\aQueue_{n}$}{
                \lForEach{$\aTransfer=\transfer{\aTrip[i]}{\aTrip'[i']}\in\textcolor{blue}{\transfers_l}$ with~$j \leq i \leq k$}{$\Enqueue\left(\aTransfer,\textcolor{blue}{\tilde{\aQueue}_{n+1}},\textcolor{blue}{l}\right)$
                }
            }
    \end{algorithm}
	
	\subparagraph*{Customization.}
    For each journey~$\aJourney$ found by TB and each transfer~$\aTransfer$ in~$\aJourney$, the customization must ensure that Equation~(\ref{eq:lcl}) is fulfilled.
    We observe the following:
    At every level~$\aLevel$ such that~$\sourceStop$ and~$\targetStop$ are not in the same cell in the nested bipartition~$\bipartition_{\numLevels}$, we can decompose~$\aJourney$ into event subjourneys along cell boundaries.
    For each subjourney that traverses a cell, the ranks of its transfers must be at least~$\aLevel+1$.
    For a cell~$\cell$, a stop event~$\aTrip[i]$ is an \emph{incoming border event} (IBE) if $\aTrip[i]$ does not lie in~$\cell$, but~$\aTrip[i+1]$ does.
    Analogously, $\aTrip[i]$ is an \emph{outgoing border event} (OBE) if~$\aTrip[i]$ belongs to~$\cell$, but~$\aTrip[i+1]$ does not.
    We find journeys that traverse~$\cell$ by performing a search from each IBE~$\aTrip[i]$ of cell~$\cell$ to all OBEs of~$\cell$ with a special variant of TB, which we call \emph{Event-TB}.
    Because all journeys must start with~$\tripA[i]$, the initial queue~$\aQueue_1$ is initialized with only the trip segment $\tripSegment{\tripA}{i}{\absoluteVal{\tripA}}$.
    Target pruning is omitted, and instead of tracking a tentative Pareto front, the first loop over the queue $\aQueue_n$ collects all reached OBEs in a set $\mathcal{O}$.
    The search is restricted to~$\cell$ by cutting off each trip segment after the first encountered OBE.
	Using this search, the customization operates in a bottom-up fashion:
    For every level $\aLevel$ from~$0$ to~$\numLevels-1$, the IBEs of all cells are processed in parallel.
    For each IBE~$\aTrip[i]$ of a cell~$\cell \in \allCells^\aLevel$, an Event-TB search is performed.
    During the $\Enqueue$ method, transfers with a rank smaller than~$\aLevel$ are discarded because they were already found to be irrelevant for traversing the subcells of~$\cell$.
    For each reached OBE~$\aTrip'[k] \in \mathcal{O}$, the found~$\aTrip[i]$-$\aTrip'[k]$-journey $\aJourney$ is unpacked and~$\rank{\aTransfer} = \aLevel+1$ is set for every transfer~$\aTransfer$ in~$\aJourney$.

    \subparagraph*{Correctness.}
    The correctness of T-REX relies on the \emph{subjourney closure property} of TB:
    \begin{restatable}{theorem}{subjourneyClosure}
        \label{th:event-subjourney-closure}
        Let~$\aJourney$ be a journey returned by TB for a query~$\query$ and let~$\eventSubjourney{\sourceStopEvent}{\targetStopEvent}$ be an event subjourney of~$\aJourney$.
        An Event-TB search from~$\sourceStopEvent$ finds the journey~$\eventSubjourney{\sourceStopEvent}{\targetStopEvent}$.
    \end{restatable}

    We give a detailed proof of this claim in Appendix~\ref{app:proof}.
    We stress that although it seems intuitively clear that TB has this property, the proof is far from trivial and requires careful analysis of the interaction between different pruning rules in the transfer generation and query phases.
    In fact, Baum \etal\cite{Bau23} showed that an analogous property does not hold for RAPTOR.
    Using this property, the correctness of T-REX is fairly straightforward.
    
    \begin{theorem}
		\label{th:trexcustocorrect}
        Let $\aJourney$ be a journey returned by TB for a query $\query = \left(\sourceStop, \targetStop, \departureTime\right)$.
		After the T-REX customization, every transfer $\aTransfer$ in $\aJourney$ fulfills Equation~(\ref{eq:lcl}).
	\end{theorem}
	
	\begin{proof}[Proof]	
		Let $\aTransfer = \left(\aTrip_m[j_m], \aTrip_n[i_n]\right)$,
		$\aStop=\aStop(\aTrip_m[j_m])$ and~$r = \min\left\{\lcl{\aStop}{\sourceStop}, \lcl
		{\aStop}{\targetStop}\right\}$.
        We show by induction that for each level $\aLevel$ with~$0 \leq \aLevel < r$, the customization extracts~$\aTransfer$ and sets~$\rank{\aTransfer} \gets \aLevel+1$.
        This ensures that~$\rank{\aTransfer} \geq r$ holds after the customization, \ie Equation~(\ref{eq:lcl}) holds.
        Let $\anEventSubJourney = \eventSubjourney{\tripA[i_a]}{\tripB[j_b]}$ be the maximal subjourney of $\aJourney$ such that $\anEventSubJourney$ uses $\aTransfer$ and all stop events of $\anEventSubJourney$ except (possibly) for~$\tripA[i_a]$ lie inside the cell $\cellIdL{\aStop}{\aLevel}$.
        Because~$\aLevel<r$, neither~$\sourceStop$ nor~$\targetStop$ is located in this cell, so it follows that~$\tripA[i_a]$ is an IBE of the cell and~$\tripB[j_b]$ is an OBE.
        We show that the Event-TB search from $\tripA[i_a]$ during the customization of cell $\cellIdL{\aStop}{l}$ finds~$\anEventSubJourney$.
        Without the rank-based pruning rule in the modified \Enqueue method, this follows from~\Cref{th:event-subjourney-closure}.
        This rule prunes the search at any transfer~$\aTransfer^*$ with~$\rank{\aTransfer^*} < \aLevel$.
        In the base case~$\aLevel=0$, this is trivially false for all transfers in~$\anEventSubJourney$, and otherwise it follows from the induction hypothesis.
        Hence, the search finds~$\anEventSubJourney$ and sets the rank of~$\aTransfer\in\anEventSubJourney$ to~$\aLevel+1$.
	\end{proof}

	\subparagraph*{Updates.}
    Real-time updates can affect all three components of the T-REX preprocessing: the partition, the transfer set, and the ranks.
    Changes to the compact layout graph (\eg rerouting a line) are fairly rare and can often be incorporated without recomputing the partition.
    Updating the transfer set is discussed in~\cite{Wit21}.
    The ranks can be updated in a bottom-up fashion by re-running the customization in the affected cells.
    Moreover, a cell~$\cell$ on level~$\aLevel$ does not need to be processed if the customization on level~$\aLevel-1$ did not change the ranks of any transfers.
    In Appendix~\ref{app:optimizations}, we present more sophisticated update algorithms that exploit the fact that updates usually only affect a short time window in practice.

	\section{Experiments}
	\label{ch:experiments}
	All algorithms are implemented in \cpp and compiled using GCC $14.2.0$ with \texttt{-march=native -O3}.
	The code is publicly available~\cite{githubGitHubPatrickSteilTREX, githubGitHubTransitRoutingFLASHTB}, except for CSA and ACSA, which were provided to us by the authors~\cite{Dib18}.
	All precomputations besides partitioning were run on a machine with two $\numprint{64}$-core AMD Epyc 7742 CPUs clocked at $\numprint{2.25}\,\si{\giga\hertz}$, with a boost frequency of $\numprint{3.4}\,\si{\giga\hertz}$, $\numprint{1024}\,\si{\giga\byte}$ of DDR4-3200 RAM, and $\numprint{256}\,\si{\mega\byte}$ of L3 cache.
	Partitioning and queries were performed on a machine with two $\numprint{8}$-core Intel Xeon Gold-6144 CPUs clocked at $\numprint{3.5}\,\si{\giga\hertz}$, with a boost frequency of $\numprint{4.2}\,\si{\giga\hertz}$, $\numprint{192}\,\si{\giga\byte}$ of DDR4-2666 RAM, and $\numprint{24.75}\,\si{\mega\byte}$ of L3 cache.
	
	\subparagraph*{Datasets.}
    Our datasets~\cite{sauer_2026_19697732} include a metropolitan region, two country-wide networks and a Europe-wide network; \cf Table~\ref{tab:datasets} (additional datasets are evaluated in Appendix~\ref{app:experiments:gb-berlin}).
    The datasets of Switzerland and Paris were extracted from official General Transit Feed Specification (GTFS) feeds~\cite{gtfsHomeGeneral}.
    The Germany dataset was extracted from gtfs.de~\cite{gtfsGTFSDEGTFS}.
    We created the Europe dataset by extracting the datasets from all European countries from Transitous~\cite{transitousTransitous} and merging with using the tool gtfstidy~\cite{githubGitHubPatrickbrgtfstidy}.
    The footpaths are taken from the GTFS data except for Europe, where we connected all pairs of stops closer than $\numprint{100}\,\si{\metre}$.
    For all networks, we then computed the transitive closure.
    As a result, the Europe dataset contains many footpaths longer than $\numprint{100}\,\si{\metre}$.
    To compute the traversal times, we assumed a walking speed of $\numprint{1.43}$\,m/s.
    For each dataset, we selected a service period of two consecutive weekdays.
    We grouped the trips into lines with the simple greedy algorithm described in~\cite{Bau23}.

	\begin{table}[ht]
        \caption{ 
            An overview of the networks used in our experiments.
            Also shown are the sizes of the TB transfer set~$\transfers$ and the compact layout graph~$\layoutGraph$.
        }
		\label{tab:datasets}
		\centering
		\begin{tabular*}{\textwidth}{@{\,}l@{\extracolsep{\fill}}r@{\extracolsep{\fill}}r@{\extracolsep{\fill}}r@{\extracolsep{\fill}}r@{\,}}
			\toprule                      & Europe                     & Germany                    & Switzerland               & Paris                      \\
			\midrule Stops                & $\numprint{1346013}$         & $\numprint{435550}$         & $\numprint{29045}$          & $\numprint{41757}$          \\
			Stop events                   & $\numprint{107252199}$       & $\numprint{30680868}$      & $\numprint{5032795}$        & $\numprint{4636238}$        \\
			Lines                         & $\numprint{384075}$          & $\numprint{202238}$          & $\numprint{15967}$          & $\numprint{9558}$          \\
			Trips                         & $\numprint{4848876}$         & $\numprint{1547126}$         & $\numprint{319159}$         & $\numprint{215526}$         \\
			Footpaths                     & $\numprint{1752418}$         & $\numprint{1116976}$         & $\numprint{22186}$          & $\numprint{445912}$          \\
            \noalign{\vskip 5pt}
			TB $\absoluteVal{\transfers}$ & $\numprint{269766664}$       & $\numprint{59181359}$       & $\numprint{8075627}$        & $\numprint{23284123}$        \\
			$\absoluteVal{\layoutVertices}$      & $\numprint{778419}$          & $\numprint{308004}$         & $\numprint{24257}$          & $\numprint{13765}$          \\
			$\absoluteVal{\layoutEdges}$  & $\numprint{2298358}$         & $\numprint{1038032}$         & $\numprint{60154}$          & $\numprint{45882}$          \\
			\bottomrule
		\end{tabular*}
	\end{table}

    \subparagraph*{Setup.}
    We compare T-REX to TB, RAPTOR, CSA, ACSA and FLASH-TB.
    In order to use compact data structures, we limit all round-based algorithms to $\numprint{16}$ rounds.
    Optimal journeys with more than $\numprint{16}$ trips are extremely rare and very unlikely to be considered desirable by users.
    Further implementation details are discussed in Appendix~\ref{app:experiments:implementation}.
    We report average query times over $\numprint{10000}$ queries, with the source and target stops chosen uniformly at random, and the departure times chosen uniformly at random within the first day of the service period.
    The time for journey unpacking, which takes a few microseconds, is excluded.
    For ACSA, we evaluated the same partitions as for T-REX and report the configuration with the lowest query times.
    For FLASH-TB, we computed the partitions using KaHIP, as in~\cite{Gro25}.

	\subparagraph*{Partitioning.}
    To compute the multilevel partition of the compact layout graph, we use a custom version of the open-source partitioner Mt-KaHyPar~\cite{Got22,Got23,Got24,githubGitHubKahyparmtkahypar}.
    Although Mt-KaHyPar was designed primarily for hypergraph partitioning, it also supports ordinary graphs.
    Normally, Mt-KaHyPar adaptively reduces the allowed imbalance on the higher levels to ensure that  the final $2^\numLevels$-partition still respects the global imbalance bound~$\varepsilon$~\cite{Sch16}.
    Our custom version~\cite{githubGitHubPatrickSteilmtkahypar} instead uses the given imbalance~$\varepsilon$ on each level.
    Allowing extreme imbalances at the lower levels can negatively affect both the customization time, as the larger cells are more expensive to process, and the query time, as more source-target pairs have a low LCL.
    However, we found that these effects were outweighed by the smaller cut size.
    We parallelized Mt-KaHyPar with~$6$ cores, as we observed that adding further cores degraded the partition quality.
    In all subsequent experiments, we evaluated four imbalance values ($\numprint{25}\,\%$, $\numprint{50}\,\%$, $\numprint{75}\,\%$ and $\numprint{100}\,\%$) for each network and each choice for the number of levels~$\numLevels$.
    We report the configuration with the lowest query times (cf.~\Cref{tab:varinglevel} and Appendix~\ref{app:experiments:metrics}).

    Details on the choice of the partitioner are given in Appendix~\ref{app:experiments:partition}.
    We observed that our version of Mt-KaHyPar produces much better results on our networks than KaHIP~\cite{San13,githubGitHubKaHIPKaHIP}.
    In particular, the cuts on the highest few levels are drastically smaller, with a reduction in the number of IBEs by up to two orders of magnitude.
    The partitioning time is reduced by a factor of~$\numprint{14.4}$, the customization time by a factor of~$\numprint{2.0}$ and the query time by a factor of~$\numprint{2.5}$.
	
    \subparagraph*{Preprocessing.}
    On the Europe instance, the customization requires $46$ minutes when run sequentially.
    The efficiency of the parallelization, \ie the speedup over sequential execution divided by the number of used threads, remains above $\numprint{0.6}$ for up to $\numprint{64}$ threads.
    For $\numprint{128}$ threads, the efficiency decreases to $\numprint{0.4}$, for a total customization time of~$\numprint{54}$ seconds.
    The decrease in the efficiency is likely to the limited memory bandwidth of the AMD Eypc machine, which has been observed before~\cite{Gro25}.
    \Cref{tab:prepro} reports the times for the different preprocessing phases of T-REX and ACSA.
    The T-REX customization takes roughly twice as long as the transfer generation.
    The total T-REX preprocessing time is similar to that of ACSA on the larger networks and faster on the smaller ones, despite having an additional phase.
    On the smaller networks, the transfer generation and customization are fast enough that they can be re-run every few seconds to incorporate updates.
    This is not feasible on Europe, where processing the entire two-day timetable takes $78$ seconds.
    However, updates typically only affect a small time window, which can be processed much faster.
    For example, processing the entire rush-hour time window between $7$ and $9$\,AM on the first day requires~$8$ seconds for the transfer generation and~$5$ seconds for the customization.

    \begin{table}[t]
		\caption{
    		Preprocessing times for T-REX and ACSA by phase, measured in seconds.
		}
		\label{tab:prepro}
		\centering
		\begin{tabular*}{\textwidth}{@{\,}l@{\extracolsep{\fill}}r@{\extracolsep{\fill}}r@{\extracolsep{\fill}}r@{\extracolsep{\fill}}r@{\extracolsep{\fill}}r@{\extracolsep{\fill}}r@{\extracolsep{\fill}}@{\,}}
			\toprule
            \multirow{2}{*}{Network} & \multirow{2}{*}{Partitioning} & \multicolumn{3}{c}{T-REX} & \multicolumn{2}{c}{ACSA} \\
            \cmidrule(){3-5} \cmidrule(){6-7}
            & & Transfers & Customization & Total & Customization & Total \\
			\midrule
            Europe & 53.52 & 24.02 & 54.28 & 131.82 & 94.19 & 147.71 \\
            Germany & 22.67 & 5.40 & 10.97 & 39.04 & 12.89 & 35.56 \\
            Switzerland & 1.49 & 0.42 & 0.85 & 2.76 & 2.72 & 4.21\\
            Paris & 1.04 & 2.09 & 4.82 & 7.95 & 17.19 & 18.22 \\
			\bottomrule
		\end{tabular*}
	\end{table}

    \subparagraph*{Memory Consumption.}
    T-REX without overlays has a memory overhead~$5$--$9$\,\% over TB, required for the cell IDs and transfer ranks.
    The overlay variant requires~$2$--$4$ times as much space as TB and~$12$--$18$ times as much as the timetable itself.
    This is mainly due to the $\successor{\aTrip[i]}{\aLevel}$ data structure and the transfer overlays, which exist per level.
    For Europe, the total memory consumption is $\numprint{8.1}\,\si{\giga\byte}$ for TB, $\numprint{8.5}\,\si{\giga\byte}$ for basic T-REX and~$\numprint{27.8}\,\si{\giga\byte}$ for the overlay variant.
    Of these, $\numprint{5.2}\,\si{\giga\byte}$ are required for the timetable and the transfers.
    We note that both TB and T-REX are tuned to prioritize query speed over a low memory footprint.
    We discuss further details and ways to reduce the memory consumption in Appendix~\ref{app:experiments:memory}.

    \subparagraph*{Query.}
    \Cref{tab:varinglevel} reports the performance of T-REX on Europe depending on the number of levels.
    The query time correlates most closely with the number of relaxed transfers, as most time is spent in the~$\Enqueue$ method.
    Adding more levels consistently improves the performance until a plateau is reached at nine levels.
    With a mean query time of $\numprint{8.7}\,\si{\milli\second}$, T-REX is fast enough for interactive applications.

    \begin{table}[ht]
		\caption{
            Performance of algorithms, averaged over $\numprint{10000}$ random queries.
        }
		\label{tab:uniformrandomquery:country}
		\centering
		\begin{tabular*}{\textwidth}{@{\,}l@{\extracolsep{\fill}}l@{\extracolsep{\fill}}c@{\extracolsep{\fill}}r@{\extracolsep{\fill}}r@{\extracolsep{\fill}}r@{\,}}
			\toprule
            \multirow{2}{*}{Network} & \multirow{2}{*}{Algorithm} & \multirow{2}{*}{Pareto} & \multirow{2}{*}{Query $\left[\si{\mu\second}\right]$} & Preprocessing & \multirow{2}{*}{$k$}\\
            & & & & $\left[ \mathrm{hh}{:}\mathrm{mm}{:}\mathrm{ss}\right]$ & \\
			\midrule 
			\multirow{5}{*}{Europe} & CSA       & $\circ$   & $\numprint{392486}$          & --                                                                        & --               \\
			& ACSA      & $\circ$   & $\numprint{36110}$           & \printTime{0}{2}{28}                                                     & $\numprint{16384}$ \\
			& RAPTOR    & $\bullet$ & $\numprint{715512}$          & --                                                                        & --               \\
			& TB        & $\bullet$ & $\numprint{196536}$          & \printTime{0}{0}{24}                                                     & --               \\
			& T-REX     & $\bullet$ & $\numprint{8725}$           & \printTime{0}{2}{12}                                                    & $\numprint{4096}$ \\
            \noalign{\vskip 8pt}
			\multirow{7}{*}{Germany} & CSA       & $\circ$   & $\numprint{83101}$           & --                                                                        & --               \\
			& ACSA      & $\circ$   & $\numprint{19206}$           & \printTime{0}{0}{36}                                                     & $\numprint{1024}$ \\
			& RAPTOR    & $\bullet$ & $\numprint{212255}$          & --                                                                        & --               \\
			& TB        & $\bullet$ & $\numprint{65060}$           & \printTime{0}{0}{5}                                                     & --               \\
			& T-REX     & $\bullet$ & $\numprint{4442}$            & \printTime{0}{0}{39}                                                     & $\numprint{16384}$ \\
			& FLASH-TB (moderate) & $\bullet$ & $\numprint{7871}$             & \printTime{29}{39}{30}                                                   & $\numprint{8}$\\
			& FLASH-TB (fastest) & $\bullet$ & $\numprint{82}$             & \printTime{31}{08}{21}                                                   & $\numprint{8192}$ \\
            \noalign{\vskip 8pt}
			\multirow{7}{*}{Switzerland} & CSA       & $\circ$   & $\numprint{3931}$               & --                                                                        & --               \\
			& ACSA      & $\circ$   & $\numprint{1901}$               & \printTime{0}{0}{4}                                                      & $\numprint{256}$ \\
			& RAPTOR    & $\bullet$ & $\numprint{10719}$           & --                                                                        & --               \\
			& TB        & $\bullet$ & $\numprint{4569}$            & < \printTime{0}{0}{1}                                                      & --               \\
			& T-REX     & $\bullet$ & $\numprint{710}$            & \printTime{0}{0}{3}                                                      & $\numprint{512}$ \\
			& FLASH-TB (moderate) & $\bullet$ & $\numprint{454}$              & \printTime{0}{23}{53}                                                     & $\numprint{8}$ \\
			& FLASH-TB (fastest) & $\bullet$ & $\numprint{15}$              & \printTime{0}{16}{43}                                                     & $\numprint{8192}$ \\
            \noalign{\vskip 8pt}
			\multirow{7}{*}{Paris} & CSA       & $\circ$   & $\numprint{3280}$            & --                                                                        & --                \\
			& ACSA      & $\circ$   & $\numprint{6312}$            & \printTime{0}{0}{18}                                                     & $\numprint{64}$    \\
			& RAPTOR    & $\bullet$ & $\numprint{8562}$            & --                                                                        & --                \\
			& TB        & $\bullet$ & $\numprint{3472}$            & \printTime{0}{0}{2}                                                     & --                \\
			& T-REX     & $\bullet$ & $\numprint{1641}$            & \printTime{0}{0}{8}                                                     & $\numprint{256}$   \\
			& FLASH-TB (moderate) & $\bullet$ & $\numprint{899}$              & \printTime{0}{29}{25}                                                    & $\numprint{8}$ \\
			& FLASH-TB (fastest)  & $\bullet$ & $\numprint{22}$              & \printTime{0}{29}{46}                                                    & $\numprint{16384}$ \\
			\bottomrule
		\end{tabular*}
	\end{table}
    
    \Cref{tab:uniformrandomquery:country} compares the performance of T-REX to other algorithms.
    The speedup over TB is $\numprint{22.2}$ on Europe, $\numprint{14.6}$ on Germany, $\numprint{6.4}$ on Switzerland and $\numprint{2.1}$ on Paris.
    Compared to RAPTOR, which is the fastest Pareto-optimal algorithm without significant preprocessing, T-REX achieves a speedup between $\numprint{5.2}$ and~$\numprint{82.0}$.
    The speedup over TB on each network is strongly correlated with the distribution of the ranks among the transfer set, depicted in \Cref{fig:rank}.
    On Europe and Germany, more than~$50$\,\% of all transfers have a rank of~$0$, which indicates that they are only relevant for local journeys.
    On the smaller networks, this share is smaller but still a sizable plurality.
    The smaller speedup on Paris is explained by the flatter rank distribution.
    This is in line with similar findings for ACSA~\cite{Dib18} and FLASH-TB~\cite{Gro25}, which suggest that metropolitan networks exhibit only a weak hierarchy.
    However, we stress that RAPTOR is already fast enough for interactive applications on these networks.

    \begin{figure}[t]
        \centering
        \begin{tikzpicture}

\begin{groupplot}[
    group style={
        group size=2 by 2, 
        horizontal sep=2cm, 
        vertical sep=2cm},
    width=6cm,
    height=5cm,
    ybar=0pt,
    /pgf/bar width=7pt,
    ymin=0,
    ymax=100,
    xlabel={Rank},
    ylabel={Transfers (\%)},
    ymajorgrids,
    grid style={dashed,gray!60},
    every axis/.append style={
        tick label style={font=\footnotesize},
        label style={font=\small},
    }
]

\nextgroupplot[title={Europe}]
\addplot+[draw=black, fill=blue!40]
table[x=rank,y=percent,col sep=comma] {plot_scripts/data/europe.csv};

\nextgroupplot[title={Germany}]
\addplot+[draw=black, fill=blue!40]
table[x=rank,y=percent,col sep=comma] {plot_scripts/data/germany.csv};

\nextgroupplot[title={Switzerland}]
\addplot+[draw=black, fill=blue!40]
table[x=rank,y=percent,col sep=comma] {plot_scripts/data/switzerland.csv};

\nextgroupplot[title={Paris}]
\addplot+[draw=black, fill=blue!40]
table[x=rank,y=percent,col sep=comma] {plot_scripts/data/paris.csv};

\end{groupplot}

\end{tikzpicture}
        \caption{
            Distribution of the ranks among the transfer set. 
		}
		\label{fig:rank}
    \end{figure}
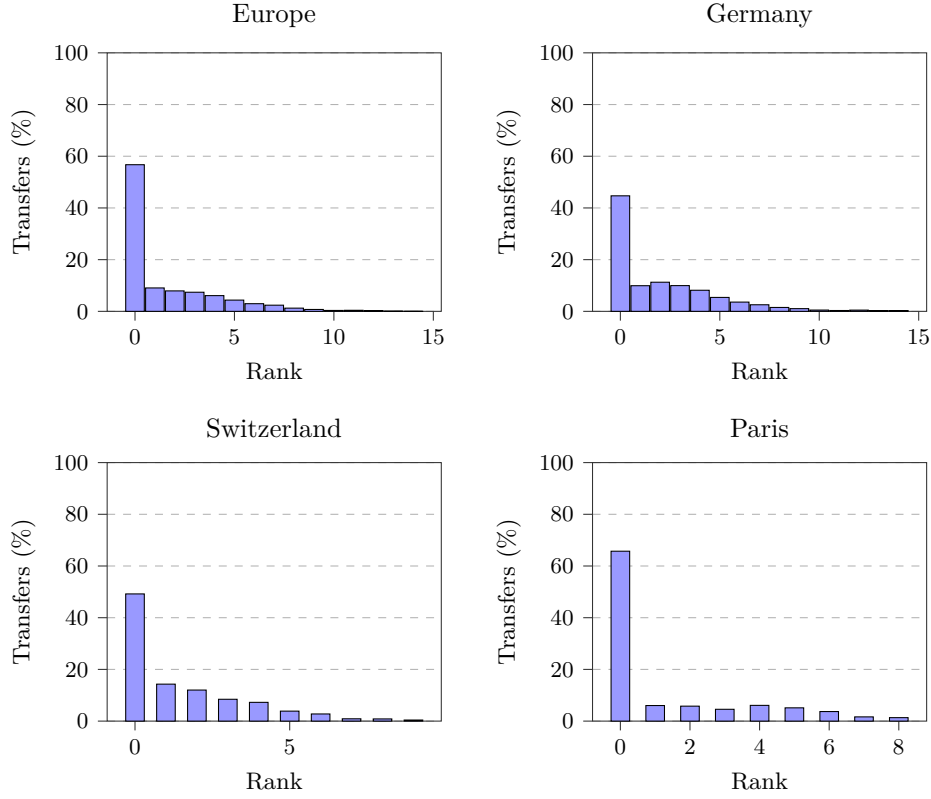

    The preprocessing of FLASH-TB does not scale to Europe.
    On the other networks, the fastest configuration is faster than T-REX by up to two orders of magnitude.
    A moderate configuration with~$8$ cells (with a memory overhead similar to T-REX without overlays) has query times similar to T-REX (more details in Appendix~\ref{app:experiments:flash-tb}).
    Compared to ACSA, which only minimizes the arrival time, T-REX is faster by a factor of~$3$--$4$ and consistently achieves a larger speedup over TB than ACSA over CSA.
    On Paris, ACSA even causes a slowdown because the overhead for merging the relevant connections dominates the scanning time.
    By contrast, T-REX achieves a speedup due to a much lower overhead during the query.

    In Appendix~\ref{app:experiments:georank}, we additionally evaluate how the query time of T-REX scales with the distance between source and target.
    In Appendix~\ref{app:experiments:profile}, we evaluate profile queries, in which the objective is to find Pareto-optimal journeys within a range of possible departure times.

    \subparagraph*{Comparison with Scalable TP.}
    As we do not have access to the code for Scalable TP, we compare to the performance reported in~\cite{Bas16} for a Germany instance with a similar number of stop events as ours, but half as many stops.
    Reported are a sequential precomputation time of $\numprint{16.5}\,\si{\hour}$ and a memory overhead of $\numprint{1160}\,\si{\mega\byte}$.
    By contrast, T-REX without overlays requires $10\,\si{\minute}$ and a memory overhead of $\numprint{1763}\,\si{\mega\byte}$.
	The reported query times for Scalable TP are $\numprint{0.1}\,\si{\milli\second}$ for local queries in convex clusters, $\numprint{5}\,\si{\milli\second}$ in non-convex clusters, and $\numprint{32}\,\si{\milli\second}$ for global queries.
    T-REX without overlays requires $\numprint{7.8}\,\si{\milli\second}$ for random queries and is much faster for local queries (\cf Appendix~\ref{app:experiments:georank}).
    Accounting for differences in machines and instances, we observe that T-REX significantly outperforms Scalable TP in terms of preprocessing and query time and has a comparable memory consumption if overlays are not used.
	
	\begin{table}[t]
		\caption{
            Performance of T-REX (with overlays) with different numbers of levels on the Europe network, averaged over $\numprint{10000}$ random queries.
    		For~$0$ levels, we report the metrics for TB.
        }
		\label{tab:varinglevel}
		\centering
		\begin{tabular*}{\textwidth}{@{\,}r@{\extracolsep{\fill}}r@{\extracolsep{\fill}}r@{\extracolsep{\fill}}r@{\extracolsep{\fill}}r@{\extracolsep{\fill}}r@{\,}}
			\toprule
			\multirow{2}{*}{Levels} & Imbalance & \multirow{2}{*}{Scanned trips} & \multirow{2}{*}{Relaxed transfers} & Query & Customization \\
            &  $\left[\%\right]$ & & & $\left[\si{\mu\second}\right]$ & $\left[ \mathrm{mm}{:}\mathrm{ss}\right]$\\
			\midrule
            $\numprint{0}$ & -- & $\numprint{750177}$ & $\numprint{14648126}$ & $\numprint{196536}$ & -- \\
            $\numprint{1}$ & $\numprint{100}$ & $\numprint{660736}$ & $\numprint{12912414}$ & $\numprint{185961}$ & \highlightCell \printMinuteSeconds{0}{4} \\
            $\numprint{2}$& $\numprint{25}$ & $\numprint{466844}$ & $\numprint{9171587}$ & $\numprint{125162}$ & \printMinuteSeconds{0}{20} \\
            $\numprint{3}$ & $\numprint{25}$ & $\numprint{275178}$ & $\numprint{5262951}$ & $\numprint{70455}$ & \printMinuteSeconds{0}{22} \\
            $\numprint{4}$ & $\numprint{50}$ & $\numprint{204557}$ & $\numprint{3751287}$ & $\numprint{50840}$ & \printMinuteSeconds{0}{26} \\
            $\numprint{5}$ & $\numprint{25}$ & $\numprint{105291}$ & $\numprint{1767891}$ & $\numprint{25034}$ & \printMinuteSeconds{0}{18} \\
            $\numprint{6}$ & $\numprint{25}$ & $\numprint{77219}$ & $\numprint{1151529}$ & $\numprint{17650}$ & \printMinuteSeconds{0}{21} \\
			$\numprint{7}$ & $\numprint{25}$ & $\numprint{55840}$ & $\numprint{719835}$ & $\numprint{12409}$ & \printMinuteSeconds{0}{18} \\
			$\numprint{8}$ & $\numprint{50}$ & $\numprint{52291}$ & $\numprint{652466}$ & $\numprint{11262}$ & \printMinuteSeconds{0}{19} \\
			$\numprint{9}$ & $\numprint{25}$ & $\numprint{44365}$ & $\numprint{445781}$ & $\numprint{9224}$ & \printMinuteSeconds{0}{26} \\
			$\numprint{10}$ & $\numprint{25}$ & $\numprint{43574}$ & $\numprint{409629}$ & $\numprint{9031}$ & \printMinuteSeconds{0}{36} \\
			$\numprint{11}$ & $\numprint{50}$ & $\numprint{44366}$ & $\numprint{415962}$ & $\numprint{9603}$ & \printMinuteSeconds{0}{32} \\
			$\numprint{12}$ & $\numprint{25}$ & $\numprint{42951}$ & $\numprint{371848}$ & $\numprint{9596}$ & \printMinuteSeconds{0}{47} \\
            $\numprint{13}$ & $\numprint{25}$ & $\numprint{43530}$ & $\numprint{377816}$ & $\numprint{9104}$ & \printMinuteSeconds{0}{58} \\
            $\numprint{14}$ & $\numprint{50}$ & \highlightCell $\numprint{41966}$ & $\numprint{366894}$ & \highlightCell $\numprint{8725}$ & \printMinuteSeconds{0}{54} \\
            $\numprint{15}$ & $\numprint{50}$ & $\numprint{42212}$ & \highlightCell $\numprint{358524}$ & $\numprint{9006}$ & \printMinuteSeconds{1}{3} \\
            $\numprint{16}$ & $\numprint{75}$ & $\numprint{43870}$ & $\numprint{404335}$ & $\numprint{9469}$ & \printMinuteSeconds{0}{53} \\
			\bottomrule
		\end{tabular*}
	\end{table}

    \subparagraph*{Impact of Ranked Element and Overlays.}
    \Cref{tab:rankedelement} reports the performance of T-REX without overlays if the ranks are assigned to stop events instead of transfers, \ie the $\Enqueue$ operation for a transfer~$\transfer{\tripA[j]}{\tripB[i]}$ performs the LCL test using~$\rank{\tripB[i]}$.
    This variant achieves a speedup of $\numprint{3.3}$ over TB.
    Ranking transfers instead yields a further speedup of $\numprint{2.4}$, which supports the findings by Großmann \etal\cite{Gro25} that transfers offer more fine-grained information than stop events.
    With ranked transfers but without overlays, the number of relaxed transfers is reduced by a factor of $\numprint{7.6}$ over TB, yielding a speedup of $\numprint{7.9}$.
    Overlays lead to further improvements despite the overhead of the $\Split$ method: they reduce the number of relaxed transfers by a factor of $\numprint{4.8}$ and provide a further speedup of $\numprint{2.5}$.

    \begin{table}[t]
		\caption{
    		Query performance on Europe (14 levels, $50\,\%$ imbalance), depending on which network element is ranked. In the overlay variant, scanned trips refer to the trip segments after splitting.
		}
		\label{tab:rankedelement}
		\centering
		\begin{tabular*}{\textwidth}{@{\,}l@{\extracolsep{\fill}}c@{\extracolsep{\fill}}r@{\extracolsep{\fill}}r@{\extracolsep{\fill}}r@{\extracolsep{\fill}}@{\,}}
			\toprule
            Ranks for & Overlays & Scanned trips & Relaxed transfers & Query $\left[\si{\mu\second}\right]$ \\
			\midrule
            None & $\circ$ & \numprint{750177} & \numprint{14648126} & \numprint{196536} \\
            Stop events & $\circ$ & \numprint{79774} & \numprint{3272304}  & \numprint{46706} \\
            Transfers & $\circ$ & \numprint{22696} & \numprint{1180213}  & \numprint{15949} \\
            Transfers & $\bullet$ & \numprint{41966} & \numprint{366894}   & \numprint{8725} \\
			\bottomrule
		\end{tabular*}
	\end{table}

	\section{Conclusion}
	\label{ch:conclusion}

	We presented T-REX, a new public transit journey planning algorithm that integrates hierarchical pruning information into Trip-Based Public Transit Routing~\cite{Wit15} (TB).
    T-REX achieves a speedup of up to~20 over TB, offering query times of less than $\numprint{10}\,\si{\milli\second}$ on continent-sized networks with moderate memory overhead and only a few minutes of precomputation time.
    Together with TB and FLASH-TB, T-REX demonstrates that precomputed transfers are a powerful tool for accelerating public transit journey planning and that they harmonize much better with additional pruning information than other network elements.
    	
	An open problem is that the performance of the TB transfer generation step is dependent on the size of the footpath set.
    Although ULTRA~\cite{Bau23} can be used to extend TB to unlimited footpaths, its preprocessing step does not scale well to large networks.
    For T-REX, we introduced further restrictions by requiring the footpath set to be transitively closed, and by contracting footpaths during partitioning.
    This was done mainly for ease of exposition, and we believe that these restrictions can be lifted.
    For example, initial footpaths that leave the source cell can be handled by treating their incident stop events as border events.
    However, this introduces new scalability issues, and thus extending T-REX to unlimited footpaths will likely require new ideas.
    Other interesting avenues for future research include combining hierarchical and goal-directed pruning in a TB-based algorithm, and extending T-REX to handle timetables that span longer time periods, which has already been done for TB~\cite{Wit21}.

    \bibliography{bibliography}
	\clearpage
    \appendix
    \section{Details on Trip-Based Public Transit Routing}
    \label{app:tb}
    We provide details on the TB algorithm that were omitted from~\Cref{ch:fundamentals}, specifically the transfer precomputation and the journey unpacking step.

    \subparagraph*{Transfer Precomputation.}
    This step computes a set $\transfers\subseteq\stopEvents\times\stopEvents$ of transfers between pairs of stop events that is sufficient for answering all queries optimally.
    This is done in three steps:
    First, a large set of potentially relevant transfers is generated.
    Then, two pruning rules are applied to remove irrelevant transfers.
    All three steps are trivial to parallelize, as trips are processed independently.

    For the initial generation step, Witt~\cite{Wit15} originally proposed a simple scheme.
    It generates all transfers~$\transfer{\tripA[j]}{\tripB[i]}$ such that~$1 < j \leq \absoluteVal{\tripA}$, $1 \leq u < \absoluteVal{\tripB}$, and~$\tripB$ is the earliest trip of its line that can be entered at~$\aStop(\tripB[i])$, \ie the earliest trip that fulfills
	\begin{equation*}
		\arrTime{\tripA[j]}+\transfertime{\aStop(\tripA[j])}{\aStop(\tripB[i])}\leq \depTime {\tripB[i]}.
	\end{equation*}
    The only situation if which the transfer is directly discarded is if both trips belong to the same line and staying seated in $\tripA$ is preferable, \ie if $\tripA \preceq \tripB$ and $j \leq i$.
    In our implementation, we use the alternative transfer generation scheme proposed by Lehoux and Loiodice~\cite{Leh20}, which improves the preprocessing time by discarding some superfluous transfers before they are generated.
    If a transfer~$\transfer{\tripA[i]}{\tripB[j]}$ exists, then all other transfers of the form~$\transfer{\tripA[k]}{\tripC[\ell]}$ with~$k \leq i$, $\ell \geq j$ and~$\tripC\succeq\tripB$ are not generated.

    After the initial transfer set has been generated, the \emph{U-turn rule} prunes transfers $\aTransfer = \left(\tripA[i], \tripB[i]\right)$ with $\aStop(\tripA[j-1]) = \aStop(\tripB[i+1])$ and $ \left(\tripA[j-1], \tripB[i+1]\right) \in \transfers$.
    Note that the original TB publication~\cite{Wit15} used a weaker variant of the second condition: it was merely required that~$\arrivalTime(\tripA[j-1]) \leq \arrivalTime(\tripB[i+1])$, i.e., the transfer $ \left(\tripA[j-1], \tripB[i+1]\right)$ is valid but not necessarily included in~$\transfers$.
    This is correct as long as the U-turn rule is applied first, but it causes issues if other pruning rules are applied before it.
    
	The last rule, which is referred to as \textit{latest-exit} rule~\cite{Gro25}, removes transfers that are ``dominated'' by others starting from the same trip. 
    For a transfer $\aTransfer = \left(\tripA[j], \tripB[i]\right)$, consider the journey $\aJourney = \left< \tripSegment{\tripA}{k}{j}, \tripSegment{\tripB}{i}{l}\right>$. The transfer can be pruned if there exists another journey $\aJourney' = \left< \tripSegment{\tripA}{k}{j'}, \tripSegment{\tripC}{i'}{l'}, \footpath\right>$, with $j < j'$ and $\footpath = \left(\aStop\left(\tripC[l']\right), \aStop\left(\tripB[l]\right)\right)$ which has a better (or equal) arrival time at $\aStop\left(\tripB[l]\right)$.

    \subparagraph*{Enqueue Operation.}
    Pseudocode for the $\Enqueue$ operation of a transfer~$\transfer{\tripA[i]}{\tripB[j]}$ in round~$n$ is shown in~\Cref{alg:tb:enq}.
    
    \begin{algorithm}[t]
        \caption{
            TB enqueuing operation.
        }\label{alg:tb:enq}
        \myproc{\Enqueue{$\transfer{\tripA[i]}{\tripB[j]},\aQueue_{n+1}$\label{alg:tb:enq:begin}}}{
            \lIfComment{trip segment already reached}{$\reachedIndex\left(\tripB\right)\leq{}j+1$}{\Return}
            $\aQueue_{n+1}\leftarrow\aQueue_{n+1}\cup\left\{\tripSegment{\tripB}{j+1}{\reachedIndex(\tripB)-1}\right\}$\;
            \ForEachComment{update reached index}{$\tripB'\succeq\tripB$}{
                \lIf{$\reachedIndex\left(\tripB'\right)\leq{}j+1$}{\Break}
                $\reachedIndex\left(\tripB'\right)\leftarrow{}j+1$\label{alg:tb:enq:end}\;
            }
        }
    \end{algorithm}

    \subparagraph*{Journey Unpacking.}
    The found journeys are represented implicitly using parent pointers.
    Let~$\tripSegment{\aTrip_{n-1}}{j_{n-1}}{k_{n-1}} \in \aQueue_{n-1}$ be a trip segment that is scanned in round~$n-1$, and let~$\transfer{\aTrip_{n-1}[i_{n-1}]}{\aTrip_n[j_n]}$ be an outgoing transfer that is relaxed, and let~$\tripSegment{\aTrip_n}{j_n}{k_n}$ be the trip segment that is added to~$\aQueue_n$ as a result.
    Then~$\tripSegment{\aTrip_n}{j_n}{k_n}$ stores a pointer to~$\tripSegment{\aTrip_{n-1}}{j_{n-1}}{k_{n-1}} \in \aQueue_{n-1}$.
    The index~$i_{n-1}$ at which~$\aTrip_{n-1}$ is not stored explicitly but reconstructed during journey unpacking by rescanning the outgoing transfers of~$\tripSegment{\aTrip_{n-1}}{j_{n-1}}{k_{n-1}}$ and choosing the first one that leads to~$\aTrip_n[j_n]$.
    This is faster than tracking this value explicitly during the query.
    Along with each cost vector in~$\labels$, the algorithm also stores a pointer to the trip segment from which the final footpath was relaxed.
    Using this information, the journey can be reconstructed by iteratively following the chain of parent pointers to the first trip segment in~$\aQueue_1$.

    \subparagraph*{Profile Queries.}
    TB can be extended to answer \emph{profile queries}.
    A profile query $\query = (\sourceStop, \targetStop, [\atime_{\mathrm{start}}, \atime_{\mathrm{end}}])$ consists of a source stop~$\sourceStop$, a target stop~$\targetStop$ and an interval~$[\atime_{\mathrm{start}}, \atime_{\mathrm{end}}]$ of departure times.
    It asks for sets~$C$ and~$\paretoRep$ such that $C$ is the union of the Pareto fronts for all  queries~$\query'=(\sourceStop,\targetStop,\departureTime)$ and~$\paretoRep$ is the union of representative sets for all such queries.
    To differentiate between the two query types, we refer to queries with a single given departure time as \emph{fixed departure time queries}.
    
    To answer a profile query~$\query = (\sourceStop, \targetStop, [\atime_{\mathrm{start}}, \atime_{\mathrm{end}}])$, the Profile-TB query algorithm exploits the \emph{self-pruning}~\cite{Del12} principle: solutions computed by a fixed departure time query for the departure time~$\departureTime$ are still feasible for an earlier departure time~$\departureTime'$.
    Furthermore, the set of relevant departure times is discrete: it is the set of all~$\departureTime(\aTrip[i])$ such that there is an initial footpath~$(\sourceStop,\aStop(\aTrip[i]))$ with~$\departureTime(\aTrip[i])-\transfertime{\sourceStop}{\aStop(\aTrip[i])} \in [\atime_{\mathrm{start}}, \atime_{\mathrm{end}}]$.
    Note that the initial footpath may be empty, \ie $\sourceStop = \aStop(\aTrip[i])$.
    These departure times are collected, sorted in descending order and the regular TB query algorithm is performed for each of them.
    The crucial optimization is that the reached indices and cost vectors are not reset between each run.
    This ensures that journeys found for a later departure time can prune dominated journeys for an earlier one.

    \section{T-REX Optimizations}
    \label{app:optimizations}
    In this section, we discuss performance optimizations for the T-REX query and customization that were omitted from~\Cref{ch:trex}.

    \subparagraph*{Query.}
    In the T-REX query, recall that a transfer~$\aTransfer$ starting from a stop~$\aStop$ only needs to be relaxed if it passes the LCL test, which is
    \begin{equation}\label{eq:lcl-app}
		\rank{\aTransfer}\geq \min\left\{\lcl{\aStop}{\sourceStop}, \lcl{\aStop}{\targetStop}\right\}.
	\end{equation}
    This test can be performed efficiently via bit manipulation.
    For simplicity, we write $\cellId{\vertexA}=\cellIdL{\vertexA}{0}$.
    Note that in binary representation, we have $\cellIdL{\vertexA}{\aLevel} = \cellId{\vertexA} \shift \aLevel$, where $\shift$ denotes the right-shift operator.
    In \cpp, the LCL of two stops $\stopA, \stopB$ can be computed as
  	\begin{equation*}
	    \lcl{\stopA}{\stopB} = \texttt{std::bit\_width}\left(\cellId{\stopA} \oplus \cellId{\stopB}\right),
	\end{equation*}
    where $\oplus$ denotes bitwise XOR.
	The XOR value has a 1 for each level in which both vertices are in different cells.
	We obtain the highest such level with the \texttt{C++20} function \texttt{std::bit\_width}, which returns the number of bits needed to represent the number passed as an argument.
    For example, let~$\cellId{\stopA}=001$ and
	$\cellId{\stopB}=011$ in binary representation. 
    Bitwise XOR yields $\cellId{\stopA}\oplus
	\cellId{\stopB}=010$.
    Hence, $\lcl{\stopA}{\stopB} = \texttt{std::bit\_width}(010) = 2$.
    In the basic query variant, where the LCL test is performed explicitly in the~$\Enqueue$ method, we can make it more efficient without needing to calculate the two LCL values explicitly.
    Instead, we shift the XOR result exactly as many bits to the right as the rank of the transfer.
    Formally, we implement the test as
    \begin{equation}
        \label{eq:shifteq}
        \left[\left(\cellId{\aStop} \oplus \cellId{\sourceStop}\right) \shift \rank{\aTransfer}\right] \neq 0
        \;\mathbin{\&\&}\;
        \left[\left(\cellId{\aStop} \oplus \cellId{\targetStop}\right) \shift \rank{\aTransfer}\right] \neq 0,
	\end{equation}
    where $\mathbin{\&\&}$ denotes the logical AND operator.
    In the query variant with overlays, no explicit LCL test is performed, so this optimization is irrelevant.

    \begin{observation}
        Equation~(\ref{eq:shifteq}) is equivalent to the negation of Equation~(\ref{eq:lcl-app}).
    \end{observation}
    \begin{proof}
        For $\aStop' \in \{\sourceStop,\targetStop\}$, we have
        \begin{align*}
            & \left[\left(\cellId{\aStop} \oplus \cellId{\aStop'}\right) \shift \rank{\aTransfer} \right] = 0\; \\
            \iff\;& \texttt{std::bit\_width}\left(\cellId{\aStop} \oplus \cellId{\aStop'}\right) \leq \rank{\aTransfer} \\
            \iff\;& \lcl{\aStop}{\aStop'} \leq \rank{\aTransfer}.
        \end{align*}
        Thus, Equation~(\ref{eq:shifteq}) is equivalent to
        \begin{equation*}
            \lcl{\aStop}{\sourceStop} > \rank{\aTransfer} \land \lcl{\aStop}{\targetStop} > \rank{\aTransfer}
        \end{equation*}
        and thus
        \begin{equation*}
            \rank{\aTransfer} < \min\{\lcl{\aStop}{\sourceStop},\lcl{\aStop}{\targetStop}\},
        \end{equation*}
        which is the negation of Equation~(\ref{eq:lcl-app}).
    \end{proof}

    \subparagraph*{Customization.}
    Recall that the customization processes the levels in a bottom-up fashion.
    For each IBE~$\aTrip[i]$ on a level~$\aLevel$, it runs a one-to-all Event-TB search from~$\aTrip[i]$ restricted to the cell of~$\aTrip[i]$ on level~$\aLevel$.
    For each reached OBE, the found journey is unpacked and the ranks of the contained transfers are set to~$\aLevel+1$.
    We implement the following optimizations:
    \begin{itemize}
        \item Because an IBE on level~$\aLevel$ is also an IBE on all higher levels, the algorithm first collects all IBEs on level~$0$.
        When moving to a higher level, all events that are no longer IBEs on this level are discarded.
        \item The search space of each Event-TB query is quite small, so it would be comparatively expensive to reset all reached indices after every query.
        To avoid this, we use the \emph{timestamped reached index} introduced for FLASH-TB~\cite{Gro25}. 
        For each trip~$\aTrip$, we store a 16-bit \emph{timestamp}~$\eta(\aTrip)$, which indicates the last query in which~$\reachedIndex(\aTrip)$ was modified.
        When~$\reachedIndex(\aTrip)$ is accessed during the $i$-th query and~$\eta\left(\aTrip\right) \neq i$, $\reachedIndex(\aTrip)$ is reset to~$\absoluteVal{\aTrip}$ and~$\eta(\aTrip)$ is set to $i$.
        Every~$2^{16}$ queries, the timestamp overflows, so all reached indices are reset.
        \item Consider the Event-TB search for an IBE~$\aTrip[i]$.
        Because the arrival time of a stop event~$\stopEvent$ is fixed, the Pareto front for the~$\aTrip[i]$-$\stopEvent$-query contains at most one cost vector, which is the one with the minimal number of used trips.
        Therefore, the parent pointers computed by Event-TB form a shortest path tree in which transfers do not appear multiple times.
        We exploit this by marking transfers once they are encountered during the journey unpacking.
        If a marked transfer~$\aTransfer$ is encountered again while unpacking another journey, then we know that both journeys share the same prefix up to~$\aTransfer$ and that all transfers in this prefix are already marked.
        Hence, the rest of the journeys does not need to be unpacked because these transfers already have the correct rank.
        As with the reached index, we use timestamps to mark transfers without having to reset the markings.
    \end{itemize}

    A notable optimization that we do not include in the customization is profile search.
    For each pair of line~$\aLine$ and index~$i$, one could group all IBEs of the form~$\aTrip[i]$ with~$\aLine(\aTrip)=\aLine$ together and process them in descending order of departure time.
    Then, one could attempt to benefit from self-pruning by not resetting the reached indices between the Event-TB searches.
    The motivation behind this is as follows:
    Let~$\tripA[i]$ and~$\tripA'[i]$ be IBEs with~$\tripA \prec \tripA'$, and let~$\tripB[j]$ be an OBE.
    If there is a~$\tripA'[i]$-$\tripB[j]$-journey~$\aJourney'_a$ with~$k$ trips, then there is also a~$\tripA[i]$-$\tripB[j]$-journey~$\aJourney_a$ with~$k$ trips.
    Furthermore, for any Pareto-optimal journey~$\aJourney$ that contains~$\aJourney_a$ as a subjourney, we can construct an equivalent journey~$\aJourney'$ that contains~$\aJourney'_a$ instead.
    Therefore, it might be tempting to assume that the customization does not need find~$\aJourney_a$ and set the ranks of its transfers accordingly.

    This approach has two issues.
    Firstly, it is not clear that the transfers in~$\aJourney'$ are contained in the generated transfer set~$\transfers$.
    Secondly, the stated assumption is not compatible with reached index pruning during the query, as noted in~\cite{Bau23,Gro25}.
    If the query algorithm reaches~$\tripA$, it will not scan~$\tripA'$ due to reached index pruning, and therefore it will not find~$\aJourney'_a$.
    However, if self-pruning is applied during the customization, it will also fail to find~$\aJourney_a$.
    It is possible to fix this issue by weakening the self-pruning rule, but this causes it to lose most of its runtime benefit~\cite{Bau23,Gro25}.
    Moreover, it would introduce additional complications in the proof of correctness for T-REX, which is already fairly technical.
    As the customization is already very fast compared to other preprocessing techniques, we opted to omit self-pruning for the sake of simplicity.
    
    \subparagraph*{Updates.}
    In a scenario with real-time updates to the timetable, the transfer ranks can be updated without re-running the entire customization.
    Clearly, only cells that are affected in some way by the update need to be processed.
    A simple approach would be to consider a stop event affected if its departure or arrival time has changed, and to consider a cell affected if it contains at least one affected stop event.
    This is wasteful because Event-TB does not inspect departure or arrival times, but rather determines reachability solely based on the transfer set and reached index pruning.
    Instead, we call a stop event~$\aTrip[i]$ \emph{directly affected} on level~$\aLevel$ if the set of incoming and outgoing transfers with rank at least~$\aLevel$ has changed.
    On level~$0$, this is the full set~$\transfers$, which may be changed by the update algorithm for the transfer generation (cf.~\cite{Wit21}).
    On higher levels, the update procedure for the cell containing~$\aTrip[i]$ on level~$\aLevel-1$ may have changed the rank of an incoming or outgoing transfer.
    We call~$\aTrip[i]$ \emph{indirectly affected} on level~$\aLevel$ if the preceding trip of~$\aTrip$ according to the total ordering~$\prec$ has changed or is directly affected.

    A cell~$\cell$ is processed if it contains at least one (directly or indirectly) affected stop event.
    A thorough approach to processing~$\cell$ is to set the rank of each transfer~$\aTransfer$ in~$\cell$ to~$\max\{\rank{\aTransfer},\aLevel\}$ and rerun the Event-TB search from all IBEs of~$\cell$.
    However, this does not take into account that updates usually only affect a short time window in practice.
    We therefore suggest an alternative approach, which does not reset the ranks in advance:
    Let~$\atime_\text{min}$ and~$\atime_\text{max}$ be the earliest and latest arrival time of an affected stop event in~$\cell$, respectively.
    For each combination of line~$\aLine$ and index~$i$, we process all IBEs~$\aTrip[i]$ with~$\aLine(\aTrip)=\aLine$ in descending order of departure time, ignoring those that depart after~$\atime_\text{max}$.
    After each Event-TB search, we check whether all found journeys arrive before~$\atime_\text{min}$.
    If so, then the earlier IBEs for~$\aLine$ and~$i$ are skipped.
    Because the ranks are not reset, this variant cannot decrement the rank of a transfer if it becomes irrelevant.
    Therefore, the search space of the query algorithm may gradually increase with successive updates.
    To limit this effect, one can switch to the thorough approach or even a full re-customization occasionally.

    \section{Proof of Subjourney Closure}
    \label{app:proof}
    We give a detailed proof of~\Cref{th:event-subjourney-closure}, which was omitted from~\Cref{ch:trex}:
    \subjourneyClosure*

    To prove this property, we formalize the search space of TB and the order in which it is explored.
    As a byproduct, we give (to our knowledge) the first complete proof of correctness for TB.
    Previously, Großmann \etal\cite{Gro25} gave a proof for the query phase but only in conjunction with a different transfer generation algorithm.

    We say that a stop event~$\aTrip[\ell]$ is \emph{visited} by TB in round~$n$ if there is a trip segment $\tripSegment{\aTrip}{i}{j} \in \aQueue_n$ with~$i < \ell \leq j$. In that case, journey unpacking by recursively following the parent pointer of~$\tripSegment{\aTrip}{i}{j}$ yields a~$\sourceStop$-$\aTrip[\ell]$-journey. We say that this is the journey \emph{found} by TB for~$\aTrip[\ell]$ with~$n$ trips.
    The following is immediate from the design of TB and its journey unpacking procedure:
    \begin{observation}
        \label{obs:prefix-found}
        Let~$\aJourney=\langle\tripSegment{\aTrip_1}{i_1}{j_1},\dots,\tripSegment{\aTrip_k}{i_k}{j_k}\rangle$ be a journey that is found by (Event-)TB.
        Then for each~$1 \leq n \leq k$ and~$i_n < \ell \leq k_n$, the stop event~$\aTrip_n[\ell]$ is visited and the corresponding found journey is $\langle\tripSegment{\aTrip_1}{i_1}{j_1},\dots,\tripSegment{\aTrip_n}{i_n}{\ell}\rangle$.
    \end{observation}

    \subsection{A Total Ordering of Feasible Journeys}
    For the following, we require the notion of a \emph{one-to-all query}~$\query=(\sourceStop,\departureTime)$ with a source stop~$\sourceStop$ and departure time~$\departureTime$, but no specified target stop.
    Here, the objective is to find the Pareto front and a representative set for every possible target stop in~$\stops$.
    Such a query can be answered by TB if target pruning and the maintenance of the tentative Pareto front~$\labels$ are disabled.
    Instead, to reconstruct the solution for a target stop~$\targetStop$, all visited stop events at~$\targetStop$ are examined, the corresponding found journeys are unpacked, and redundant representatives of the same cost vector are discarded.
    
    Given a one-to-all query~$\query=(\sourceStop,\departureTime)$, let~$\journeys(\query,\transfers)$ denote the set of all feasible journeys for~$\query$ that only use transfers in~$\transfers$.
    For simplicity, we also use the notation~$\journeys(\query,\transfers)$ if~$\query$ is a stop-to-stop query with a specified target stop~$\targetStop$; however, in this case, $\journeys(\query,\transfers)$ still contains journeys that end at all rechable stops, not just~$\targetStop$.
    We formalize the order in which TB finds journeys by defining a total ordering~$\prec_\query$.
    We can then show that the feasible journeys that are not found by TB are exactly those for which a prefix was discarded due to reached index pruning or target pruning.
    Then it follows that the representative found for a cost vector~$(\arrivalTime,n)$ is the $\prec_\query$-minimal one among the non-pruned representatives.

    The TB query algorithm includes three steps that explore parts of the network without specifying any particular order: 
    The outgoing footpaths of the source stop~$\sourceStop$ and the outgoing transfers of a scanned stop event are relaxed in an arbitrary order.
    Furthermore, after relaxing an initial footpath~$(\sourceStop,\aStop)$, including the case~$\sourceStop=\aStop$, the lines visiting~$\aStop$ are examined in an arbitrary order to find and enqueue the earliest reachable trip.
    To allow for formal reasoning about the algorithm, we need to decide on a fixed order for these steps.
    In a practical implementation, this is achieved by assigning integer IDs to all lines, footpaths and transfers.
    The elements are then examined in ascending order of their ID.
    To reflect this, we define the functions~$\lineID\colon\lines\to\{0,\dots,\absoluteVal{\lines}-1\}$, $\footpathID\colon\footpaths\to\{0,\dots,\absoluteVal{\footpaths}-1\}$ and~$\transferID\colon\transfers\to\{0,\dots,\absoluteVal{\transfers}-1\}$.
    We require that for every stop~$\stopA$, the outgoing footpaths~$(\stopA,\stopB)\in\footpaths$ are assigned consecutive IDs.
    Similarly, for every stop event~$\tripA[j]$, the outgoing transfers~$\transfer{\tripA[j]}{\tripB[i]}\in\transfers$ are assigned consecutive IDs.

    To describe the order in which TB finds journeys, we define a unique \emph{tiebreaking sequence} for each journey in $\journeys(\query,\transfers)$.
    Let~$\aJourney = \langle \tripSegment{\aTrip_1}{i_1}{j_1},\dots,\tripSegment{\aTrip_k}{i_k}{j_k}\rangle \in \journeys(\query,\transfers)$ be a journey.
    Then the tiebreaking sequence of~$\aJourney$ is given by $\tiebreakingSequence(\aJourney) = \langle k \rangle \circ \tiebreakingSequencePartial(\aJourney)$ with
    \begin{equation*}
        \tiebreakingSequencePartial(\aJourney) = \begin{cases}
            \langle \footpathID((\sourceStop,\aStop(\aTrip_1[i_1]))),\, \lineID(\aLine(\aTrip_1)),\, i_1,\, \departureTime(\aTrip_1[1]),\, j_1\rangle & \text{if } k=1,\\
            \langle \tiebreakingSequencePartial(\subjourney{1}{k-1}) \circ \langle \transferID(\transfer{\aTrip_{k-1}[j_{k-1}]}{\aTrip_k[i_k]}),\, j_k \rangle & \text{otherwise.}
        \end{cases}
    \end{equation*}

    For two journeys~$\aJourney,\aJourney'\in\journeys(\query,\transfers)$, we write~$\aJourney \prec_\query \aJourney'$ if~$\tiebreakingSequence(\aJourney)$ is smaller than~$\tiebreakingSequence(\aJourney')$ lexicographically.
    Note that because the tiebreaking sequences are unique for each journey, $\prec_\query$ is a total ordering of~$\journeys(\query,\transfers)$.

    In the following, we also consider event-based queries.
    Here, instead of a source stop~$\sourceStop$ and a departure time~$\departureTime$, we are given an initial stop event~$\sourceStopEvent$ and journeys are feasible if they start with~$\sourceStopEvent$.
    The tiebreaking sequence changes accordingly: we have~$\tiebreakingSequencePartial(\aJourney) = \langle j_1 \rangle$ in the case~$k=1$, but the other components remain the same.

    \begin{lemma}
        \label{lem:tb-order}
        For a query~$\query$, (Event-)TB finds journeys in increasing order of~$\prec_\query$.
    \end{lemma}
    \begin{proof}
        This follows directly by comparing the definition of the tiebreaking sequence with the algorithm description.
        For completeness, we argue this for each step of the algorithm explicitly.
        Journeys are found in increasing order with respect to the number of trips, which is encoded in the first entry of the tiebreaking sequence.
        Within each round~$n$, journeys are found in the order in which the corresponding final trip segments were added to~$\aQueue_n$.
        For each trip segment, the stop events are scanned in increasing order, which is encoded in the final entry of the tiebreaking sequence.
        For the order in which the trip segments are enqueued, we distinguish between~$n=1$ and $n > 1$.

        For $n=1$, we further distinguish between event-based and stop-based queries.
        For event-based queries, the only feasible trip segment that can be enqueued is~$\tripSegment{\aTrip_1}{i_1}{\absoluteVal{\aTrip}}$.
        For stop-based queries, the outgoing transfers of the source stop~$\sourceStop$ are explored in ascending order of~$\footpathID$, which is encoded in the first entry of~$\tiebreakingSequencePartial(\aJourney)$.
        At each stop~$\aStop$ reached via a footpath~$(\sourceStop,\aStop)$, the lines that visit~$\aStop$ are explored in ascending order of~$\lineID$, which is encoded in the second entry.
        The third entry is the index~$i_1$ at which the line is entered.
        This is required because a line~$\aLine$ may visit~$\aStop$ multiple times (\eg in ring lines).
        For example, let $i < j$ be two indices such that~$\aStop(\aTrip[i]) = \aStop(\aTrip[j]) = \aStop$ for every trip~$\aTrip$ of~$\aLine$.
        Then the index~$i$ is examined before the index~$j$.
        For each index~$i$, the algorithm identifies the earliest trip~$\aTrip$ of~$\aLine$ such that~$\aTrip[i]$ is reachable. This is encoded by the entry~$\departureTime(\aTrip[1])$. Due to the requirement that trips of the same line are totally ordered in terms of departure time at every stop, this ensures that the earliest trip is preferred such that the resulting journey is feasible.
        
        For~$n > 1$, consider two journeys~$\aJourney=\langle \tripSegment{\aTrip_1}{i_1}{j_1},\dots,\tripSegment{\aTrip_n}{i_n}{j_n}\rangle$ and $\aJourney'=\langle \tripSegment{\aTrip'_1}{i'_1}{j'_1},\dots,\allowbreak\tripSegment{\aTrip'_n}{i'_n}{j'_n}\rangle$ with~$\aTrip_n[i_n] \neq \aTrip'_n[i'_n]$.
        The order in which these two journeys are found is determined by the order in which the two trip segments~$\tripSegment{\aTrip_n}{i_n}{j_n}$ and~$\tripSegment{\aTrip'_n}{i'_n}{j'_n}$ are enqueued.
        If~$\aTrip_{n-1}[j_{n-1}] = \aTrip'_{n-1}[j'_{n-1}]$, \ie both segments originate from the same parent, then the order depends on the order in which the transfers~$\transfer{\aTrip_{n-1}[j_{n-1}]}{\aTrip_{n}[i_n]}$ and~$\transfer{\aTrip_{n-1}[j_{n-1}]}{\aTrip'_{n}[i'_n]}$ are relaxed, which is determined by~$\transferID$.
        Otherwise, the order depends on the order in which~$\aTrip_{n-1}[j_{n-1}]$ and~$\aTrip'_{n-1}[j'_{n-1}]$ were scanned in round~$n-1$.
        By induction, this is equivalent to the relative order of the prefixes~$\subjourney{1}{n-1}$ and~$\subjourneyAlt{1}{n-1}$.
    \end{proof}

    \subsection{Stability of the Total Ordering}
    To prove subjourney closure, we need to examine how the relative order of two journeys is affected when we add or remove prefixes or suffixes.
    The following is immediately apparent from the recursive construction of the tiebreaking sequence:
    
    \begin{observation}
        \label{obs:prefix-order}
        For two $\sourceStop$-$\targetStop$-journeys~$\aJourney \neq \aJourney' \in \journeys(\query,\transfers)$ with~$k$ trips, let~$n$ be an index with~$1 \leq n \leq k$ such that~$\subjourney{1}{n} \neq \subjourneyAlt{1}{n}$.
        Then~$\aJourney \prec_\query \aJourney'$ iff~$\subjourney{1}{n} \prec_\query \subjourneyAlt{1}{n}$.
    \end{observation}

    In particular, this means that the relative order of two journeys does not change when we add the same proper suffix to both.
    We say that the ordering is \emph{stable} under the addition of proper suffixes.
    For prefixes, the situation is slightly more complicated.
    If the added prefix is a proper journey, then the initial footpath becomes an intermediate transfer.
    Because footpaths and transfers are handled with different tiebreaking rules, it cannot be guaranteed that the relative order stays the same.
    However, we can show that the ordering is stable under the addition of prefixes that end in the middle of a trip.
    For this, we introduce notation for the concatenation of two journeys.

    \begin{definition}[Journey Concatenation]
    Let~$\aJourney_1 = \langle \tripSegment{\aTrip_a}{i_a}{j_a}, \dots, \tripSegment{\aTrip_b}{i_b}{x}\rangle$ and~$\aJourney_2 =  \langle \tripSegment{\aTrip_b}{x}{j_b},\allowbreak \dots, \tripSegment{\aTrip_c}{i_c}{j_c}\rangle$ be two journeys that coincide in the event~$\aTrip_b[x]$.
    We define the concatenated journey~$\aJourney_1 \circ \aJourney_2 =  \langle \tripSegment{\aTrip_a}{i_a}{j_a}, \dots, \tripSegment{\aTrip_b}{i_b}{j_b}, \dots \tripSegment{\aTrip_c}{i_c}{j_c}\rangle$.
    Note that this journey uses~$\absoluteVal{\aJourney_1} + \absoluteVal{\aJourney_2} - 1$ trips.
    We make the following observation:
    If~$\aJourney_1 \in \journeys(\query,\transfers)$ for some query~$\query$ and~$\aJourney_2 \in \journeys(\query',\transfers)$ for the event-based query~$\query'=(\aTrip_b[x])$, then it follows that~$\aJourney_1 \circ \aJourney_2 \in \journeys(\query,\transfers)$ because all of its transfers are in~$\aJourney_1$ or~$\aJourney_2$.
    \end{definition}
        
    \begin{lemma}
        \label{lem:stability-prefix}
        Let~$\aJourney \in \journeys(\query,\transfers)$ be a journey for a query~$\query$, let~$\aJourney_s$ be a suffix of~$\aJourney$ that starts at some stop event~$\stopEvent$, and let~$\aJourney_p$ be the corresponding prefix such that~$\aJourney = \aJourney_p \circ \aJourney_s$.
        Let~$\aJourney'_s \in \journeys(\query',\transfers)$ for the event-based query~$\query'=(\stopEvent)$, and let~$\aJourney' = \aJourney_p \circ \aJourney'_s$.
        Then we have~$\aJourney \prec_\query \aJourney'$ iff~$\aJourney_s \prec_{\query'} \aJourney'_s$.
    \end{lemma}
    \begin{proof}
        Let~$\tiebreakingSequencePartial_p$ denote the sequence that consists of all but the final entry of~$\tiebreakingSequencePartial(\aJourney_p)$, which is the exit index of the final trip segment.
        Then we have~$\tiebreakingSequence(\aJourney)=\langle\absoluteVal{\aJourney_p} + \absoluteVal{\aJourney_s} - 1\rangle \circ \tiebreakingSequencePartial_p \circ \tiebreakingSequencePartial(\aJourney_s)$ and~$\tiebreakingSequence(\aJourney')=\langle\absoluteVal{\aJourney_p} + \absoluteVal{\aJourney'_s} - 1\rangle \circ \tiebreakingSequencePartial_p \circ \tiebreakingSequencePartial(\aJourney'_s)$.
        Compared to the tiebreaking sequences of~$\aJourney_s$ and~$\aJourney'_s$, the first entries are shifted by the same value~$\absoluteVal{\aJourney_p} - 1$ and the shared infix~$\tiebreakingSequencePartial_p$ is inserted afterwards.
        Neither change has an effect on the relative order of the sequences.
    \end{proof}

    \subsection{Correctness of TB}
    In this section, we precisely characterize the representative that TB finds for each cost vector in the Pareto front: it is the~$\prec_\query$-minimal one.
    We can then use this characterization, together with the stability properties established for~$\prec_\query$, to prove the subjourney closure property.

    We begin by characterizing the journeys that are discarded by TB due to reached index pruning.
    Let~$\aJourney$ be a journey ending with a stop event~$\aTrip[i]$. We say that a journey~$\aJourney'$ \emph{event-dominates} $\aJourney$ if~$\aJourney' \prec_\query \aJourney$ and~$\aJourney'$ ends in some~$\aTrip'[i]$ with~$\aTrip' \preceq \aTrip$.
    A journey~$\aJourney$ is \emph{pruned} if a prefix of~$\aJourney$ is event-dominated by a non-pruned journey.
    We say that~$\aJourney$ is \emph{stop-dominated} by~$\aJourney'$ if they end at the same stop, $\aJourney'$ does not arrive later and~$\aJourney' \prec_\query \aJourney$.
    If~$\aJourney$ is event-dominated by $\aJourney'$, it is also stop-dominated (but not necessarily vice versa).

    \begin{observation}
        \label{obs:tb-non-pruned}
        For a query~$\query$, (Event-)TB without target pruning finds exactly the non-pruned journeys in~$\journeys(\query,\transfers)$.
    \end{observation}

    This allows us to show that TB finds~$\prec_\query$-minimal journeys among the non-pruned ones.

    \begin{lemma}
        \label{lem:tb-journey}
        For a query~$\query$, the representative returned by (Event-)TB for the cost vector~$(\arrivalTime,k)$ is $\prec_\query$-minimal among all non-pruned representatives in~$\journeys(\query,\transfers)$ that are not stop-dominated by a non-pruned journey.
    \end{lemma}
    \begin{proof}
        We show the claim for a stop-based query~$\query=(\sourceStop,\targetStop,\departureTime)$; the proof for event-based queries is analogous.
        The returned journey~$\aJourney$ is the first found $\sourceStop$-$\targetStop$-journey with arrival time $\arrivalTime$, and there are no found $\sourceStop$-$\targetStop$-journeys with an earlier arrival time in round~$k$.
        Assume momentarily that target pruning is disabled.
        Then it follows from~\Cref{lem:tb-order} and~\Cref{obs:tb-non-pruned} that~$\aJourney$ is $\prec_\query$-minimal among all non-pruned representatives for~$(\arrivalTime,k)$ in~$\journeys(\query,\transfers)$.
        For every other non-pruned representative~$\aJourney'$, we either have~$\arrivalTime(\aJourney')>\arrivalTime$, or~$\arrivalTime(\aJourney')=\arrivalTime$ and~$\aJourney \prec_\query \aJourney'$.
        Hence, $\aJourney'$ is stop-dominated by the non-pruned journey~$\aJourney$.

        Now consider a representative~$\aJourney'$ that is not found due to target pruning.
        Then we have~$\arrivalTime(\aJourney')\geq\minTime$ at the moment of pruning.
        If~$\aJourney$ has not been found yet, then we have~$\arrivalTime(\aJourney')\geq\minTime>\arrivalTime$.
        Otherwise, we have~$\arrivalTime(\aJourney')\geq\arrivalTime$ and~$\aJourney \prec_\query \aJourney'$.
        In both cases, $\aJourney'$ is stop-dominated by $\aJourney$.
    \end{proof}
    
    To show that the returned representative is in fact $\prec_\query$-minimal among \emph{all} representatives, including the pruned ones, we need to analyze how the pruning rules in the transfer generation step and the query algorithm interact with each other.
        
    \begin{lemma}[Transfer Generation Correctness]
    \label{lem:transfer-generation}
    Let~$\aJourney = \langle \tripSegment{\aTrip_1}{i_1}{j_1},\dots,\tripSegment{\aTrip_k}{i_k}{j_k}\rangle$ be a journey such that~$\aTransfer_x=\transfer{\aTrip_x[j_x]}{\aTrip_{x+1}[i_{x+1}]}$ is the first transfer in~$\aJourney$ that is not in~$\transfers$.
    Then there is a journey~$\aJourney'=\langle \tripSegment{\aTrip_1}{i_1}{j_1},\dots,\tripSegment{\aTrip_x}{i_x}{j'_x},\tripSegment{\aTrip'_{x+1}}{i'_{x+1}}{j'_{x+1}},\dots,\tripSegment{\aTrip'_\ell}{i'_\ell}{j'_\ell}\rangle \in \journeys(\query,\transfers)$ with~$\ell \leq k$ that arrives at~$\targetStop$ not later than~$\aJourney$.
    \end{lemma}
    \begin{proof}
        We first show that each step of the transfer generation algorithm upholds the claim for~$k=2$ as an invariant.
        Let~$\aJourney=\langle\tripSegment{\aTrip_1}{i_1}{j_1},\tripSegment{\aTrip_2}{i_2}{j_2}\rangle$ be a journey such that~$\aTransfer=\transfer{\aTrip_1[j_1]}{\aTrip_2[j_2]}\notin\transfers$:
        Then there is a journey~$\aJourney'\in \journeys(\query,\transfers)$ with~$\aJourney'=\langle\tripSegment{\aTrip_1}{i_1}{j'_1},\tripSegment{\aTrip'_2}{i'_2}{j'_2}\rangle$ or~$\aJourney'=\langle\tripSegment{\aTrip_1}{i_1}{j'_1}\rangle$ that arrives at~$\targetStop$ not later than~$\aJourney$.
    
        If~$\aTransfer$ is not generated, this is for one of two reasons:
        (1) Another transfer~$\transfer{\aTrip_1[j'_1]}{\aTrip'_2[i'_2]}$ with~$j'_1 \geq j_1$, $i'_2 \leq i_2$ and~$\aTrip_2 \succeq \aTrip'_2$ is generated.
        Then the claim holds with~$\aJourney'=\langle\tripSegment{\aTrip_1}{i_1}{j'_1},\tripSegment{\aTrip'_2}{i'_2}{j_2}\rangle$.
        (2) We have~$\aTrip_1 \preceq \aTrip_2$ and~$j_1 \leq i_2$.
        Then the claim holds with~$\aJourney'=\langle\tripSegment{\aTrip_1}{i_1}{j_2}\rangle$.
        If~$\aTransfer$ is removed due to the U-turn rule, then we have~$\aStop(\aTrip_1[j_1-1])=\aStop(\aTrip_2[i_2+1])$ and~$\transfer{\aTrip_1[j_1-1]}{\aTrip_2[i_2+1]}\in\transfers$ and the claim holds with~$\aJourney'=\langle\tripSegment{\aTrip_1}{i_1}{j_1-1},\tripSegment{\aTrip_2}{i_2+1}{j_2}\rangle$.
        For the latest-exit rule, it is easy to verify that it upholds the invariant by design.
    
        We now prove the claim for general~$k$ by induction over~$x$.
        The first~$x-1$ intermediate transfers in~$\aJourney$ are in~$\transfers$.
        We show that there is a journey~$\aJourney'$ with~$\ell < k$ trips that arrives at~$\targetStop$ not later than~$\aJourney$ and for which the first~$x$ intermediate transfers are in~$\transfers$.
        Consider the subjourney~$\aJourney_s=\langle\tripSegment{\aTrip_x}{i_x}{j_x},\tripSegment{\aTrip_{x+1}}{i_{x+1}}{j_{x+1}}\rangle$.
        By the $k=2$ case, there is a journey~$\aJourney'_s\in\journeys(\query,\transfers)$ with at most two trips that arrives at~$\targetStop$ not later than~$\aJourney_s$.
        We consider the case that~$\aJourney'_s$ uses exactly two trips, i.e., $\aJourney_s'=\langle\tripSegment{\aTrip_x}{i_x}{j'_x},\tripSegment{\aTrip'_{x+1}}{i'_{x+1}}{j'_{x+1}}\rangle$; the other case is analogous.
        Consider the journey
        \[\aJourney'=\langle \tripSegment{\aTrip_1}{i_1}{j_1},\dots,\tripSegment{\aTrip_x}{i_x}{j'_x},\tripSegment{\aTrip'_{x+1}}{i'_{x+1}}{j'_{x+1}},\tripSegment{\aTrip_{x+2}}{i_{x+2}}{j_{x+2}}\dots,\tripSegment{\aTrip_\ell}{i_\ell}{j_\ell}\rangle\]
        that replaces~$\aJourney_s$ with~$\aJourney'_s$ in~$\aJourney$.
        This journey uses the same transfers as~$\aJourney$ with two exceptions: the transfer~$\aTransfer_x$ is replaced by~$\aTransfer'_x=\transfer{\aTrip_x[j'_x]}{\aTrip'_{x+1}[i'_{x+1}]}$ and~$\aTransfer_{x+1}=\transfer{\aTrip_{x+1}[j_{x+1}]}{\aTrip_{x+2}[i_{x+2}]}$ is replaced by~$\aTransfer'_{x+1}=\transfer{\aTrip'_{x+1}[j'_{x+1}]}{\aTrip_{x+2}[i_{x+2}]}$.
        We observe that~$\aTransfer'_x \in \transfers$ and that~$\aTransfer'_{x+1}$ is feasible because~$\aTransfer_{x+1}$ is feasible and we have~$\aStop(\aTrip_{x+1}[j_{x+1}])=\aStop(\aTrip'_{x+1}[j'_{x+1}])$ and~$\arrivalTime(\aTrip'_{x+1}[j'_{x+1}])\leq\arrivalTime(\aTrip_{x+1}[j_{x+1}])$.
        Hence, $\aJourney'$ is feasible, arrives at~$\targetStop$ not later than~$\aJourney$, and the first~$x$ intermediate transfers are in~$\transfers$.
    \end{proof}

    It is not immediately clear that the latest-exit rule is compatible with reached index pruning during the query:
    The idea behind reached index pruning is that if stop events~$\tripA[j]$ and~$\tripA'[j]$ with~$\tripA \prec \tripA'$ are both reachable, then it is always preferable to use~$\tripA[j]$.
    However, consider a journey that uses a transfer~$\transfer{\tripA[j]}{\tripB[i]}$.
    The latest-exit rule may discard this transfer if there is a dominating alternative for every possible target stop, but it may retain the transfer~$\transfer{\tripA'[j]}{\tripB[i]}$.
    This is only correct if the dominating alternative is itself not pruned and can therefore be found by TB; see~\Cref{fig:conflict}.
    We show that this is the case.

    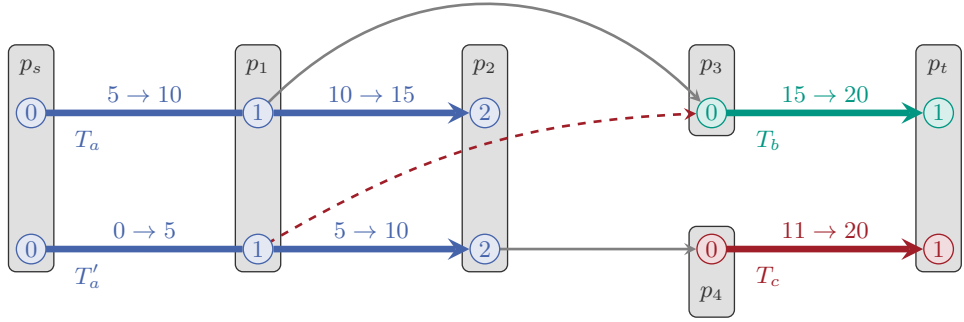
\begin{figure}[t]
        \centering
        \begin{tikzpicture}
	\begin{scope}[yscale=1.2]
		\node (ta0) at (0.00, 0.00) {};%
		\node (ta1) at (3.00, 0.00) {};%
        \node (ta2) at (6.00, 0.00) {};%
        \node (tb0) at (0.00, -1.50) {};%
        \node (tb1) at (3.00, -1.50) {};%
        \node (tb2) at (6.00, -1.50) {};%
		\node (tc0) at  (9.00, 0.00) {};%
		\node (tc1) at  (12.00, 0.00) {};%
        \node (td0) at  (9.00, -1.50) {};%
        \node (td2) at  (12.00, -1.50) {};%
		\node (source_label) at (0.00, 0.50) {};%
		\node (p1) at (3.00, 0.50) {};%
		\node (p2) at (6.00, 0.50) {};%
		\node (p3) at (9.00, 0.50) {};%
		\node (p4) at (9.00, -2.00) {};%
		\node (target_label) at (12.00, 0.50) {};%

        \begin{pgfonlayer}{background}
		\node [fit=(ta0)(tb0)(source_label),stop] {};%
		\node [fit=(ta1)(tb1)(p1),stop] {};%
        \node [fit=(ta2)(tb2)(p2),stop] {};%
        \node [fit=(tc0)(p3),stop] {};%
        \node [fit=(td0)(p4),stop] {};%
		\node [fit=(tc1)(td2)(target_label),stop] {};%
		\end{pgfonlayer}
		
		\node [align=left,text=KITblue] at ( 0.75, -0.30) {\small{$\tripA$}};%
		\node [align=left,text=KITblue] at ( 0.75, -1.80) {\small{$\tripA'$}};%
		\node [align=left,text=KITgreen] at ( 9.75, -0.30) {\small{$\tripB$}};%
		\node [align=left,text=KITred] at ( 9.75, -1.80) {\small{$\tripC$}};%
        
        \node [align=left,text=KITblue] at (1.50, 0.20) {\small{$5\rightarrow10$}};%
		\node [align=left,text=KITblue] at (1.50, -1.30) {\small{$0\rightarrow5$}};%
        \node [align=left,text=KITblue] at (4.50, 0.20) {\small{$10\rightarrow15$}};%
		\node [align=left,text=KITblue] at (4.50, -1.30) {\small{$5\rightarrow10$}};%
        \node [align=left,text=KITgreen] at (10.50, 0.20) {\small{$15\rightarrow20$}};%
        \node [align=left,text=KITred] at (10.50, -1.30) {\small{$11\rightarrow20$}};%
        
		\node (ta0_v) at (ta0) [vertex,draw=KITblue,fill=KITblue!15] {\gs};%
		\node (tb0_v) at (tb0) [vertex,draw=KITblue,fill=KITblue!15] {\gs};%
		\node (ta1_v) at (ta1) [vertex,draw=KITblue,fill=KITblue!15] {\gs};%
		\node (tb1_v) at (tb1) [vertex,draw=KITblue,fill=KITblue!15] {\gs};%
        \node (ta2_v) at (ta2) [vertex,draw=KITblue,fill=KITblue!15] {\gs};%
		\node (tb2_v) at (tb2) [vertex,draw=KITblue,fill=KITblue!15] {\gs};%
		\node (tc0_v) at (tc0)  [vertex,draw=KITgreen,fill=KITgreen!15] {\gs};%
		\node (tc1_v) at (tc1)  [vertex,draw=KITgreen,fill=KITgreen!15] {\gs};%
        \node (td0_v) at (td0)  [vertex,draw=KITred,fill=KITred!15] {\gs};%
        \node (td2_v) at (td2)  [vertex,draw=KITred,fill=KITred!15] {\gs};%
        
        \begin{pgfonlayer}{background}
        \draw [KITblue, route] (ta0) -- (ta1_v) -- (ta2_v);
        \draw [KITblue, route] (tb0) -- (tb1_v) -- (tb2_v);
        \draw [KITgreen, route] (tc0) -- (tc1_v);
        \draw [KITred, route] (td0) -- (td2_v);
                
        \draw [directedEdge]  (ta1) to[bend left=40] (tc0_v);
        \draw [directedEdge]  (tb2) -- (td0_v);
        \draw [directedEdge,dashed,KITred]  (tb1) to[bend left=12] (tc0_v);
        \end{pgfonlayer}
		
		\node at (ta0) [text=KITblue] {\small{$0$}};%
		\node at (tb0) [text=KITblue] {\small{$0$}};%
		\node at (ta1) [text=KITblue] {\small{$1$}};%
		\node at (tb1) [text=KITblue] {\small{$1$}};%
        \node at (ta2) [text=KITblue] {\small{$2$}};%
		\node at (tb2) [text=KITblue] {\small{$2$}};%
		\node at (tc0)  [text=KITgreen] {\small{$0$}};%
		\node at (tc1)  [text=KITgreen] {\small{$1$}};%
        \node at (td0)  [text=KITred] {\small{$0$}};%
        \node at (td2)  [text=KITred] {\small{$1$}};%
		\node at (source_label) [text=nodeColor!100] {\small{$\sourceStop$}};%
        \node at (p1) [text=nodeColor!100] {\small{$\aStop_1$}};%
		\node at (p2) [text=nodeColor!100] {\small{$\aStop_2$}};%
		\node at (p3) [text=nodeColor!100] {\small{$\aStop_3$}};%
		\node at (p4) [text=nodeColor!100] {\small{$\aStop_4$}};%
		\node at (target_label) [text=nodeColor!100] {\small{$\targetStop$}};%
	\end{scope}
\end{tikzpicture}
        \caption{
            An example showing the potential conflict between the latest-exit rule and reached index pruning.
            Gray boxes represent stops.
            Nodes within the boxes represent stop events and are labeled with their indices along the respective trip.
            Colored edges represent trips, which are labeled with the departure and arrival times of the respective stop events.
            Gray edges represent transfers.
            The red dashed edge represents the possible transfer~$\transfer{\tripA'[1]}{\tripB[0]}$, which is discarded due to latest-exit pruning.
            The reason for this is the transfer~$\transfer{\tripA'[2]}{\tripC[0]}$, which allows the stop~$\targetStop$ to be reached with the same arrival time.
            On the other hand, the transfer~$\transfer{\tripA[2]}{\tripC[0]}$ is not feasible, and there is no footpath from $\aStop_1$ to~$\aStop_4$ or from~$\aStop_2$ to~$\aStop_3$.
            During a query from~$\sourceStop$ to~$\targetStop$ with departure time~$0$, reached index pruning ensures that~$\tripA'$ is scanned but~$\tripA$ is not.
            Hence, the transfer~$\transfer{\tripA[1]}{\tripB[0]}$ is not relaxed.
            To resolve the conflict, it must be ensured that the query algorithm relaxes the transfer~$\transfer{\tripA'[2]}{\tripC[0]}$.
        }
        \label{fig:conflict}
    \end{figure}
    
    \begin{lemma}[Reached Index Pruning Correctness]
    \label{lem:line-pruning}
    For each pruned optimal journey~$\aJourney \in \journeys(\query,\transfers)$, there exists a non-pruned journey~$\aJourney' \in \journeys(\query,\transfers)$ that stop-dominates~$\aJourney$.
    \end{lemma}
    \begin{proof}
        We show the existence of a journey~$\aJourney' \in \journeys(\query,\transfers)$ that stop-dominates~$\aJourney$.
        If~$\aJourney'$ is also pruned, then we can apply the same argument iteratively because stop domination is transitive.
        Because~$\aJourney' \prec_\query \aJourney$ and because the number of journeys in~$\journeys(\query,\transfers)$ with at most~$|\aJourney|$ trips is finite, it follows that there is a choice of~$\aJourney'$ such that~$\aJourney'$ is non-pruned.
    
        Because~$\aJourney$ is pruned, there is a prefix~$\aJourney_p$ of~$\aJourney$ that is event-dominated by a non-pruned journey~$\aJourney'_p \in \journeys(\query,\transfers)$.
        Let~$\tripSegment{\aTrip_x}{i_x}{j_x}$ denote the final trip segment of~$\aJourney_p$.
        Then the final stop event of~$\aJourney'_p$ is~$\aTrip'_x[j_x]$ with~$\aTrip'_x \preceq \aTrip_x$.
        Because the transfer~$\aTransfer_x = \transfer{\aTrip_x[j_x]}{\aTrip_{x+1}[j_{x+1}]}$ is feasible, the transfer~$\aTransfer'_x = \transfer{\aTrip'_x[j_x]}{\aTrip_{x+1}[j_{x+1}]}$ is feasible as well.
        Hence, if we replace~$\aJourney_p$ with~$\aJourney'_p$ in~$\aJourney$, we obtain a feasible journey~$\aJourney_c$ with~$|\aJourney|$ trips that arrives at~$\targetStop$ not later than~$\aJourney$.
        We distinguish between two cases:
        \begin{itemize}
            \item If $\aTransfer'_x \in \transfers$, then we have~$\aJourney_c \in \journeys(\query,\transfers)$.
            Furthermore, we have~$\aJourney_p = \aJourney[1,x]$ and~$\aJourney'_p = \aJourney'[1,x]$, so it follows from~\Cref{obs:prefix-order} and~$\aJourney'_p \prec_\query \aJourney_p$ that~$\aJourney_c \prec_\query \aJourney$.
            Hence, $\aJourney'=\aJourney_c$ is the desired journey.
            \item If $\aTransfer'_x \notin \transfers$, then we have~$\aTransfer'_x \neq \aTransfer_x$ and therefore~$\aTrip'_x \prec \aTrip_x$.
            By~\Cref{lem:transfer-generation}, there is a journey~$\aJourney_r \in \journeys(\query,\transfers)$ with at most~$|\aJourney|$ trips that is identical to~$\aJourney'_p$ up to~$\aTrip'[i_x]$ and arrives at~$\targetStop$ not later than~$\aJourney_c$ and therefore also not later than~$\aJourney'$.
            We show that~$\aJourney_r \prec_\query \aJourney$ and therefore~$\aJourney'=\aJourney_r$ is the desired journey.
            Note that~$\aJourney_r[1,x-1] = \aJourney'_p[1,x-1]$ and~$\aJourney[1,x-1] = \aJourney_p[1,x-1]$.
            Thus, if~$\aJourney'_p[1,x-1] \neq \aJourney_p[1,x-1]$, then it follows from~\Cref{obs:prefix-order} and~$\aJourney'_p \prec_\query \aJourney_p$ that~$\aJourney_c \prec_\query \aJourney$.
            Otherwise, the relative order of the journeys depends only on trip segment~$x$, which is~$\tripSegment{\aTrip'_x}{i_x}{j'_x}$ for~$\aJourney_r$, $\tripSegment{\aTrip'_x}{i_x}{j_x}$ for~$\aJourney'_p$, and~$\tripSegment{\aTrip_x}{i_x}{j_x}$ for~$\aJourney$ and~$\aJourney_p$.
            It follows from~$\aJourney'_p \prec_\query \aJourney_p$ that trip segments starting with~$\aTrip'_x[i_x]$ come before those starting with~$\aTrip_x[i_x]$ in the order, so we have~$\aJourney_r \prec_\query \aJourney$.\qedhere
        \end{itemize}
    \end{proof}
    
    \begin{theorem}[TB Correctness]
        \label{th:tb-correctness}
        For a query~$\query$, the representative returned by (Event-)TB for the cost vector~$(\arrivalTime,k)$ is the $\prec_\query$-minimal representative.
    \end{theorem}
    \begin{proof}
        Let~$\aJourney$ be the returned representative.
        By~\Cref{lem:tb-journey}, it suffices to show that no pruned journey stop-dominates~$\aJourney$.
        If a pruned journey~$\aJourney'$ stop-dominates~$\aJourney$, then by~\Cref{lem:line-pruning}, there is a non-pruned journey~$\aJourney_u$ that stop-dominates~$\aJourney'$.
        Because stop domination is transitive, $\aJourney_u$ also stop-dominates~$\aJourney$, which contradicts~\Cref{lem:tb-journey}.
    \end{proof}

    \subsection{Subjourney Closure Property}
    Finally, we put the pieces together and prove the subjourney closure property.

    \subjourneyClosure*
    \begin{proof}
        Let~$\query'=(\stopEvent,\targetStop)$ be the event-to-stop query and let~$\aJourney_s$ be the $\stopEvent$-$\targetStop$-subjourney of~$\aJourney$.
        It follows from the optimality of~$\aJourney$ that~$\aJourney_s$ is optimal.
        We show that~$\aJourney_s$ is the~$\prec_{\query'}$-minimal representative for its cost vector.
        Then it follows by~\Cref{th:tb-correctness} that Event-TB returns it and it follows by~\Cref{obs:prefix-found} that the prefix~$\eventSubjourney{\sourceStopEvent}{\targetStopEvent}$ is found.
        Let~$\aJourney'_s$ be another~$\stopEvent$-$\targetStop$-journey.
        Let~$\aJourney_p$ be the journey such that~$\aJourney = \aJourney_p \circ \aJourney_s$ and let~$\aJourney' = \aJourney_p \circ \aJourney'_s$.
        By~\Cref{th:tb-correctness}, we have~$\aJourney \prec_\query \aJourney'$.
        Then it follows by~\Cref{lem:stability-prefix} that~$\aJourney_s \prec_{\query'} \aJourney'_s$.
    \end{proof}

    \newpage
    \section{Omitted Experiments}
    \label{app:experiments}

    This section presents additional experiments and details on the experimental setup that were omitted from~\Cref{ch:experiments}.

    \subsection{Implementation Details}
    \label{app:experiments:implementation}
    Our TB implementation closely follows that of Witt~\cite{Wit15} but omits some low-level performance optimizations.
    In particular, we do not use SIMD intrinsics for the sake of portability.
    Furthermore, we ensure that all data types use enough bits to store any realistic network.
    For most network elements (\eg stops, stop events), we use $\numprint{32}$-bit IDs.
    Because all of our datasets have fewer than $2^{24}$ trips and fewer than $\numprint{256}$ stops per line, we assign a $\numprint{32}$-bit ID to each IBE, with the stop index in the lower 8 bits and the trip ID in the remaining $\numprint{24}$ bits.
	Cell IDs use~$\numprint{16}$ bits, and the rank of a transfer uses~$\numprint{8}$ bits.
    To represent arrival and departure times, we use $\numprint{32}$ bits, which are sufficient to represent a service period of two days at a time resolution of seconds.

    Witt's implementation makes stronger assumptions about the network's size in order to pack bits as tightly as possible.
    As a consequence, it cannot load our Europe instance.
    Moreover, it uses a coarser time resolution in which one time unit represents~$5$ seconds, which allows it to use only $\numprint{16}$ bits to represent arrival and departure times.
    As a result of these differences, Witt's implementation is faster than ours by a factor of~$\numprint{2}$--$\numprint{3}$.
    We note that similar low-level optimizations could be applied to T-REX as well, which would likely benefit to a similar degree.

    \subsection{Visualization}
    \begin{figure}[ht]
		\centering
        \subfloat[Europe\label{fig:europe}]{\includegraphics[width=0.35\textwidth]{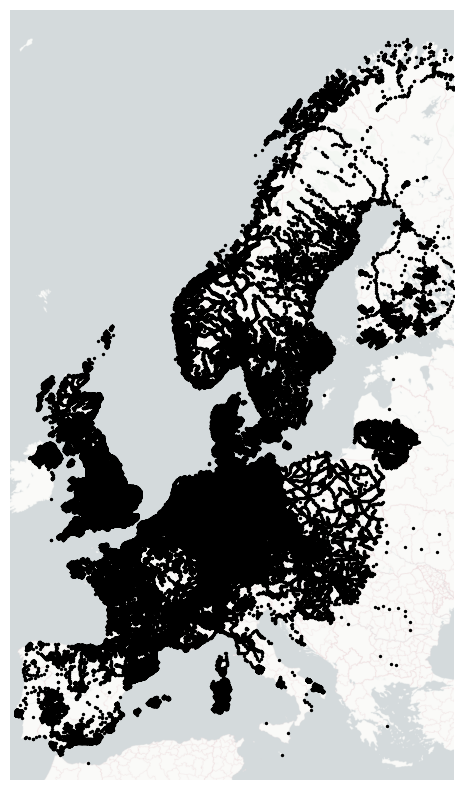}}%
		\hfill
        \subfloat[Germany\label{fig:germany}]{\includegraphics[width=0.35\textwidth]{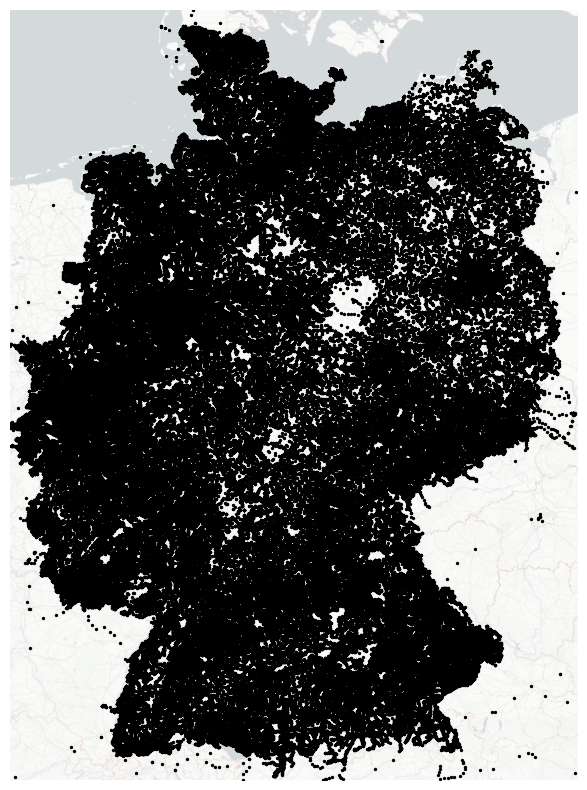}}%
		\hfill
        \subfloat[Great Britain\label{fig:gb}]{\includegraphics[width=0.35\textwidth]{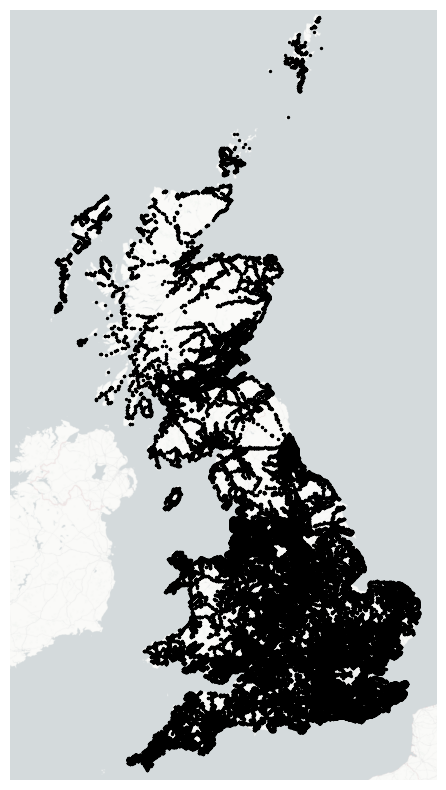}}%
		\hfill
        \subfloat[Switzerland\label{fig:swiss}]{\includegraphics[width=0.35\textwidth]{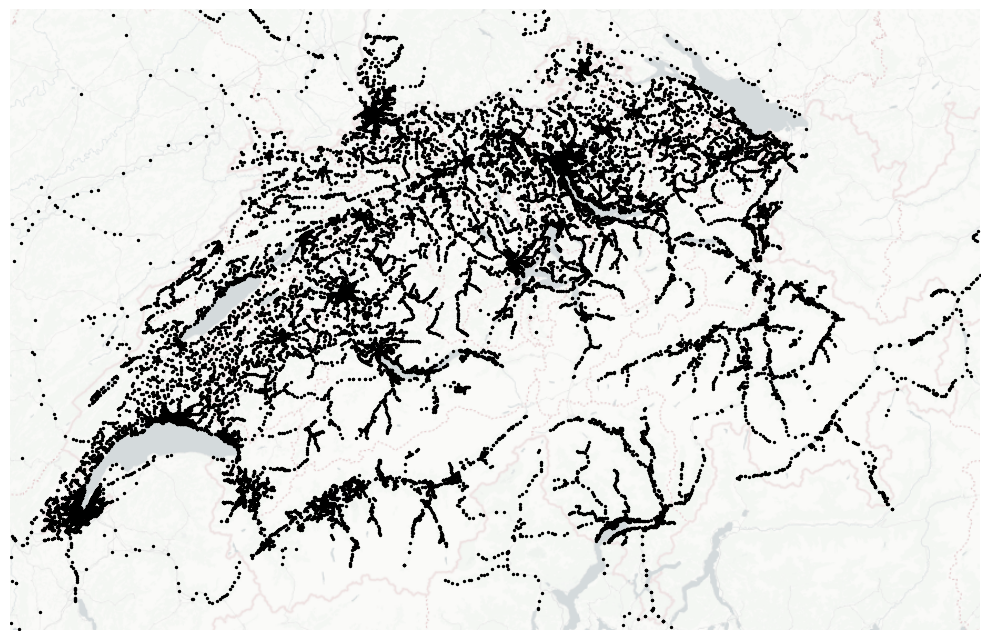}}%
		\caption{Visualization of the large datasets: Europe, Germany, Great Britain and Switzerland. For visualization purposes only, each network is shown restricted to a bounding box that centers on the core service area and excludes distant outliers. Basemap tiles by CARTO; data \textcopyright OpenStreetMap contributors.}
		\label{fig:countries_continents}
	\end{figure}

    \begin{figure}[ht]
		\centering
        \subfloat[Paris\label{fig:paris}]{\includegraphics[width=0.48\textwidth]{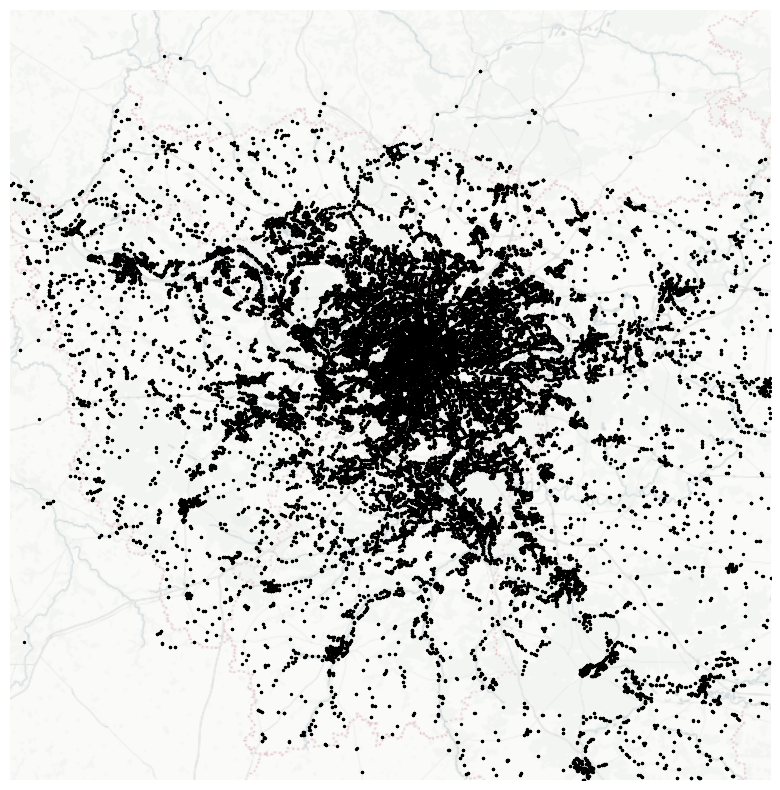}}%
		\hfill
        \subfloat[Berlin\label{fig:berlin}]{\includegraphics[width=0.48\textwidth]{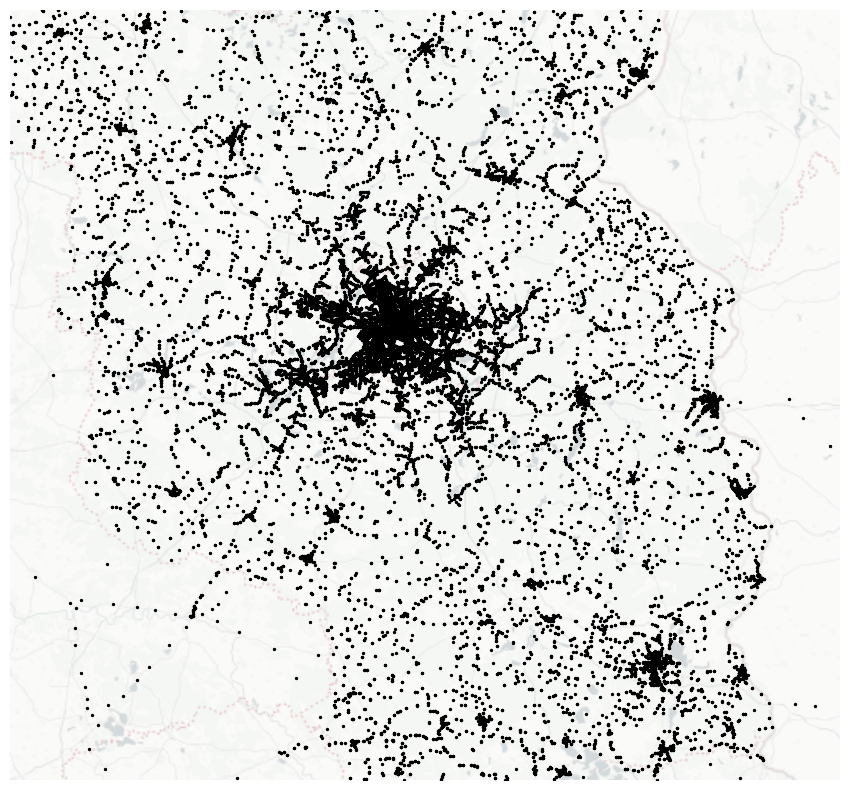}}%
		\caption{Visualization of the metropolitan datasets: Paris and Berlin (\cf Figure~\ref{fig:countries_continents}). Basemap tiles by CARTO; data \textcopyright OpenStreetMap contributors.}
		\label{fig:cities}
	\end{figure}
	
	The stops of each dataset are visualized in Figures~\ref{fig:countries_continents} and \ref{fig:cities}.
    Figure~\ref{fig:searchspace} compares the search space of TB and T-REX for an exemplary query on the Europe network.
    The search space of TB is spread out fairly evenly and includes much of the local transport in the visited metropolitan regions.
    By contrast, T-REX has a much sparser search space and mostly considers long-distance trains.
    The exceptions are the regions near the source and target, which is expected, and a dense cluster located roughly in southwestern Germany.
    As shown in~\Cref{fig:partition.plots}, this cluster is close to the cut on the second-highest level in the multilevel partition.
    Near a cut, transfers from and to regional transport are more likely to be required to traverse the cell and therefore tend to have higher ranks.
    This explains why the search space of T-REX is concentrated in this area.
    
    \begin{figure}[ht]
		\centering
        \includegraphics[width=0.65\textwidth]{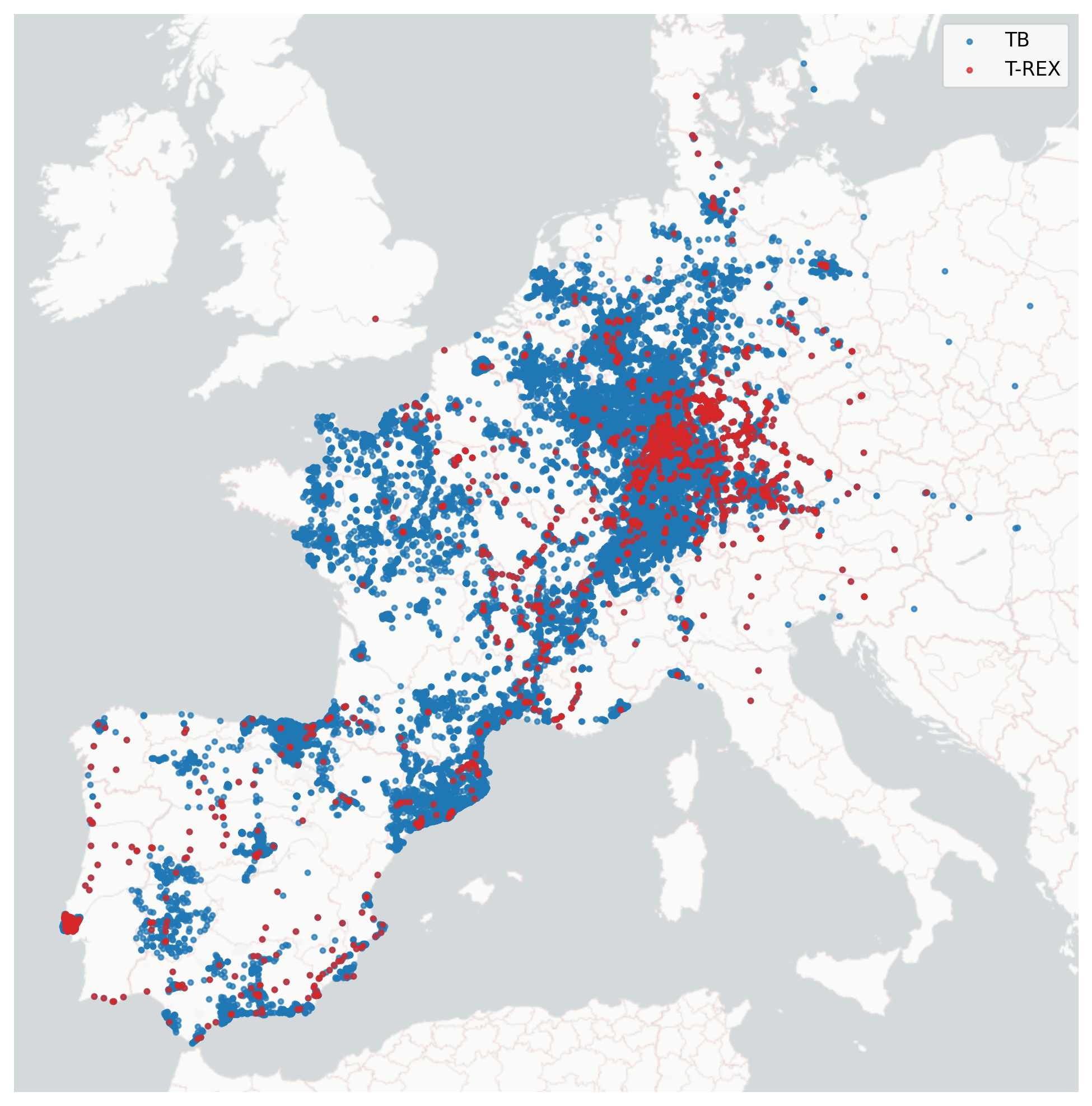}%
		\caption{Search space of TB and T-REX for a query from the ``Lapa'' bus station in Lisbon to the subway station ``Eberswalder Straße'' in Berlin with a departure time of \printMinuteSeconds{8}{0}. Shown are the stops from which outgoing transfers are relaxed during the query. Basemap tiles by CARTO; data \textcopyright OpenStreetMap contributors.}
		\label{fig:searchspace}
	\end{figure}
    \FloatBarrier

    \subsection{Additional Networks}
    \label{app:experiments:gb-berlin}
    We report experiments on two additional networks: Great Britain and Berlin.
    As with the networks of Switzerland and Paris, these were extracted from official GTFS feeds~\cite{gtfsHomeGeneral}.
    Statistics for all six networks are listed in~\Cref{tab:datasets:all}.

    \begin{table}[ht]
        \caption{ 
            An overview of the networks used in our experiments.
            Also shown are the sizes of the TB transfer set~$\transfers$ and the compact layout graph~$\layoutGraph$.
            GB denotes Great Britain.
        }
		\label{tab:datasets:all}
		\centering
		\begin{tabular*}{\textwidth}{@{\,}l@{\extracolsep{\fill}}r@{\extracolsep{\fill}}r@{\extracolsep{\fill}}r@{\extracolsep{\fill}}r@{\extracolsep{\fill}}r@{\extracolsep{\fill}}r@{\,}}
			\toprule                      & Europe                     & Germany                    & GB             & Switzerland               & Paris                      & Berlin                    \\
			\midrule Stops                & $\numprint{1346013}$         & $\numprint{435550}$          & $\numprint{260150}$         & $\numprint{29045}$          & $\numprint{41757}$           & $\numprint{25843}$          \\
			Stop events                   & $\numprint{107252199}$       & $\numprint{30680868}$        & $\numprint{19692037}$      & $\numprint{5032795}$        & $\numprint{4636238}$         & $\numprint{3318251}$        \\
			Lines                         & $\numprint{384075}$          & $\numprint{202238}$          & $\numprint{36071}$          & $\numprint{15967}$          & $\numprint{9558}$           & $\numprint{10891}$          \\
			Trips                         & $\numprint{4848876}$         & $\numprint{1547126}$         & $\numprint{506704}$         & $\numprint{319159}$         & $\numprint{215526}$          & $\numprint{146897}$         \\
			Footpaths                     & $\numprint{1752418}$         & $\numprint{1116976}$         & $\numprint{345594}$         & $\numprint{22186}$          & $\numprint{445912}$          & $\numprint{29428}$          \\
            \noalign{\vskip 5pt}
			TB $\absoluteVal{\transfers}$ & $\numprint{269766664}$       & $\numprint{59181359}$        & $\numprint{31443687}$       & $\numprint{8075627}$        & $\numprint{23284123}$        & $\numprint{6104156}$        \\
			$\absoluteVal{\layoutVertices}$      & $\numprint{778419}$          & $\numprint{308004}$          & $\numprint{142221}$         & $\numprint{24257}$          & $\numprint{13765}$           & $\numprint{14157}$          \\
			$\absoluteVal{\layoutEdges}$  & $\numprint{2298358}$         & $\numprint{1038032}$         & $\numprint{393256}$         & $\numprint{60154}$          & $\numprint{45882}$           & $\numprint{42422}$          \\
			\bottomrule
		\end{tabular*}
	\end{table}

    \Cref{fig:rank:extra} shows the rank distributions for the Great Britain and Berlin networks (\cf \Cref{fig:rank}).
    Great Britain exhibits a similar distribution to Europe.
    Berlin exhibits a flatter distribution, which is similar to Paris but not as pronounced.

    \begin{figure}[t]
		\centering
		\begin{tikzpicture}

\begin{groupplot}[
    group style={
        group size=2 by 1, 
        horizontal sep=2cm, 
        vertical sep=2cm},
    width=6cm,
    height=5cm,
    ybar=0pt,
    /pgf/bar width=7pt,
    ymin=0,
    ymax=100,
    xlabel={Rank},
    ylabel={Transfers (\%)},
    ymajorgrids,
    grid style={dashed,gray!60},
    every axis/.append style={
        tick label style={font=\footnotesize},
        label style={font=\small},
    }
]

\nextgroupplot[title={Great Britain}]
\addplot+[draw=black, fill=blue!40]
table[x=rank,y=percent,col sep=comma] {plot_scripts/data/gb.csv};

\nextgroupplot[title={Berlin}]
\addplot+[draw=black, fill=blue!40]
table[x=rank,y=percent,col sep=comma] {plot_scripts/data/berlin.csv};

\end{groupplot}

\end{tikzpicture}
        \caption{
            Distribution of the ranks among the transfer set for the Great Britain and Berlin networks.
		}
		\label{fig:rank:extra}
	\end{figure}
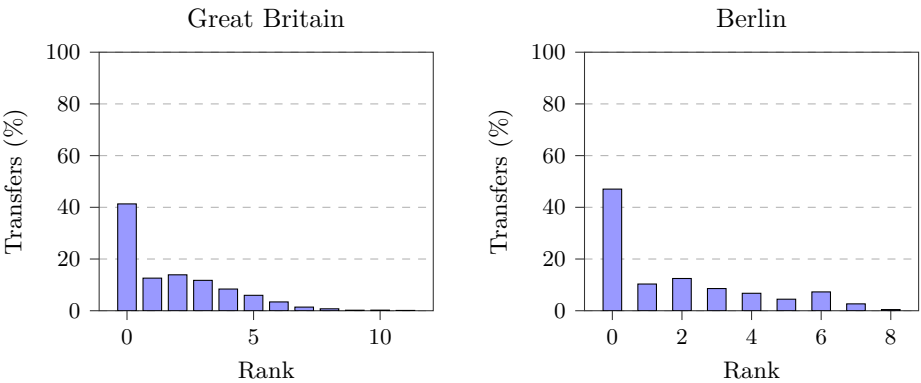

    \begin{table}[t]
		\caption{
    		Preprocessing times on the Great Britain and Berlin networks for T-REX and ACSA by phase, measured in seconds.
		}
		\label{tab:prepro:extra}
		\centering
		\begin{tabular*}{\textwidth}{@{\,}l@{\extracolsep{\fill}}r@{\extracolsep{\fill}}r@{\extracolsep{\fill}}r@{\extracolsep{\fill}}r@{\extracolsep{\fill}}r@{\extracolsep{\fill}}r@{\extracolsep{\fill}}@{\,}}
			\toprule
            \multirow{2}{*}{Network} & \multirow{2}{*}{Partitioning} & \multicolumn{3}{c}{T-REX} & \multicolumn{2}{c}{ACSA} \\
            \cmidrule(){3-5} \cmidrule(){6-7}
            & & Transfers & Customization & Total & Customization & Total \\
			\midrule
            Great Britain & 8.01 & 2.38 & 3.55 & 13.94 & 8.41 & 16.42\\
            Berlin & 0.99 & 0.36 & 0.75 & 2.10 & 2.78 & 3.77\\
			\bottomrule
		\end{tabular*}
	\end{table}

    \Cref{tab:prepro:extra} and \Cref{tab:uniformrandomquery:extra} compare the memory consumption and performance of T-REX to other algorithms, respectively (cf.~\Cref{tab:prepro} and~\Cref{tab:uniformrandomquery:country}).
    The achieved speedup over TB is $\numprint{7.7}$ for Great Britain and~$\numprint{2.5}$ for Berlin.
    Compared to the other metropolitan network, Paris, the speedup on Berlin is slightly better.
    This is likely due to the high footpath density in the Paris network.
    As footpaths are contracted, the compact layout graph becomes denser and smaller, making it more difficult to find good partitions.
    
    \begin{table}[ht]
		\caption{
            Performance of algorithms on the Great Britain and Berlin networks, averaged over $\numprint{10000}$ random queries. 
        }
		\label{tab:uniformrandomquery:extra}
		\centering
		\begin{tabular*}{\textwidth}{@{\,}l@{\extracolsep{\fill}}l@{\extracolsep{\fill}}c@{\extracolsep{\fill}}r@{\extracolsep{\fill}}r@{\extracolsep{\fill}}r@{\,}}
			\toprule
            \multirow{2}{*}{Network} & \multirow{2}{*}{Algorithm} & \multirow{2}{*}{Pareto} & Query & Preprocessing & \multirow{2}{*}{$k$}\\
            & & & $\left[\si{\mu\second}\right]$ & $\left[ \mathrm{hh}{:}\mathrm{mm}{:}\mathrm{ss}\right]$ & \\
			\midrule 
			\multirow{7}{*}{Great Britain} & CSA       & $\circ$   & $\numprint{43189}$               & --                                                                        & --               \\
			& ACSA      & $\circ$   & $\numprint{9388}$               & \printTime{0}{0}{16}                                                      & $\numprint{1024}$    \\
			& RAPTOR    & $\bullet$ & $\numprint{65776}$           & --                                                                        & --               \\
			& TB        & $\bullet$ & $\numprint{18051}$           & \printTime{0}{0}{2}                                                      & --               \\
			& T-REX     & $\bullet$ & $\numprint{1192}$            & \printTime{0}{0}{14}                                                      & $\numprint{2048}$  \\
			& FLASH-TB (moderate) & $\bullet$ & $\numprint{2315}$            & \printTime{11}{42}{4}                                                      & $\numprint{8}$  \\
			& FLASH-TB (fastest) & $\bullet$ & $\numprint{59}$            & \printTime{10}{03}{14}                                                      & $\numprint{16384}$  \\
            \noalign{\vskip 8pt}
			\multirow{7}{*}{Berlin} & CSA       & $\circ$   & $\numprint{2415}$            & --                                                                        & --                \\
			& ACSA      & $\circ$   & $\numprint{2493}$            & \printTime{0}{0}{4}                                                      & $\numprint{128}$   \\
			& RAPTOR    & $\bullet$ & $\numprint{7259}$            & --                                                                        & --                \\
			& TB        & $\bullet$ & $\numprint{2931}$            & < \printTime{0}{0}{1}                                                      & --                \\
			& T-REX     & $\bullet$ & $\numprint{786}$            & \printTime{0}{0}{2}                                                      & $\numprint{512}$   \\
			& FLASH-TB (moderate) & $\bullet$ & $\numprint{599}$              & \printTime{0}{15}{59}                                                    & $\numprint{8}$ \\
			& FLASH-TB (fastest) & $\bullet$ & $\numprint{33}$              & \printTime{0}{15}{20}                                                    & $\numprint{16384}$ \\
			\bottomrule
		\end{tabular*}
	\end{table}
    \FloatBarrier

    \subsection{Partitioning}
	\label{app:experiments:partition}

    For the multilevel partitioning of the compact layout graph, we tested two open-source partitioners: KaHIP~\cite{San13,githubGitHubKaHIPKaHIP}, which is used by ACSA and FLASH-TB, and the more recent Mt-KaHyPar~\cite{Got22,Got23,Got24,githubGitHubKahyparmtkahypar}.
    Although Mt-KaHyPar was designed primarily for hypergraph partitioning, it also supports ordinary graphs.
    In preliminary experiments, we also evaluated two partitioning approaches that use the geographical coordinates of the stops: a $2$d-tree and an approach based on hierarchical $k$-means clustering, similar to the one used by Scalable TP~\cite{Bas16}.
    However, these performed much worse than the black-box partitioners KaHIP and Mt-KaHyPar.

    For KaHIP, we evaluated the configurations \texttt{strong} and \texttt{ssocial}.
    The latter was designed for graphs with an inhomogeneous degree distribution and was observed to yield better results for FLASH-TB.
    For Mt-KaHyPar, we used the preset \texttt{highest\_quality} and \texttt{flow\_cutter} as the initial partitioner.
    KaHIP was run sequentially as it does not support parallelization for multilevel partitioning.
    We parallelized Mt-KaHyPar with~$6$ cores, as we observed that adding further cores degraded the partition quality.
		
	As recursive bipartitioning is performed top-down, imbalance can accumulate across levels.
	To ensure that the final $2^\numLevels$-partition still respects the global imbalance bound~$\varepsilon$, Mt-KaHyPar adaptively reduces the allowed imbalance for the higher levels~\cite{Sch16}.
    We implemented a modified version, denoted as Mt-KaHyPar$^*$~\cite{githubGitHubPatrickSteilmtkahypar}, that disables this behavior.
    Instead, the given imbalance~$\varepsilon$ is used on every level except the highest level, where the imbalance constraint is omitted entirely.
    Because the heuristics employed by Mt-KaHyPar are geared towards finding balanced partitions, we observed that it still produces a fairly balanced top-level partition without any constraint.
    Allowing extreme imbalances at the lower levels can negatively affect both the customization time, as the larger cells are more expensive to process, and the query time, as more source-target pairs have a low LCL.
    However, we found that these effects were outweighed by the smaller cut size.

    \Cref{tab:partitioner} compares the different partitioners on the Europe network with the configuration that yielded the best results overall ($14$ levels, $\varepsilon=50$\,\%).
    \Cref{fig:partition:ibes} additionally compares the number of IBEs on each level.
    The number of IBEs on the lowest level is similar for all configurations except Mt-KaHyPar$^*$, which reduces the total number of IBEs by a factor~$2$ and shows improvements across all levels.
    Because it does not attempt to keep the overall $2^\numLevels$-partition balanced, it has more flexibility to produce a highly granular subdivision in sparse regions without being forced to also partition metropolitan clusters that do not have small cuts.

    \begin{table}[ht]
    	\caption{Results for different partitioners on the Europe network with $14$ levels and $50$\,\% imbalance. Mt-KaHyPar is parallelized with~$6$ cores. Query times are T-REX with overlays, averaged over $\numprint{10000}$ random queries.}
		\label{tab:partitioner}
		\centering
		\begin{tabular*}{\textwidth}{@{\,}l@{\extracolsep{\fill}}r@{\extracolsep{\fill}}r@{\extracolsep{\fill}}r@{\extracolsep{\fill}}r@{\extracolsep{\fill}}r@{\,}}
			\toprule
            \multirow{2}{*}{Partitioner} & Partitioning & IBEs & IBEs & Customization & Query\\
            & $\left[\mathrm{mm}{:}\mathrm{ss}\right]$ & (total) & (top-level) & $\left[\mathrm{mm}{:}\mathrm{ss}\right]$ & $\left[\si{\mu\second}\right]$\\
			\midrule 
            KaHIP \texttt{strong} & \printMinuteSeconds{20}{34} & \numprint{8838870} & \numprint{399427} & \printMinuteSeconds{5}{53} & \numprint{44727} \\
            KaHIP \texttt{ssocial} & \printMinuteSeconds{13}{00} & \numprint{8774535} & \numprint{66794} & \printMinuteSeconds{01}{50} & \numprint{21744} \\
            Mt-KaHyPar & \printMinuteSeconds{01}{07} & \numprint{9293770} & \numprint{2793} & \printMinuteSeconds{01}{11} & \numprint{11028} \\
            Mt-KaHyPar$^*$ & \printMinuteSeconds{00}{54} & $\numprint{4648113}$ & \numprint{701} & \printMinuteSeconds{0}{54} & $\numprint{8725}$ \\
			\bottomrule
		\end{tabular*}
    \end{table}
    
    \begin{figure}[H]
        \centering
        \begin{tikzpicture}
\begin{axis}[
    ybar,
    bar width=3pt,
    width=\textwidth,
    height=7cm,
    legend style={at={(0.5,-0.2)}, anchor=north, legend columns=4, column sep=0.2cm},
    legend cell align=left,
    ylabel={IBEs},
    xlabel={Level},
    symbolic x coords={0,1,2,3,4,5,6,7,8,9,10,11,12,13},
    xtick=data,
    ymode=log,
    ymin=1,
    ymax=2e7,
    grid=major,
    log basis y={10},
    ytick={1e0,1e1,1e2,1e3,1e4,1e5,1e6,1e7},
]

\addplot coordinates {
(0,8838870) (1,7119616) (2,5259610)
(3,3627055) (4,2346351) (5,1579029) (6,1018132)
(7,697988) (8,589273) (9,512141) (10,467077) (11,434382)
(12,419755) (13,399427)
};
\addlegendentry{KaHIP \texttt{strong}}

\addplot coordinates {
(0,8774535) (1,7107747) (2,5241244) (3,3579632) (4,2258963) (5,1364422) (6,736335) (7,393080) (8,242654) (9,175043) (10,144667) (11,118546) (12,77646) (13,66794)
};
\addlegendentry{KaHIP \texttt{ssocial}}

\addplot coordinates {
(0,9275607) (1,6804734) (2,4701457) (3,3006206) (4,1809233) (5,1053819) (6,556637) (7,282978) (8,146398) (9,81601) (10,44817) (11,22711) (12,9892) (13,2793)
};
\addlegendentry{Mt-KaHyPar}

\addplot coordinates {
(0,4470928) (1,3055459) (2,2012568) (3,1294270) (4,792838) (5,457654) (6,250998) (7,132590) (8,74821) (9,44368) (10,26844) (11,15155) (12,4219) (13,701)
};
\addlegendentry{Mt-KaHyPar$^*$}

\end{axis}
\end{tikzpicture}
        \caption{
            Number of IBEs per level for different partitioners, measured on the Europe network with $14$ levels and $50$\,\% imbalance.
		}
		\label{fig:partition:ibes}
    \end{figure}
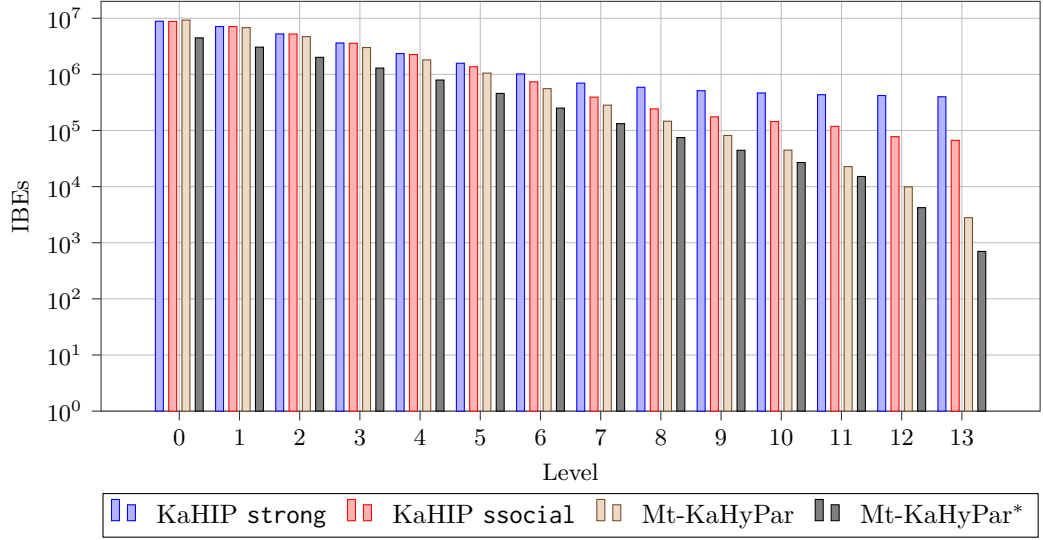

    On the higher levels, the differences between KaHIP and Mt-KaHyPar are drastic.
    KaHIP \texttt{strong} produces a very large top-level cut and only achieves a twofold reduction in the number of IBEs between level~$7$ and level~$13$.
    The \texttt{ssocial} configuration indeed finds significantly smaller cuts on the higher levels, but the reduction in the number of IBEs still stagnates.
    By contrast, Mt-KaHyPar achieves a strong reduction with each subsequent level.
    As a result, the number of IBEs for the top-level cut is smaller by more than a factor of~$100$ compared to KaHIP \texttt{strong}.
    \Cref{fig:partition.plots} visualizes the two topmost levels for the different partitioners.
    We observe that the cells computed by KaHIP are highly fragmented, whereas Mt-KaHyPar produces geographically coherent cells. 
    In particular, KaHIP often separates metropolitan areas (e.g., London) from their surroundings and instead groups them with far-away regions.

    \begin{figure}[ht]
		\centering
        \includegraphics[width=\textwidth]{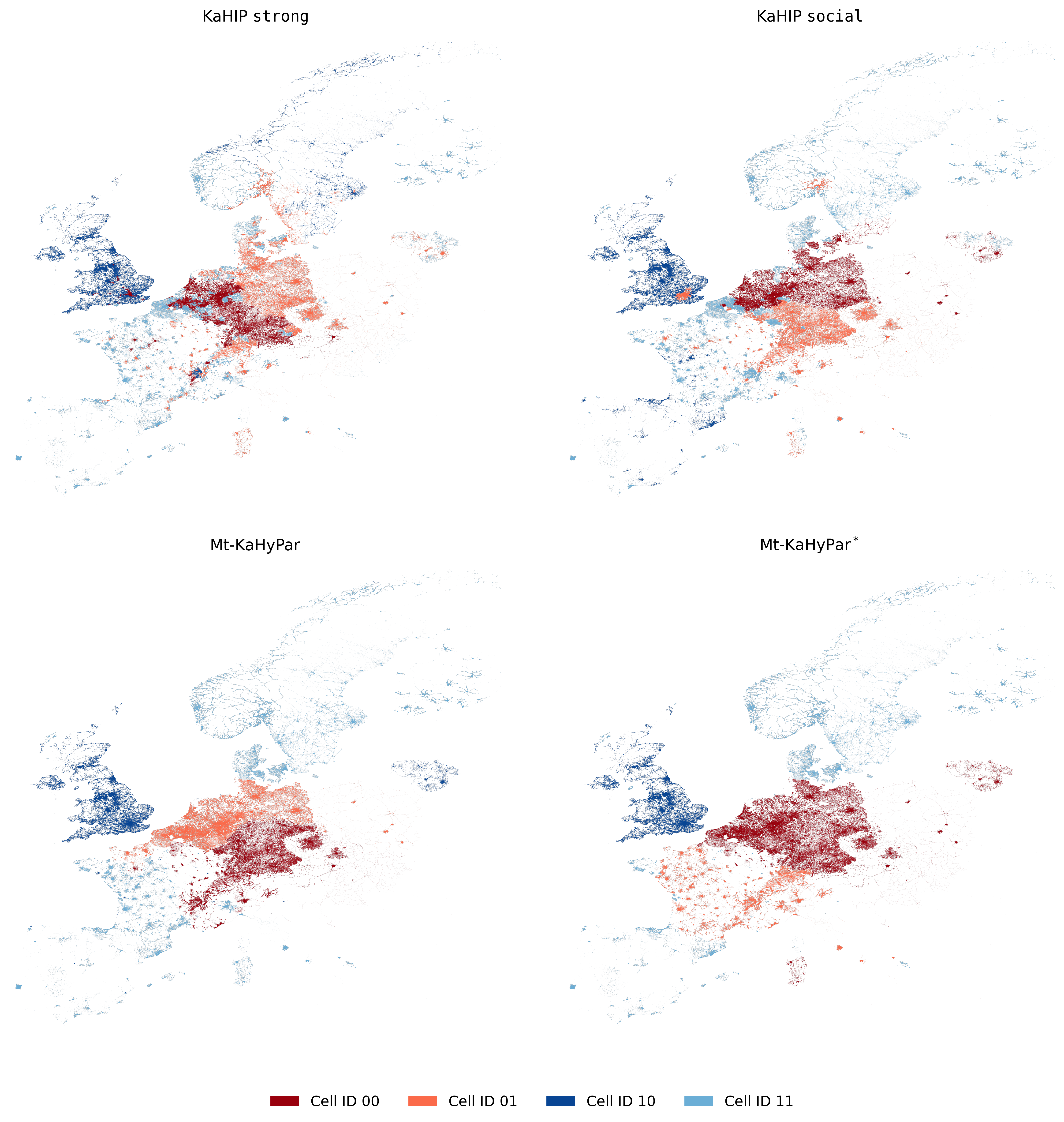}
        \caption{The first two levels of the different partitions for Europe ($14$ levels, $50\,\%$ imbalance). The cells with IDs $00$ and $01$ (shaded in red) form top-level cell~$0$, whereas the cells with IDs $10$ and $11$ (shaded in blue) form top-level cell~$1$.}
        \label{fig:partition.plots}
	\end{figure}
        
    The reductions in the number of IBEs also correspond to a reduction in the transfer ranks (cf.~\Cref{fig:partition:transfers}) and to faster customization and query times.
    Compared to KaHIP \texttt{strong}, Mt-KaHyPar$^*$ speeds up the customization by a factor of~\numprint{5.4} and the query by a factor of~\numprint{5.1}.
    In particular, Mt-KaHyPar$^*$ yields faster customization and query times than Mt-KaHyPar, which indicates that allowing high imbalances pays off overall.
    The partitioning time is also drastically reduced by switching from KaHIP to Mt-KaHyPar.
    In all subsequent experiments, we only use Mt-KaHyPar$^*$.
    For each network and each choice for the number of levels~$\numLevels$, we evaluate four imbalance values ($\numprint{25}\,\%$, $\numprint{50}\,\%$, $\numprint{75}\,\%$ and $\numprint{100}\,\%$) and report the configuration with the lowest query times (cf.~\Cref{tab:varinglevel} and Appendix~\ref{app:experiments:metrics}).

    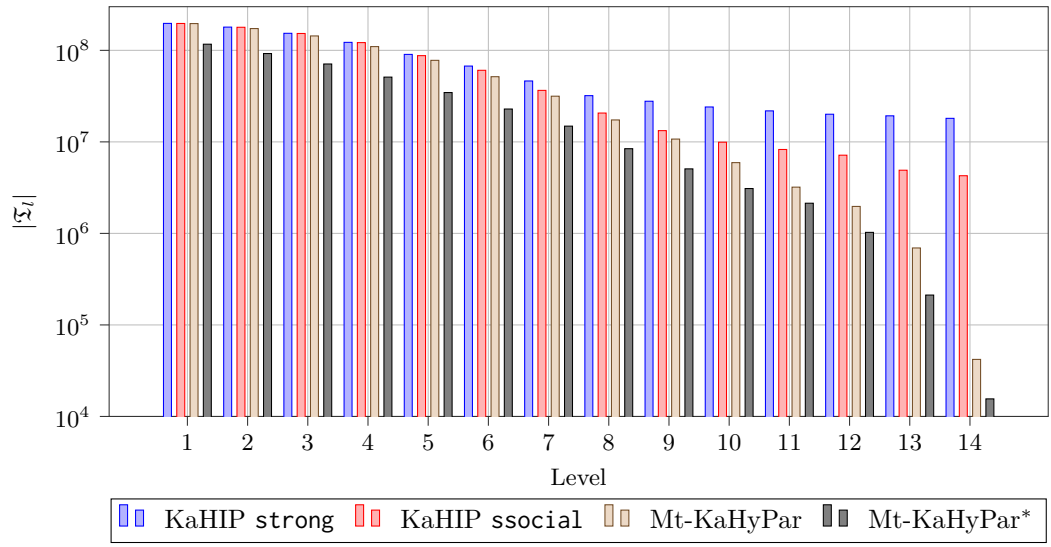
\begin{figure}[H]
        \centering
        \begin{tikzpicture}
\begin{axis}[
    ybar,
    bar width=3pt,
    width=\textwidth,
    height=7cm,
    legend style={at={(0.5,-0.2)}, anchor=north, legend columns=4, column sep=0.2cm},
    legend cell align=left,
    ylabel={$\absoluteVal{\transfers_l}$},
    xlabel={Level},
    symbolic x coords={1,2,3,4,5,6,7,8,9,10,11,12,13,14},
    xtick=data,
    ymode=log,
    ymin=1e4,
    ymax=3e8,
    grid=major,
    log basis y={10},
    ytick={1e4,1e5,1e6,1e7,1e8},
]

\addplot coordinates {
(1,197068348) (2,179466845) (3,153566789)
(4,122278270) (5,90325129) (6,67350985) (7,46281816)
(8,32006734) (9,27800260) (10,24063078) (11,21810528)
(12,20049736) (13,19265280) (14,18080370)
};
\addlegendentry{KaHIP \texttt{strong}}

\addplot coordinates {
(1,196678398) (2,178758082) (3,152876068)
(4,121530899) (5,87544543) (6,60544890) (7,36535821)
(8,20654333) (9,13272422) (10,9918411) (11,8259700)
(12,7155119) (13,4902287) (14,4268050)
};
\addlegendentry{KaHIP \texttt{ssocial}}

\addplot coordinates {
(1,196183171) (2,173033913) (3,143728000)
(4,109915406) (5,77810955) (6,51549978) (7,31568197)
(8,17366545) (9,10730751) (10,5934536) (11,3202666)
(12,1971555) (13,693353) (14,42011)
};
\addlegendentry{Mt-KaHyPar}

\addplot coordinates {
(1,116738265) (2,92309919) (3,70963005)
(4,51013611) (5,34594757) (6,22855984) (7,14859785)
(8,8423220) (9,5070643) (10,3092791) (11,2137041)
(12,1026292) (13,212016) (14,15558)
};
\addlegendentry{Mt-KaHyPar$^*$}

\end{axis}
\end{tikzpicture}
        \caption{
            Size of the transfer overlay per level for different partitioners, measured on the Europe network with $14$ levels and $50$\,\% imbalance.
		}
		\label{fig:partition:transfers}
    \end{figure}
    \FloatBarrier

    \subsection{Memory Consumption}
    \label{app:experiments:memory}
    The memory consumption of the basic and overlay variants of T-REX is reported in~\Cref{tab:memory}.
    In addition to the TB data structures, the basic variant stores two bytes per stop and per stop event for the cell IDs and one byte per transfer for the rank.
    This amounts to a total memory overhead of~$5$--$9$\,\% over TB.
    The overlay variant is significantly more expensive; it requires~$3$--$4$ times as much space as TB and~$12$--$18$ times as much as the timetable itself.
    This is mainly due to the $\successor{\aTrip[i]}{\aLevel}$ data structure and the transfer overlays, which exist once for each level.
    Because the number of transfers in the overlay decreases with each level, the overall space required for the transfers themselves is significantly smaller than~$\numLevels \absoluteVal{\transfers}$.
    However, because the overlays are stored in an adjacency array format, there is a fixed overhead of~$4\absoluteVal{\stopEvents}$ bytes per level to store the pointer to the outgoing transfers of each stop event.
    
    We note that, compared to algorithms such as RAPTOR, TB already prioritizes query speed over memory overhead.
    In particular, the precomputed transfer set requires more space than the input network.
    For T-REX, we stayed within this design paradigm and tuned the algorithm for maximum query speed.
    However, several steps could be taken to reduce the memory consumption of the overlay variant at the expense of some query speed:
    As \Cref{tab:varinglevel} shows, the number of levels can be reduced substantially without a large degradation in the query performance.
    It would also be possible to designate some levels as \emph{phantom levels}, similarly to CRP~\cite{Del17}:
    These levels are computed during the customization but then discarded.
    During a query, the rank information is instead taken from the next-lowest level that is still stored.
    Finally, one could assign a second, smaller set of consecutive stop event IDs for the higher levels.
    This would make it possible to shrink the adjacent array representations because events without outgoing transfers could be skipped.
    
    \begin{table}[ht]
    	\caption{Memory consumption in $\si{\mega\byte}$ for the timetables and the TB transfer set, as well the total memory consumption for TB and T-REX.}
		\label{tab:memory}
		\centering
		\begin{tabular*}{\textwidth}{@{\,}l@{\extracolsep{\fill}}r@{\extracolsep{\fill}}r@{\extracolsep{\fill}}r@{\extracolsep{\fill}}r@{\extracolsep{\fill}}r@{\,}}
			\toprule
            \multirow{2}{*}{Network} & \multirow{2}{*}{Timetable} & \multirow{2}{*}{Transfers} & \multirow{2}{*}{TB} & \multicolumn{2}{c}{T-REX} \\
            \cmidrule(){5-6}
            & & & & Without overlays & With overlays\\
			\midrule 
            Europe & \numprint{1813} & \numprint{3496} & \numprint{8256} & \numprint{8773} & \numprint{28459} \\
            Germany & \numprint{541} & \numprint{794} & \numprint{2190} & \numprint{2304} & \numprint{7974} \\
            Great Britain & \numprint{329} & \numprint{435} & \numprint{1303} & \numprint{1363} & \numprint{4401} \\
            Switzerland & \numprint{83} & \numprint{112} & \numprint{331} & \numprint{347} & \numprint{1026} \\
            Paris & \numprint{79} & \numprint{284} & \numprint{492} & \numprint{537} & \numprint{1378} \\
            Berlin & \numprint{55} & \numprint{83} & \numprint{228} & \numprint{240} & \numprint{683} \\
			\bottomrule
		\end{tabular*}
    \end{table}

    \subsection{Performance by Number of Levels}
    \label{app:experiments:metrics}
    The performance of T-REX and ACSA depending on the number of levels is reported for Germany in~\Cref{tab:varinglevel:eu}, Great Britain and Switzerland in~\Cref{tab:varinglevel:gbswiss}, and Paris and Berlin in~\Cref{tab:varinglevel:pabe} (cf.~\Cref{tab:varinglevel} for Europe).
    Depending on the network, the query performance starts to stagnate between six to nine levels.
    On the metropolitan networks, ACSA does not achieve any speedup due to the overhead required for merging and sorting the relevant connections, and this effect is largely independent of the number of levels.
    By contrast, T-REX achieves a modest but clear speedup and benefits from additional levels.
    
    \begin{table}[ht]
        \caption{
            Detailed performance measurements for T-REX (with overlays) and ACSA depending on the number of levels, averaged over $\numprint{10000}$ random queries on the Germany network.
            For level~$0$, we report the metrics for TB and CSA.
            Query times are shown as $\left[\si{\mu\second}\right]$, customization as $\left[ \mathrm{mm}{:}\mathrm{ss}\right]$.
            For each level, we report the imbalance (measured in $\%$) that led to the best query performance.
        }
        \label{tab:varinglevel:eu}
        \centering
        \begin{tabular*}{\textwidth}{@{\,}r@{\extracolsep{\fill}}r@{\extracolsep{\fill}}r@{\extracolsep{\fill}}r@{\extracolsep{\fill}}r@{\extracolsep{\fill}}r@{\extracolsep{\fill}}r@{\extracolsep{\fill}}r@{\extracolsep{\fill}}r@{\extracolsep{\fill}}r@{\,}}
            \toprule
            \multirow{2}{*}{Levels} & \multicolumn{5}{c}{T-REX} & \multicolumn{3}{c}{ACSA} \\
            \cmidrule(){2-6} \cmidrule(){7-9}
            & Scan. trips & Rel. transfers & Query & Custo. & Imbal. & Query & Custo. & Imbal. \\
            \midrule
            $\numprint{0}$ & $\numprint{316797}$ & $\numprint{4675832}$ & $\numprint{65060}$ & -- & -- & $\numprint{82261}$ & -- & -- \\
            $\numprint{1}$ & $\numprint{245541}$ & $\numprint{3619390}$ & $\numprint{55604}$ & \printMinuteSeconds{0}{5} & $\numprint{25}$ & $\numprint{139463}$ & \printMinuteSeconds{0}{41} & $\numprint{50}$ \\
            $\numprint{2}$ & $\numprint{177605}$ & $\numprint{2623455}$ & $\numprint{38855}$ & \printMinuteSeconds{0}{5} & $\numprint{50}$ & $\numprint{122523}$ & \printMinuteSeconds{0}{28} & $\numprint{50}$ \\
            $\numprint{3}$ & $\numprint{117957}$ & $\numprint{1694979}$ & $\numprint{24641}$ & \printMinuteSeconds{0}{5} & $\numprint{25}$ & $\numprint{91759}$ & \printMinuteSeconds{0}{25} & $\numprint{25}$ \\
            $\numprint{4}$ & $\numprint{86371}$ & $\numprint{1181547}$ & $\numprint{17664}$ & \printMinuteSeconds{0}{5} & $\numprint{100}$ & $\numprint{71654}$ & \printMinuteSeconds{0}{18} & $\numprint{100}$ \\
            $\numprint{5}$ & $\numprint{51628}$ & $\numprint{607465}$ & $\numprint{9703}$ & \highlightCell \printMinuteSeconds{0}{4} & $\numprint{25}$ & $\numprint{50294}$ & \printMinuteSeconds{0}{11} & $\numprint{25}$ \\
            $\numprint{6}$ & $\numprint{43456}$ & $\numprint{438698}$ & $\numprint{7751}$ & \highlightCell \printMinuteSeconds{0}{4} & $\numprint{25}$ & $\numprint{38472}$ & \highlightCell \printMinuteSeconds{0}{9} & $\numprint{25}$ \\
            $\numprint{7}$ & $\numprint{36916}$ & $\numprint{313479}$ & $\numprint{6295}$ & \highlightCell \printMinuteSeconds{0}{4} & $\numprint{25}$ & $\numprint{30074}$ & \highlightCell \printMinuteSeconds{0}{9} & $\numprint{25}$ \\
            $\numprint{8}$ & $\numprint{30537}$ & $\numprint{219917}$ & $\numprint{4928}$ & \highlightCell \printMinuteSeconds{0}{4} & $\numprint{25}$ & $\numprint{20810}$ & \highlightCell \printMinuteSeconds{0}{9} & $\numprint{25}$ \\
            $\numprint{9}$ & $\numprint{29397}$ & $\numprint{201902}$ & $\numprint{4673}$ & \printMinuteSeconds{0}{5} & $\numprint{50}$ & $\numprint{19695}$ & \printMinuteSeconds{0}{10} & $\numprint{50}$ \\
            $\numprint{10}$ & $\numprint{29739}$ & $\numprint{178544}$ & $\numprint{4698}$ & \printMinuteSeconds{0}{6} & $\numprint{25}$ & \highlightCell $\numprint{19206}$ & \printMinuteSeconds{0}{13} & $\numprint{25}$ \\
            $\numprint{11}$ & $\numprint{29034}$ & $\numprint{167440}$ & $\numprint{4595}$ & \printMinuteSeconds{0}{7} & $\numprint{25}$ & $\numprint{19804}$ & \printMinuteSeconds{0}{17} & $\numprint{25}$ \\
            $\numprint{12}$ & $\numprint{28741}$ & $\numprint{173235}$ & $\numprint{4561}$ & \printMinuteSeconds{0}{7} & $\numprint{75}$ & $\numprint{20350}$ & \printMinuteSeconds{0}{26} & $\numprint{75}$ \\
            $\numprint{13}$ & $\numprint{28151}$ & $\numprint{167009}$ & $\numprint{4450}$ & \printMinuteSeconds{0}{9} & $\numprint{50}$ & $\numprint{19540}$ & \printMinuteSeconds{0}{44} & $\numprint{50}$ \\
            $\numprint{14}$ & \highlightCell $\numprint{27998}$ & \highlightCell $\numprint{160570}$ & \highlightCell $\numprint{4442}$ & \printMinuteSeconds{0}{11} & $\numprint{50}$ & $\numprint{20369}$ & \printMinuteSeconds{1}{50} & $\numprint{50}$ \\
            \bottomrule
        \end{tabular*}
    \end{table}
    
    \begin{table}[ht]
        \caption{
            Detailed performance measurements for T-REX (with overlays) and ACSA depending on the number of levels, averaged over $\numprint{10000}$ random queries on the Great Britain network (above) and the Switzerland network (below).
            For level~$0$, we report the metrics for TB and CSA.
            Query times are shown as $\left[\si{\mu\second}\right]$, customization as $\left[ \mathrm{mm}{:}\mathrm{ss}\right]$.
            For each level, we report the imbalance (measured in $\%$) that led to the best query performance.
        }
        \label{tab:varinglevel:gbswiss}
        \centering
        \begin{tabular*}{\textwidth}{@{\,}r@{\extracolsep{\fill}}r@{\extracolsep{\fill}}r@{\extracolsep{\fill}}r@{\extracolsep{\fill}}r@{\extracolsep{\fill}}r@{\extracolsep{\fill}}r@{\extracolsep{\fill}}r@{\extracolsep{\fill}}r@{\extracolsep{\fill}}r@{\,}}
            \toprule
            \multirow{2}{*}{Levels} & \multicolumn{5}{c}{T-REX} & \multicolumn{3}{c}{ACSA} \\
            \cmidrule(){2-6} \cmidrule(){7-9}
            & Scan. trips & Rel. transfers & Query & Custo. & Imbal. & Query & Custo. & Imbal. \\
            \midrule
            $\numprint{0}$ & $\numprint{79054}$ & $\numprint{1388371}$ & $\numprint{18051}$ & -- & -- & $\numprint{43137}$ & -- & --  \\
            $\numprint{5}$ & $\numprint{12848}$ & $\numprint{155005}$ & $\numprint{2253}$ & \printMinuteSeconds{0}{1} & $\numprint{25}$ & $\numprint{26560}$ & \highlightCell \printMinuteSeconds{0}{5} & $\numprint{25}$ \\
            $\numprint{6}$ & $\numprint{10334}$ & $\numprint{108105}$ & $\numprint{1712}$ & \printMinuteSeconds{0}{1} & $\numprint{75}$  & $\numprint{19077}$ & \highlightCell \printMinuteSeconds{0}{5} & $\numprint{75}$ \\
            $\numprint{7}$ & $\numprint{9072}$ & $\numprint{64136}$ & $\numprint{1345}$ & \printMinuteSeconds{0}{1} & $\numprint{25}$  & $\numprint{13190}$ & \highlightCell \printMinuteSeconds{0}{5} & $\numprint{25}$ \\
            $\numprint{8}$ & $\numprint{8712}$ & $\numprint{50687}$ & $\numprint{1240}$ & \printMinuteSeconds{0}{1} & $\numprint{25}$  & $\numprint{10616}$ & \printMinuteSeconds{0}{6} & $\numprint{25}$ \\
            $\numprint{9}$ & $\numprint{8756}$ & $\numprint{49866}$ & $\numprint{1260}$ & \printMinuteSeconds{0}{2} & $\numprint{50}$  & $\numprint{9641}$ & \printMinuteSeconds{0}{7} & $\numprint{25}$  \\
            $\numprint{10}$ & $\numprint{8681}$ & $\numprint{42698}$ & $\numprint{1200}$ & \printMinuteSeconds{0}{3} & $\numprint{25}$  & \highlightCell $\numprint{9388}$ & \printMinuteSeconds{0}{8} & $\numprint{25}$  \\
            $\numprint{11}$ & \highlightCell $\numprint{8606}$ & \highlightCell $\numprint{40730}$ & \highlightCell $\numprint{1192}$ & \printMinuteSeconds{0}{4} & $\numprint{25}$  & $\numprint{9647}$ & \printMinuteSeconds{0}{9} & $\numprint{25}$ \\
            \noalign{\vskip 7pt}
            $\numprint{0}$ & $\numprint{32577}$ & $\numprint{309660}$ & $\numprint{4510}$ & -- & --  & $\numprint{3931}$ & -- & --  \\
            $\numprint{4}$ & $\numprint{9064}$ & $\numprint{66733}$ & $\numprint{1170}$ & \highlightCell < \printMinuteSeconds{0}{1} & $\numprint{25}$ & $\numprint{3710}$ & \highlightCell \printMinuteSeconds{0}{1} & $\numprint{50}$ \\
            $\numprint{5}$  & $\numprint{7938}$ & $\numprint{47363}$ & $\numprint{924}$ & \highlightCell < \printMinuteSeconds{0}{1} & $\numprint{25}$ & $\numprint{3039}$ & \highlightCell \printMinuteSeconds{0}{1} & $\numprint{25}$ \\
            $\numprint{6}$ & $\numprint{6985}$ & $\numprint{35989}$ & $\numprint{772}$ & \highlightCell < \printMinuteSeconds{0}{1} & $\numprint{25}$ & $\numprint{2412}$ & \printMinuteSeconds{0}{2} & $\numprint{25}$ \\
            $\numprint{7}$ & $\numprint{6710}$ & $\numprint{31571}$ & $\numprint{728}$ & \highlightCell < \printMinuteSeconds{0}{1} & $\numprint{25}$ & $\numprint{2048}$ & \printMinuteSeconds{0}{2} & $\numprint{25}$ \\
            $\numprint{8}$ & $\numprint{6824}$ & $\numprint{35134}$ & $\numprint{735}$ & \highlightCell < \printMinuteSeconds{0}{1} & $\numprint{50}$ & $\numprint{1965}$ & \printMinuteSeconds{0}{2} & $\numprint{25}$ \\
            $\numprint{9}$ & $\numprint{6688}$ & \highlightCell $\numprint{29671}$ & \highlightCell $\numprint{710}$ & \printMinuteSeconds{0}{1} & $\numprint{25}$  & \highlightCell $\numprint{1901}$ & \printMinuteSeconds{0}{3} & $\numprint{25}$ \\
            $\numprint{10}$ & $\numprint{6732}$ & $\numprint{33474}$ & $\numprint{746}$ & \printMinuteSeconds{0}{1} & $\numprint{25}$ & $\numprint{2009}$  & \printMinuteSeconds{0}{3} & $\numprint{25}$  \\
            $\numprint{11}$ & \highlightCell $\numprint{6470}$ & $\numprint{31171}$ & $\numprint{717}$ & \printMinuteSeconds{0}{1} & $\numprint{75}$  & $\numprint{2010}$ & \printMinuteSeconds{0}{3} & $\numprint{25}$  \\
            \bottomrule
        \end{tabular*}
    \end{table}
    
    \begin{table}[ht]
        \caption{
            Detailed performance measurements for T-REX (with overlays) and ACSA depending on the number of levels, averaged over $\numprint{10000}$ random queries on the Paris network (above) and the Berlin network (below).
            For level~$0$, we report the metrics for TB and CSA.
            Query times are shown as $\left[\si{\mu\second}\right]$, customization as $\left[ \mathrm{mm}{:}\mathrm{ss}\right]$.
            For each level, we report the imbalance (measured in $\%$) that led to the best query performance.
        }
        \label{tab:varinglevel:pabe}
        \centering
        \begin{tabular*}{\textwidth}{@{\,}r@{\extracolsep{\fill}}r@{\extracolsep{\fill}}r@{\extracolsep{\fill}}r@{\extracolsep{\fill}}r@{\extracolsep{\fill}}r@{\extracolsep{\fill}}r@{\extracolsep{\fill}}r@{\extracolsep{\fill}}r@{\extracolsep{\fill}}r@{\,}}
            \toprule
            \multirow{2}{*}{Levels} & \multicolumn{5}{c}{T-REX} & \multicolumn{3}{c}{ACSA} \\
            \cmidrule(){2-6} \cmidrule(){7-9}
            & Scan. trips & Rel. transfers & Query & Custo. & Imbal. & Query & Custo. & Imbal. \\
            \midrule
            0 & $\numprint{20574}$ & $\numprint{422123}$ & $\numprint{3526}$ & -- & -- & \highlightCell $\numprint{3280}$ & -- & -- \\
            5 & $\numprint{13304}$ & $\numprint{125124}$ & $\numprint{1826}$ & \highlightCell \printMinuteSeconds{0}{3} & $\numprint{25}$ & $\numprint{6581}$ & \highlightCell \printMinuteSeconds{0}{13} & $\numprint{25}$ \\
            6 & $\numprint{12665}$ & $\numprint{104719}$ & $\numprint{1693}$ & \printMinuteSeconds{0}{4} & $\numprint{25}$ & $\numprint{6315}$ & \printMinuteSeconds{0}{16} & $\numprint{25}$ \\
            7 & $\numprint{13168}$ & $\numprint{104414}$ & $\numprint{1747}$ & \printMinuteSeconds{0}{5} & $\numprint{50}$ & $\numprint{6312}$ & \printMinuteSeconds{0}{21} & $\numprint{25}$ \\
            8 & \highlightCell $\numprint{12123}$ & \highlightCell $\numprint{100890}$ & \highlightCell $\numprint{1641}$ & \printMinuteSeconds{0}{5} & $\numprint{25}$ & $\numprint{6569}$ & \printMinuteSeconds{0}{22} & $\numprint{25}$ \\
            \noalign{\vskip 7pt}
            0 & $\numprint{21302}$ & $\numprint{256373}$ & $\numprint{2931}$ & -- & -- & \highlightCell $\numprint{2415}$ & -- & -- \\
            5 & $\numprint{8255}$ & $\numprint{45089}$ & $\numprint{929}$ & \highlightCell < \printMinuteSeconds{0}{1} & $\numprint{25}$ & $\numprint{2866}$ & \highlightCell \printMinuteSeconds{0}{3} & $\numprint{25}$ \\
            6 & $\numprint{7576}$ & $\numprint{39025}$ & $\numprint{837}$ & \highlightCell < \printMinuteSeconds{0}{1} & $\numprint{50}$ & $\numprint{2620}$ & \highlightCell \printMinuteSeconds{0}{3} & $\numprint{25}$ \\
            7 & $\numprint{7831}$ & $\numprint{35087}$ & $\numprint{833}$ & \printMinuteSeconds{0}{1} & $\numprint{25}$ & $\numprint{2493}$ & \highlightCell \printMinuteSeconds{0}{3} & $\numprint{25}$ \\
            8 & \highlightCell $\numprint{7212}$ & \highlightCell $\numprint{30040}$ & \highlightCell $\numprint{786}$ & \printMinuteSeconds{0}{1} & $\numprint{25}$ & $\numprint{2547}$ & \highlightCell \printMinuteSeconds{0}{3} & $\numprint{25}$ \\
            \bottomrule
        \end{tabular*}
    \end{table}
    \FloatBarrier

    \subsection{Comparison to FLASH-TB}
    \label{app:experiments:flash-tb}
    \Cref{tab:flashtbvstrex} compares the query performance and preprocessing time of T-REX to the FLASH-TB configuration with moderate space reported in \Cref{tab:uniformrandomquery:country}.
    This configuration uses $k=8$ cells and therefore requires one byte per transfer to store the flags, which matches the one byte required for the rank in T-REX without overlays.
    On the smaller networks of Switzerland and Paris, FLASH-TB is faster than T-REX by around a factor of~$2$ and has a smaller search space.
    This indicates that the directional pruning effect achieved via Arc-Flags is stronger than the hierarchical LCL pruning.
    On Germany, T-REX has a smaller search space and is faster by a factor of~$\numprint{1.8}$.
    This is likely due to the \textit{coning} effect~\cite{Bau09}, which is inherent to unidirectional Arc-Flags:
    Once the search reaches the target cell, all transfers are flagged.
    With~$k=8$ cells, this means that the search is not sped up at all in a significant fraction of the network.

    \begin{table}[H]
		\caption{
            Performance of T-REX and FLASH-TB with $k=8$ cells, averaged over $\numprint{10000}$ random queries.
        }
		\label{tab:flashtbvstrex}
		\centering
		\begin{tabular*}{\textwidth}{@{\,}l@{\extracolsep{\fill}}l@{\extracolsep{\fill}}r@{\extracolsep{\fill}}r@{\extracolsep{\fill}}r@{\extracolsep{\fill}}r@{\,}}
			\toprule
            \multirow{2}{*}{Network} & \multirow{2}{*}{Algorithm} & Scanned & Relaxed & Query & Preprocessing\\
            & & trips & transfers & $\left[\si{\mu\second}\right]$ & $\left[ \mathrm{hh}{:}\mathrm{mm}{:}\mathrm{ss}\right]$\\
			\midrule
			\multirow{3}{*}{Germany} & T-REX (overlays)       & $\numprint{27998}$ &  $\numprint{160570}$ & $\numprint{4442}$  & \printTime{0}{0}{53}  \\
            & T-REX (no overlays) & $\numprint{16386}$ & $\numprint{518617}$ & $\numprint{7822}$ & \printTime{0}{0}{53}\\
            & FLASH-TB (moderate) & $\numprint{32084}$ & $\numprint{395224}$ & $\numprint{7871}$ & \printTime{29}{39}{30} \\
            \noalign{\vskip 6pt}
			\multirow{3}{*}{Switzerland} & T-REX (overlays)        & $\numprint{6688}$ &  $\numprint{29671}$ & $\numprint{710}$ & \printTime{0}{0}{5}  \\
            & T-REX (no overlays) & $\numprint{4174}$ & $\numprint{61659}$ & $\numprint{921}$ & \printTime{0}{0}{5}\\
            & FLASH-TB (moderate) & $\numprint{3025}$ & $\numprint{22397}$ & $\numprint{454}$ & \printTime{0}{23}{53} \\
            \noalign{\vskip 6pt}
            \multirow{3}{*}{Paris} & T-REX (overlays)       & $\numprint{12123}$ &  $\numprint{100890}$ & $\numprint{1641}$ & \printTime{0}{0}{11} \\
            & T-REX (no overlays) &  $\numprint{6858}$ & $\numprint{188500}$ & $\numprint{1877}$ & \printTime{0}{0}{11} \\
            & FLASH-TB (moderate) & $\numprint{3253}$ & $\numprint{76122}$  & $\numprint{899}$ & \printTime{0}{29}{25} \\
			\bottomrule
		\end{tabular*}
	\end{table}
    \FloatBarrier

    \subsection{Geo-Rank Queries}
    \label{app:experiments:georank}
    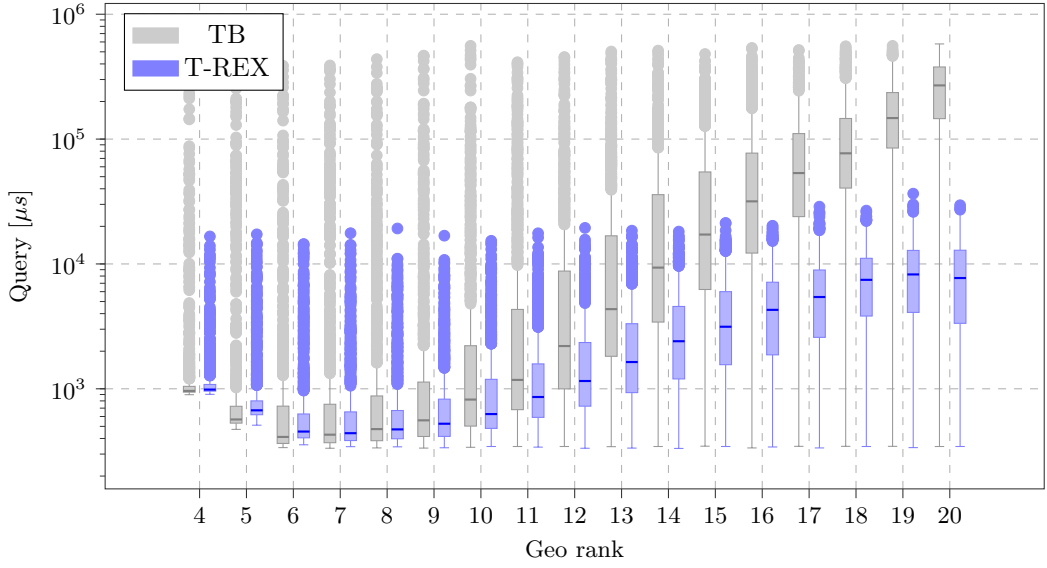
\begin{figure}[t]
		\centering
        \begin{tikzpicture}

\begin{axis}[
    width=\textwidth,
    height=8cm,
    ymode=log,
    grid=major,
    grid style={dashed,gray!60},
    xlabel={Geo rank},
    ylabel={Query [$\mu s$]},
    boxplot/box extend=0.25,
    boxplot/draw direction=y,
    xtick={1,...,17},
    xticklabels={4,5,6,7,8,9,10,11,12,13,14,15,16,17,18,19,20},
    legend style={
        at={(0.02,0.98)},
        anchor=north west,
        draw=black
    },
]

\foreach \col [count=\i from 1] in {4,5,6,7,8,9,10,11,12,13,14,15,16,17,18,19,20}{
\addplot+[
    boxplot,
    boxplot/draw position=\i-0.22,
    boxplot/every median/.style={thick,black!50},
    fill=black!20,
    draw=black!40,
    mark options={black!20},
    forget plot
]
table[y=\col, col sep=comma, header=true]
{plot_scripts/data/georank_europe/tb.europe.csv};
}

\foreach \col [count=\i from 1] in {4,5,6,7,8,9,10,11,12,13,14,15,16,17,18,19,20}{
\addplot+[
    boxplot,
    boxplot/draw position=\i+0.22,
    boxplot/every median/.style={thick,blue},
    fill=blue!30,
    draw=blue!50,
    mark options={blue!50},
    forget plot
]
table[y=\col, col sep=comma, header=true]
{plot_scripts/data/georank_europe/trex.14.0.5.csv};
}

\addlegendimage{area legend,fill=black!20,draw=black!20}
\addlegendentry{TB}

\addlegendimage{area legend,fill=blue!50,draw=blue!50}
\addlegendentry{T-REX}

\end{axis}

\end{tikzpicture}
		\caption{ 
        Query times of TB and T-REX (with overlays) for geo-rank queries on the Europe instance (with 14 levels and $50\,\%$ imbalance).
        For $\numprint{1000}$ random source stops and departure times, queries with rank~$i$ are to the $2^{i}$-th closest stop.
        Boxes denote the interquartile range, with the mean marked by a line.
        The whiskers extend to the $\numprint{1.5}\times$ value of the first and third quartile, respectively.
        Note that the y-axis is in logarithmic scale.
        }
		\label{fig:geo-ran:europe}
	\end{figure}

    To analyze how the query time is affected by the distance between the source and target stop, we evaluate \emph{geo-rank queries}.
    Here, $\numprint{1000}$ source stops and departure times are chosen uniformly at random, as before.
    For a given source stop~$\sourceStop$, we select different target stops of varying geographical distance to~$\sourceStop$.
    To do so, we sort all stops according to their geographical distance from~$\sourceStop$.
    For each \emph{geo-rank} $r \in \{0, \dots, \lfloor\log|\stops|\rfloor \}$, we choose the $2^{r}$-th closest stop to~$\sourceStop$ as the target stop.

	Figure~\ref{fig:geo-ran:europe} plots the performance on the Europe dataset.
	T-REX outperforms TB for all ranks except for the lowest rank~$r=4$, where both algorithms require just under $\numprint{1}\,\si{\milli\second}$.
    This is to be excepted, as these queries are likely to stay within the same lowest-level cells, and therefore T-REX cannot prune the search space.
    The speedup of T-REX over TB increases with the geo-rank.
	Even outlier queries are processed more efficiently by T-REX, requiring only $\numprint{55}$--$\numprint{60}\,\si{\milli\second}$; an order of magnitude faster than TB.
    \FloatBarrier

    \subsection{Profile Queries}
    \label{app:experiments:profile}
    Table~\ref{tab:uniformprofile} reports measurements for $\numprint{1000}$ random profile queries, with the departure time range chosen as the entire first day of the service period.
	The speedup of T-REX over TB is similar to that observed for fixed departure time queries.
	It is particularly high on Europe, with a factor of $\numprint{13}$ and an average runtime of just under $\numprint{100}\,\si{\milli\second}$.

	\begin{table}[t]
		\caption{
    		Performance of TB and T-REX, averaged over $\numprint{1000}$ random $24$-hour profile queries.
		}
		\label{tab:uniformprofile}
		\centering
		\begin{tabular*}{\textwidth}{@{\,}l@{\extracolsep{\fill}}l@{\extracolsep{\fill}}c@{\extracolsep{\fill}}r@{\extracolsep{\fill}}r@{\extracolsep{\fill}}r@{\extracolsep{\fill}}r@{\,}}
			\toprule
            \multirow{2}{*}{Network} & \multirow{2}{*}{Algorithm} & \multirow{2}{*}{Overlays} & Scanned & Relaxed & \multirow{2}{*}{Journeys} & Query \\
            & & & trips & transfers & & $\left[\si{\mu\second}\right]$ \\
			\midrule
			\multirow{2}{*}{Europe} & TB & $\circ$ & $\numprint{3544302}$ & $\numprint{69398487}$ & $\numprint{12.6}$ & $\numprint{1155254}$  \\
			& T-REX & $\circ$& $\numprint{151768}$ & $\numprint{7747088}$ & $\numprint{12.6}$ & $\numprint{139934}$  \\
            & T-REX & $\bullet$& $\numprint{196201}$ & $\numprint{1724965}$ & $\numprint{12.6}$ & $\numprint{56941}$ \\
            \noalign{\vskip 6pt}
			\multirow{2}{*}{Germany} & TB & $\circ$& $\numprint{1242125}$ & $\numprint{17510793}$ & $\numprint{10.8}$ & $\numprint{303435}$ \\
			& T-REX & $\circ$& $\numprint{95820}$ & $\numprint{3062589}$ & $\numprint{10.8}$ & $\numprint{58044}$ \\
            & T-REX & $\bullet$& $\numprint{110149}$ & $\numprint{637628}$ & $\numprint{10.8}$ & $\numprint{25035}$ \\
            \noalign{\vskip 6pt}
			\multirow{2}{*}{Great Britain} & TB & $\circ$& $\numprint{358507}$ & $\numprint{6579673}$ & $\numprint{11.0}$ & $\numprint{91124}$ \\
			& T-REX & $\circ$& $\numprint{32484}$ & $\numprint{1171340}$ & $\numprint{11.0}$ & $\numprint{20495}$ \\
            & T-REX & $\bullet$& $\numprint{43196}$ & $\numprint{261925}$ & $\numprint{11.0}$ & $\numprint{8277}$ \\
            \noalign{\vskip 6pt}
			\multirow{2}{*}{Switzerland} & TB & $\circ$& $\numprint{216323}$ & $\numprint{2085029}$ & $\numprint{24.2}$ & $\numprint{34851}$ \\
			& T-REX & $\circ$& $\numprint{36019}$ & $\numprint{506693}$  & $\numprint{24.2}$ & $\numprint{9955}$ \\
            & T-REX & $\bullet$& $\numprint{43082}$ & $\numprint{197340}$ & $\numprint{24.2}$ & $\numprint{6050}$ \\
            \noalign{\vskip 6pt}
			\multirow{2}{*}{Paris} & TB  & $\circ$ & $\numprint{181380}$ & $\numprint{4772478}$ & $\numprint{32.6}$ & $\numprint{37119}$\\
			& T-REX & $\circ$& $\numprint{95532}$ & $\numprint{3456960}$ & $\numprint{32.6}$ & $\numprint{34656}$ \\
            & T-REX & $\bullet$& $\numprint{119527}$ & $\numprint{1798219}$ & $\numprint{32.6}$ & $\numprint{20701}$ \\
            \noalign{\vskip 6pt}
			\multirow{2}{*}{Berlin} & TB & $\circ$& $\numprint{160394}$ & $\numprint{1848622}$ & $\numprint{29.4}$ & $\numprint{24985}$  \\
			& T-REX & $\circ$& $\numprint{45912}$ & $\numprint{876923}$ & $\numprint{29.4}$ & $\numprint{14750}$ \\
            & T-REX & $\bullet$& $\numprint{60938}$ & $\numprint{326928}$ & $\numprint{29.4}$ & $\numprint{8091}$ \\
			\bottomrule
		\end{tabular*}
	\end{table}

\end{document}